\documentclass[aps, pra, reprint, superscriptaddress]{revtex4-2}

\pdfoutput=1

\usepackage{header}

\graphicspath{{graphics/}}

\begin{document}

\bibliographystyle{apsrev4-2}

\title{A Unified Graph-Theoretic Framework for Free-Fermion Solvability}

\author{Adrian Chapman}
\email{adrian.chapman@materials.ox.ac.uk}
\affiliation{Department of Materials, University of Oxford, Parks Road, Oxford OX1 3PH, United Kingdom}

\author{Samuel J. Elman}
\email{samuel.elman@uts.edu.au}
\affiliation{Centre for Quantum Software and Information, School of Computer Science, Faculty of Engineering \& Information Technology, University of Technology Sydney, NSW 2007, Australia}
\affiliation{Centre for Engineered Quantum Systems, School of Physics, The University of Sydney, Sydney, Australia}

\author{Ryan L. Mann}
\email{mail@ryanmann.org}
\homepage{http://www.ryanmann.org}
\affiliation{Centre for Quantum Software and Information, School of Computer Science, Faculty of Engineering \& Information Technology, University of Technology Sydney, NSW 2007, Australia}
\affiliation{Centre for Quantum Computation and Communication Technology}
\affiliation{School of Mathematics, University of Bristol, Bristol, BS8 1UG, United Kingdom}

\begin{abstract}
    We show that a quantum spin system has an exact description by non-interacting fermions if its frustration graph is claw-free and contains a simplicial clique. The frustration graph of a spin model captures the pairwise anticommutation relations between Pauli terms of its Hamiltonian in a given basis. This result captures a vast family of known free-fermion solutions. In previous work, it was shown that a free-fermion solution exists if the frustration graph is either a line graph, or (even-hole, claw)-free. The former case generalizes the celebrated Jordan-Wigner transformation and includes the exact solution to the Kitaev honeycomb model. The latter case generalizes a non-local solution to the four-fermion model given by Fendley. Our characterization unifies these two approaches, extending generalized Jordan-Wigner solutions to the non-local setting and generalizing the four-fermion solution to models of arbitrary spatial dimension. Our key technical insight is the identification of a class of cycle symmetries for all models with claw-free frustration graphs. We prove that these symmetries commute, and this allows us to apply Fendley's solution method to each symmetric subspace independently. Finally, we give a physical description of the fermion modes in terms of operators generated by repeated commutation with the Hamiltonian. This connects our framework to the developing body of work on operator Krylov subspaces. Our results deepen the connection between many-body physics and the mathematical theory of claw-free graphs.
\end{abstract}

\maketitle

{
\hypersetup{linkcolor=black}
\tableofcontents
}

\section{Introduction}
\label{section:introduction}

The Jordan-Wigner transformation represents a fascinating insight into the physics of quantum many-body spin systems. It identifies collective spin degrees of freedom with those of fermions, resulting in a fermionic model with corresponding properties to the spin system of interest~\cite{jordan1928uber}. It is perhaps best-known for its application to models where the effective fermions are non-interacting, allowing for an exact solution to these otherwise non-trivial systems~\cite{lieb1961two}. Since its discovery, the Jordan-Wigner transformation has been generalized to an entire family of exact free-fermion solutions~\cite{fradkin1989jordan, wang1991ground, huerta1993bose, batista2001generalized, bravyi2002fermionic, verstraete2005mapping, nussinov2012arbitrary, chen2018exact, chen2019bosonization, backens2019jordan, tantivasadakarn2020jordan}, yielding new understanding for a wide class of spin models.

Free fermions have a rich connection to combinatorics and quantum information. Quantum circuits describing the time evolution of free-fermion systems under the Jordan-Wigner transform are the focus of fermionic linear optics, where they are also known as matchgate circuits. These circuits were initially proposed by Valiant as an instance of a \emph{holographic algorithm}~\cite{valiant2008holographic}, inspired by the Fisher-Kastelyn-Temperley algorithm~\cite{temperley1961dimer, kasteleyn1963dimer, kasteleyn1967graph} for counting weighted perfect matchings in a graph. They illustrate the deep connection between fermions and combinatorial structures. While matchgate circuits can be efficiently simulated classically in a fixed basis, simple changes to this setting make them classically intractable or even universal for quantum computation~\cite{knill2001fermionic, terhal2002classical, bravyi2006universal, brod2011extending, hebenstreit2019all}. They are thus a useful setting for understanding the transition from classical to quantum computational power. Furthermore, efficient classical algorithms are often reflected in the exact solvability of quantum models, as with Valiant's original proposal for matchgates.

The application of combinatorial tools for describing effective fermions has found renewed interest in quantum chemistry~\cite{babbush2016exponentially, mcardle2020quantum, wang2021resource}, where efficient fermion-to-qubit mappings are necessary for simulating interacting fermions on a quantum computer~\cite{bravyi2002fermionic, setia2019superfast, ball2005fermions, bravyi2017tapering, steudtner2018fermion, jiang2019majorana, havlivcek2017operator, jiang2020optimal, chiew2021optimal, derby2021acompact, derby2021bcompact, chen2022equivalence}. These mappings can, in some sense, be considered the reverse problem of finding a free-fermion solution to a spin model.

Concretely, the Jordan-Wigner transformation and its generalizations map many-qubit Pauli observables directly to fermionic operators: the \emph{Majorana modes}. These mappings are generator-to-generator, as they identify a direct correspondence between Hamiltonian terms in the spin system and terms in its dual fermion model. They are also \emph{generic} in that the solution method applies for all values of the Hamiltonian couplings.  In Ref.~\cite{chapman2020characterization}, a connection was shown between the solvability of a system by this method and its \emph{frustration graph}. This is the graph whose vertices correspond to terms in the spin Hamiltonian, written in a given Pauli basis, and are neighboring if the associated Pauli operators anticommute. It was shown in Ref.~\cite{chapman2020characterization} that a generator-to-generator free-fermion solution is possible if the frustration graph is a line graph. This property corresponds to the absence of certain \emph{forbidden induced subgraphs} of the Hamiltonian frustration graph: anticommutation structures among subsets of Hamiltonian terms that obstruct a free-fermion solution. Additionally, solutions captured by the line-graph characterization generally include a set of Hamiltonian symmetries associated to induced cycles---or holes---of the frustration graph.

\begin{figure}[ht!]
    \centering
    \begin{tikzpicture}[scale=1.8]
        \draw[thick, blue] (2,1.82) ellipse (2cm and 0.9cm);
        \node (a) at (2.95,2) {(b)};
        \node (a) at (3.00,1.8) {(even-hole, claw)-free};
        \node (a) at (2.95,1.6) {\emph{\small e.g. four-fermion}};
        \draw[thick, red] (1.02,2.4) ellipse (1.2cm and 1.25cm);
        \node (a) at (1.1,1.80) {$\begin{array}{c}\text{Line graphs of} \\ \text{even-cycle-free} \\ \text{graphs}\end{array}$};
        \node (a) at (1.01,3.4) {(a)};
        \node (a) at (1.01,3.2) {Line graphs};
        \node (a) at (.95,3.0) {\emph{\small e.g. \emph{XY} chain,}};
        \node (a) at (.95,2.8) {\emph{\small  Kitaev honeycomb}};
        \draw[thick, purple] (2,2.4) ellipse (2.2cm and 1.65cm);
        \node (a) at (2.95,3.4) {(c)};
        \node (a) at (2.95,3.2) {\textbf{Simplicial, claw-free}};
        \node (a) at (2.95,3.0) {\small \textbf{(this work)}};
        \draw[thick] (2,2.75) ellipse (2.3cm and 2.15cm);
        \node (a) at (2,4.4) {Free-fermion solvable models};
    \end{tikzpicture}
    \caption{Summary of this work in relation to earlier results. 
    (a) A generalized Jordan-Wigner solution exists if the frustration graph of the given spin model is a line graph \cite{chapman2020characterization}. (b) A solution of the type given in Ref.~\cite{fendley2019free} holds when the graph is (even-hole, claw)-free \cite{elman2021free}. Though these two graph classes intersect, neither class contains the other. (c) We unify these methods to show that a free-fermion structure exists for simplicial, claw-free frustration graphs. We expect that there are more free-fermion-solvable models beyond this characterization, such as models with non-generic solutions.}
    \label{figure:solvablitycharacterization}
\end{figure}
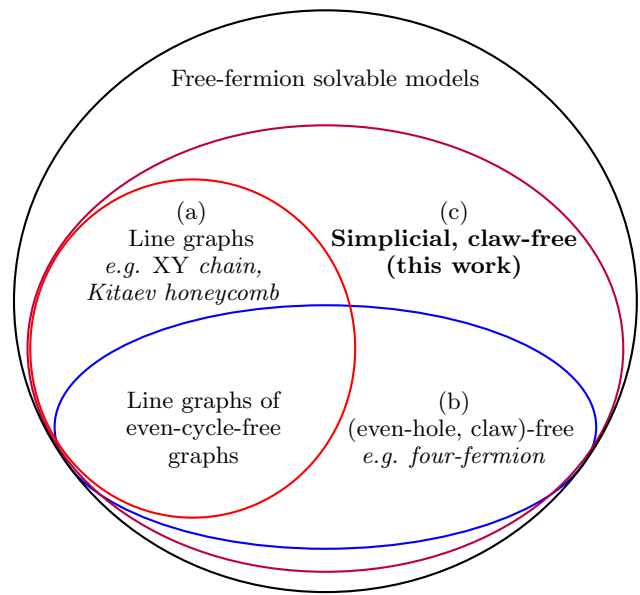

More recently, a free-fermion solvable model outside of the generalized Jordan-Wigner framework, called the four-fermion model, was given in a remarkable result by Fendley~\cite{fendley2019free}. Here, the fermions correspond to non-linear polynomials in the Pauli terms of the spin Hamiltonian, rather than individual terms. This solution holistically maps the spin Hamiltonian onto the free-fermion Hamiltonian and is generic, despite apparently transcending the generator-to-generator structure. Surprisingly, a solution of this form is also revealed by the absence of certain forbidden induced subgraphs of the Hamiltonian frustration graph~\cite{elman2021free}. These forbidden subgraphs include the claw ($K_{1, 3}$) as well as all holes of even length. As the claw is also a forbidden subgraph for line graphs, the family of line graphs and that of (even-hole, claw)-free graphs share some overlap, but also each include graphs not present in the other, as shown in Fig.~\ref{figure:solvablitycharacterization}. In a generalized Jordan-Wigner solution, even holes correspond to the aforementioned Pauli symmetries. This suggests that a set of generalized cycle symmetries exists to unify these two methods under one framework.

In this work, we give a graph-theoretic characterization of this framework. Our main result is summarized in Fig.~\ref{figure:solvablitycharacterization}. We show that if the frustration graph is claw-free and contains a structure called a simplicial clique, then it admits an exact free-fermion solution. The existence of this structure can be efficiently determined in claw-free graphs via the algorithm of Ref~\cite{chudnovsky2012growing}. We refer to this set of graphs as \emph{simplicial, claw-free} (SCF). Both graph classes of Refs.~\cite{chapman2020characterization, elman2021free} have this property \footnote{This was shown recently for (even-hole, claw)-free graphs in Ref.~\cite{chudnovsky2021note} in a stronger result motivated by the problem of Ref.~\cite{elman2021free}}. It is an interesting feature of our characterization that free-fermion solutions are generalized much in the same way as the graphs that describe them. Importantly, our result removes the even-hole-free assumption of Ref.~\cite{elman2021free}, and so extends the non-local solution method given by Fendley to systems of arbitrary spatial dimension. We are able to relax this assumption precisely by identifying a class of cycle-like symmetries, which generalize the cycle symmetries of Jordan-Wigner-type solutions. This identification can be seen as our main technical insight.

Our paper is structured as follows. In the remainder of the introduction, we summarize our main results and apply them to a small example application. In Section~\ref{section:FrustrationGraphs}, we give some background on frustration graphs and standardize our graph-theoretic notation. In Section~\ref{section:FreeFermionModels}, we review free-fermion models and motivate the use of graph theory to find free-fermion solutions. In Section~\ref{section:ClawFreeGraphs}, we refine our focus to the discussion of more technical topics surrounding claw-free graphs. We prove our main results in the following sections. We prove Theorem~\ref{theorem:conservedcharges} in Section~\ref{section:ConservedQuantities}. We show how this extends previous proof techniques to allow us to prove Theorem~\ref{theorem:freefermionsolution} in Section~\ref{section:ExactSolutions}. In Section~\ref{section:KrylovSubspaces}, we prove Theorem~\ref{theorem:polynomialdivisibility} by using a complementary set of tools to give an operational picture for the effective fermion modes in terms of operator Krylov subspaces. Finally, we give a numerical example of an explicit two-dimensional model whose free-fermion solution lies outside the generator-to-generator formalism in Section~\ref{section:ApplicationResults}. We conclude with a discussion of open questions in Section~\ref{section:Discussion}.

\subsection{Summary of Results}
\label{subsection:SummaryResults}

We consider many-body spin systems on $n$ qubits with Hamiltonians written in the Pauli basis
\begin{equation}
    \label{equation:paulihamiltonian}
    H = \sum_{\bsj\in V} b_{\bsj}\sigma^{\bsj} = \sum_{\bsj \in V} h_{\bsj}, 
\end{equation}
where $V\subseteq\{I,x,y,z\}^{\times n}$ is a set of strings labeling the $n$-qubit Pauli operators in the natural way, and $h_{\bsj}=b_{\bsj}\sigma^{\bsj}$ with $b_{\bsj}\in\mathbb{R}{\setminus}\{0\}$.

The \emph{frustration graph} $G=(V,E)$ of $H$ is the graph with vertices given by the non-zero Pauli terms in $H$, neighboring if the corresponding Pauli terms anticommute. Our main result extends the class of free-fermion solvable spin Hamiltonians $H$ based on their frustration graphs $G$.

\begin{result}[Theorem~\ref{theorem:conservedcharges} and Theorem~\ref{theorem:freefermionsolution}]
    \label{result:freefermionsolution}
    Let $H$ be a Hamiltonian whose frustration graph $G$ is connected, claw-free, and contains a simplicial clique. There exist commuting symmetries $\{\jkg{C_0}{G}\}_{\avg{C_0}}$, defined in terms of even holes in $G$, such that each symmetric subspace, labeled by $\mcj$, with projector $\Pi_{\mcj}$, admits a free-fermion solution,
    \begin{align}
        H = \sum_{\mcj}\left(\sum_{j=1}^{\alpha(G)}\varepsilon_{\mcj, j}[\psi_{\mcj,j},\psi_{\mcj, j}^{\dagger}]\right)\Pi_{\mcj}, \label{equation:firstresult}
    \end{align}
    where $\alpha(G)$ denotes the independence number of $G$. The fermionic ladder operators $\{\psi_{\mcj,j}\}_{\mcj,j}$ are constructed from another set of symmetries $\{\qkg{k}{G}\}_{k=0}^{\alpha(G)}$, defined in terms of independent sets of $G$. The symmetries $\{\qkg{k}{G}\}_{k=0}^{\alpha(G)}$ commute with each other and with the symmetries $\{\jkg{C_0}{G}\}_{\avg{C_0}}$. The single-particle energies $\{\varepsilon_{\mcj,j}\}_{\mcj,j}$ can be calculated from the roots of a generalized characteristic polynomial $Z_{G,\mcj}(-u^2)$ over each symmetric subspace specified by the projector $\Pi_{\mcj}$.
\end{result}

When the frustration graph $G$ is not connected, we have an independent solution for each connected component of $G$. The precise definitions of the fermionic ladder operators, single-particle energies, and the generalized characteristic polynomial are given in Section~\ref{section:FreeFermionModels}. Theorem~\ref{theorem:conservedcharges} shows that the symmetry operators $\{\qkg{k}{G}\}_{k=0}^{\alpha(G)}$ and $\{\jkg{C_0}{G}\}_{\avg{C_0}}$ are commuting. We use this theorem to apply the solution method of Refs.~\cite{fendley2019free, elman2021free} independently to each symmetric subspace, thus proving Eq.~(\ref{equation:firstresult}) as Theorem~\ref{theorem:freefermionsolution}. When there are no even holes, there is only a single such subspace, and we recover the result proven in Ref.~\cite{elman2021free}.

While our result gives an exact, explicit means to solve the Hamiltonian in Eq.~(\ref{equation:paulihamiltonian}), we would also like to extract a physical picture for the fermion modes from the solution. We address this with our second main result.

\begin{result}[Theorem~\ref{theorem:polynomialdivisibility} and Corollary~\ref{corollary:modifiedanticommutationrelation}]
    \label{result:inducedpath}
    Given a Hamiltonian $H$ with a connected, simplicial, claw-free frustration graph $G$. Let $\chi=\sigma^{\js}$ be a Pauli operator such that $\js$ is not in $V$, and $\chi$ anticommutes only with operators corresponding to the vertices in a simplicial clique of $G$. The operator $\chi$ commutes with each of the generalized cycle symmetries $\{\jkg{C_0}{G}\}_{\avg{C_0}}$. For each symmetric space $\mcj$, as defined in Result~\ref{result:freefermionsolution}, there exists a real matrix $\mathbf{A}_{G,\mcj}$, whose elements are indexed by induced paths in a graph $\gs$, such that 
    \begin{equation}
        \Pi_{\mcj} \adh{\chi}{k} = (-2i)^k\sum_{P} \left(\mathbf{A}_{G,\mcj}^k\right)_{\{\js\},P}h_P,
        \label{equation:krylovmatrix}
    \end{equation}
    where $\adh{\chi}{}=[iH,\chi]$. The matrix $\mathbf{A}_{G,\mcj}$ is the weighted adjacency matrix of a directed bipartite graph, with weights specified by the subspace $\mcj$. The operators $\{\adh{\chi}{k}\}_k$ satisfy
    \begin{align}
        \Pi_{\mcj} \{\adh{\chi}{j}, \adh{\chi}{k}\}
        &= 2(-2i)^{j+k}\left(\mathbf{A}_{G, \mcj}^{j+k}\right)_{\{\js\},\{\js\}}\Pi_{\mcj} \notag \\
        &= 2\left(\mathbf{M}_{G,\mcj}\right)_{jk}\Pi_{\mcj},
        \label{equation:modifiedcommutationrelation}
    \end{align}
    where the real matrix $\mathbf{M}_{G,\mcj}$ is positive definite for all $\mcj$.
\end{result}

Result~\ref{result:inducedpath} essentially says that we can define effective Majorana fermion modes by acting on $\chi$ by repeated commutation with $H$. Over each symmetric subspace $\mcj$, the matrix $\mathbf{A}_{G,\mcj}$ is a weighted unpacking of the frustration graph $G$, and we utilize the fact that the operators $\{\adh{\chi}{k}\}_k$ become linearly dependent for $k$ larger than a certain rank to refold them into effective Majorana modes by diagonalizing the matrix $\mathbf{M}_{G,\mcj}$. In this picture, the matrix $\mathbf{A}_{G,\mcj}$ acts as an effective single-particle Hamiltonian for these Majorana fermions. Before we prove these results, we first consider an example application.

\subsection{Example Application}
\label{subsection:ExampleApplication}

We consider a model defined on four qubits by the following Hamiltonian
\begin{equation}
    \begin{split}
        H &= \sigma_1^x+\sigma_1^z+\sigma_1^x\sigma_2^x+\sigma_1^z\sigma_2^x\sigma_3^x+\sigma_1^y\sigma_2^z+\sigma_1^z\sigma_3^z \\
        &\quad+ \sigma_1^z\sigma_2^x\sigma_3^y\sigma_4^x+\sigma_1^y\sigma_2^y\sigma_3^y\sigma_4^z,   
    \end{split}
    \label{equation:exampleapplicationhamiltonian}
\end{equation}
where, for succinctness, we have set $b_{\bsj}=1$ for all ${\bsj \in V}$. We denote a given Hamiltonian term by $h_j$, for $j\in\{1,2,\dots,8\}$, based on the order in which the term appears as written in Eq.~(\ref{equation:exampleapplicationhamiltonian}).

\begin{figure}[ht!]
    \centering
    \includegraphics{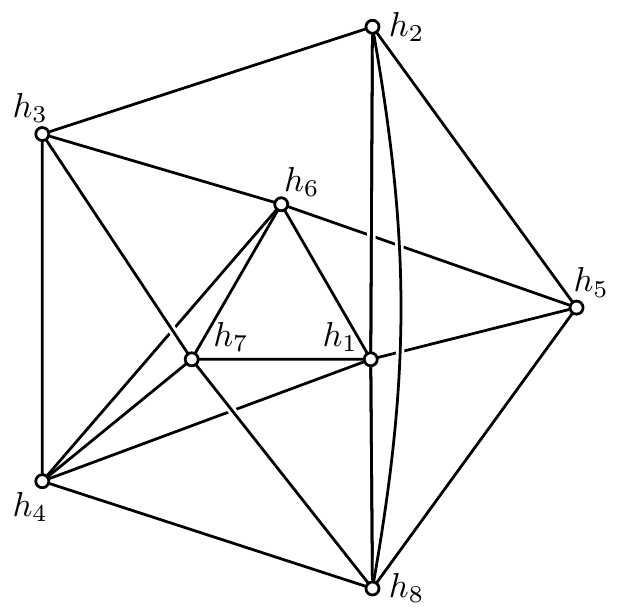}
    \caption{The frustration graph of the model given in Eq.~(\ref{equation:exampleapplicationhamiltonian}). The graph is not a line graph, but it is claw-free and contains a simplicial clique. The model therefore admits a free-fermion solution.}
    \label{figure:exampleapplicationfrustrationgraph}
\end{figure}

The frustration graph of this model, shown in Fig.~\ref{figure:exampleapplicationfrustrationgraph}, is claw-free and contains the simplicial cliques 
\begin{align}
    K_{s, 1} &= \{1,2,5,8\}, \notag \\
    K_{s, 2} &= \{2,3\}, \notag \\
    K_{s, 3} &= \{3,4,6,7\}. \notag
\end{align}
We define the following products of operators over the even holes in the frustration graph
\begin{align}
    h_{C_1} &= h_1h_3h_2h_4 = \sigma_3^x, \notag \\
    h_{C_2} &= h_1h_3h_2h_6 = \sigma_2^x\sigma_3^z, \notag \\
    h_{C_3} &= h_1h_3h_2h_7 = \sigma_3^y\sigma_4^x, \notag \\
    h_{C_4} &= h_8h_3h_2h_4 = -\sigma_1^z\sigma_2^y\sigma_3^z\sigma_4^z, \notag \\
    h_{C_5} &= h_5h_3h_2h_6 = \sigma_1^z\sigma_2^y\sigma_3^z, \notag \\
    h_{C_6} &= h_8h_3h_2h_7 = \sigma_1^z\sigma_2^y\sigma_4^y, \notag \\
    h_{C_7} &= h_8h_6h_4h_5 = \sigma_4^z, \notag \\
    h_{C_8} &= h_8h_6h_7h_5 = -\sigma_3^z\sigma_4^y, \notag
\end{align}
where the operator ordering for each even hole is such that operators are grouped into the coloring classes of the hole. We collect these operators into sums to give the generalized cycle symmetries
\begin{align}
    J_1 = \jkg{C_1}{G} &= \sum_{k=1}^6h_{C_k}, \notag \\
    J_2 = \jkg{C_7}{G} &= \sum_{k=7}^8h_{C_k}. \notag
\end{align}

It can be straightforwardly verified that the operators $J_1$ and $J_2$ commute with the Hamiltonian, even though the individual terms in $\{h_{C_k}\}_{k=1}^8$ do not. Notably, the cycle symmetry $J_1$ does not square to an operator proportional to the identity. Thus, it is not proportional to a Pauli operator in any basis, in contrast to the setting where the frustration graph of the model is a line graph. Rather, we have
\begin{align}
    J_1^2 = 6I-2J_2, &\quad \mcj_1 = \pm\sqrt{6-2\mcj_2}, \notag \\
    J_2^2 = 2I, &\quad \mcj_2 = \pm\sqrt{2}, \notag
\end{align}
and $J_1$ and $J_2$ commute, as we expect from Result~\ref{result:freefermionsolution}. We restrict to their mutual eigenspace by choosing a sign configuration in the expressions for $\mcj_1$ and $\mcj_2$. We now apply Result~\ref{result:inducedpath}. First, we choose
\begin{equation}
    \chi = \sigma_1^y\sigma_2^y\sigma_3^x\sigma_4^z, \notag
\end{equation}
as the Pauli operator corresponding to vertex $\js$. This operator anticommutes only with the Hamiltonian terms in the simplicial clique $K_{s,1}$. Next, we consider products of Hamiltonian terms along induced paths of the frustration graph with one endpoint at $\js$, and we let $K^{(\ell)}$ be the sum of all such products over paths with length $\ell$, taking the convention that $K^{(0)}=\chi$. Acting on $\chi$ with repeated commutators by the Hamiltonian gives 
\begin{align}
    \adh{\chi}{0} &= K^{(0)}, \notag \\
    \adh{\chi}{1} &= -2iK^{(1)}, \notag \\
    \adh{\chi}{2} &= (-2i)^2\left(4K^{(0)}+K^{(2)}\right), \notag \\
    \adh{\chi}{3} &= (-2i)^3\Big[\left(4K^{(1)}+K^{(3)}\right) \notag \\
    &\quad+ \chi\left(3h_1+h_2+h_5+2h_8\right)\Big], \notag \\
    \adh{\chi}{4} &= (-2i)^4\left\{\left[23-2(J_1+J_2)\right]K^{(0)}+8K^{(2)}\right\}. \notag
\end{align}
The fourth nested commutator is a linear combination of the previous ones over each mutual eigenspace of generalized cycle symmetries in the set $\{J_1,J_2\}$,
\begin{equation}
    \adh{\chi}{4} = -16\left(9+2J_1+2J_2\right)\left(\adh{\chi}{0}\right)-32\left(\adh{\chi}{2}\right). \notag
\end{equation}
As a linear map on the cyclic subspace generated by $\chi$, $\operatorname{ad}_{iH}$ thus satisfies the following characteristic polynomial over the subspace labeled by $\mcj=(\mcj_1,\mcj_2)$,
\begin{equation}
    f_{\mcj}(u) = u^4+32u^2+16(9+2\mcj_1+2\mcj_2). \notag
\end{equation}
This coincides with the generalized characteristic polynomial (Lemma~\ref{lemma:generalizedcharacteristicpolynomial}),
\begin{equation}
    Z_{G,\mcj}(x) = 1+8x+9x^2+2(\mcj_1+\mcj_2)x^2, \notag
\end{equation}
as
\begin{align}
    f_{\mcj}(u) = u^4Z_{G,\mcj}\left[(2/u)^2\right], \notag
\end{align}
so the polynomial $Z_{G,\mcj}(x)$ is equivalent to $f_{\mcj}$ up to a polynomial transformation. As will be defined, the single-particle energies are the reciprocals of the positive roots of $Z_{G,\mcj}(-u^2)$. These are the positive roots of $f_{\mcj}(2iu)$ and are related to the eigenvalues of $\operatorname{ad}_{iH}$ as a linear map on the aforementioned vector space of operators. We have
\begin{equation}
    \varepsilon_{\mcj,\pm} = \sqrt{4\pm\sqrt{7-2(\mcj_1+\mcj_2)}}. \notag
\end{equation}
This gives the spectrum of the model as
\begin{equation}
    \mce_{\mcj,\mathbf{x}} = (-1)^{x_+}\varepsilon_{\mcj,+}+(-1)^{x_-}\varepsilon_{\mcj,-}, \notag
\end{equation}
for $\mathbf{x}=(x_+,x_-)\in\{0,1\}^{\times 2}$. An explicit description of the physical Majorana modes can be found by taking the linear combinations of $\{\adh{\chi}{k}\}_{k=0}^3$ which diagonalize the matrix $\mathbf{M}_{G,\mcj}$, as defined in Eq.~(\ref{equation:modifiedcommutationrelation}), and normalizing. We see that any such linear combinations are preserved under commutation with $H$ by their definition.

Formally, we have found the Krylov subspace of $\operatorname{ad}_{iH}$ generated by $\chi$ on a particular symmetric subspace, so our treatment is completely general in that sense. Operationally, free fermions only enter into the description through Eq.~(\ref{equation:modifiedcommutationrelation}), which can be viewed as a generalized canonical anticommutation relation. We see that, even without this relation, the presence of symmetries is necessary for restricting the rank of the Krylov subspace, so Krylov subspace methods may provide a route to applying graph theory to more general models. We now describe the background of our framework in detail and prove our main results. 

\section{Frustration Graphs}
\label{section:FrustrationGraphs}

In this section we explain frustration graphs and standardize our graph-theoretic notation. A graph $G=(V,E)$ is a set $V$ of vertices, together with a set $E \subset V^{\times 2}$ of edges. Two vertices $\bsj,\bsk \in V$ are said to be \emph{neighboring} if there is an edge $\{\bsj,\bsk\} \in E$. Two edges $\{\bsj,\bsk\},\{\bsu,\bsv\} \in E$ are said to be \emph{incident} if they share a vertex, i.e., $\abs{\{\bsj,\bsk\}\cap\{\bsu, \bsv\}}=1$. A vertex $\bsj \in V$ and edge $\{\bsu,\bsv\} \in E$ are similarly incident if $\bsj \in \{\bsu,\bsv\}$. The \emph{order} of a graph is the cardinality $\abs{V}$ of its vertex set. The \emph{size} of a graph is the cardinality $\abs{E}$ of its edge set.

Since Pauli operators either commute or anticommute, it is convenient to describe the relations between terms in a spin Hamiltonian by a graph.
\begin{definition}[Frustration graph]
    \label{definition:frustrationgraph}
    The frustration graph of the Hamiltonian $H$ in Eq.~(\ref{equation:paulihamiltonian}) is the graph $G=(V,E)$ with 
    \begin{equation}
        E = \{\{\bsj,\bsk\} \mid h_{\bsj}h_{\bsk} = -h_{\bsk}h_{\bsj}\}.
    \end{equation}
\end{definition}

The frustration graph is always simple. There are no self loops because every Hamiltonian term commutes with itself, and there is at most one edge between each pair of terms. As stated previously, terms in $H$ from distinct connected components of $G$ will commute, so we have an independent solution for each such component. For this work, we additionally assume that all models have finitely many terms, so the frustration graphs we consider are finite. Without loss of generality, we assume that distinct vertices in the frustration graph correspond to distinct Pauli terms in $H$. (We can always collect repeated Pauli terms by adding their coefficients $b_{\bsj}$.) We shall often refer to terms in the Hamiltonian interchangeably with their vertices in the frustration graph. The commutation relation between Pauli terms is clearly unchanged by including the coefficients in their definitions, so we prefer to give statements in terms of the $h_{\bsj}$ rather than the $\sigma^{\bsj}$, with the understanding that $h_{\bsj}^2=b_{\bsj}^2I$. The frustration graph thus naturally captures properties of the spin model that do not depend on the coefficients. We refer to such properties as \emph{generic}.

We next consider subsets of Hamiltonian terms and the associated induced subgraphs of the frustration graph. To this end, we define our labeling scheme in Eq.~(\ref{equation:paulihamiltonian}) more precisely. We can equivalently describe a Pauli string $\bsj\in\{I,x,y,z\}^{\times n}$ by a binary string on $2n$ bits by associating each single-qubit Pauli label to a 2-bit string as $I\mapsto(0,0)$, $x\mapsto(1,0)$, $y\mapsto(1,1)$, and $z\mapsto(0,1)$. Let $\bsj=(\bsj_x,\bsj_z)\in\{0,1\}^{\times 2n}$ be the binary vector such that the $k$th component $j_{x,k}$ of $\bsj_{x}$ is the first bit of the $k$th qubit label according to this association, and, similarly, $j_{z,k}$ is the second bit of the $k$th qubit label. This gives
\begin{equation}
    \sigma^{\bsj} = i^{\bsj_x\cdot\bsj_z} \left[\bigotimes_{k=1}^n \left(\sigma_k^x\right)^{j_{x,k}}\right]\left[\bigotimes_{k=1}^n \left(\sigma_k^z\right)^{j_{z,k}}\right],
\end{equation}
where $\bsj_x\cdot\bsj_z = \sum_{k=1}^nj_{x,k}j_{z,k}$ denotes the Euclidean inner product.

The scalar commutator between Pauli terms is defined implicitly via
\begin{equation}
    \sigma^{\bsj} \sigma^{\bsk} = \scomm{\sigma^{\bsj}}{\sigma^{\bsk}} \sigma^{\bsk}\sigma^{\bsj}.
\end{equation}
Since Pauli operators either commute or anticommute, we have $\scomm{\sigma^{\bsj}}{\sigma^{\bsk}}=\pm1$. The sign factor is given by
\begin{equation}
    \scomm{\sigma^{\bsj}}{\sigma^{\bsk}} = (-1)^{\langle\bsj,\bsk\rangle},
\end{equation}
where $\langle\bsj,\bsk\rangle=\sum_{m=1}^n(j_{x, m}k_{z, m}+j_{z,m}k_{x,m}) \pmod 2$ is the binary symplectic inner product. The scalar commutator distributes over multiplication as
\begin{equation}
    \scomm{A}{BC} = \scomm{A}{B}\scomm{A}{C}, \label{equation:commutatordistribution}
\end{equation}
and, accordingly, the binary symplectic inner product is linear in $\bsj$ and $\bsk$. The scalar commutator and binary symplectic inner product are similarly well defined for the operators $\{h_{\bsj}\}_{\bsj \in V}$ as well.

For a subset $U \subseteq V$, the \emph{induced subgraph} $G[U]$ is the subgraph of $G$ whose vertex set is $U$ and whose edge set $E[U] = E \cap U^{\times 2}$ consists of all edges in $G$ which have both endpoints in $U$. We shall often refer to the vertex subset interchangeably with the subgraph it induces. Similarly, we shall use set-theoretic notation to denote the exclusion of vertices, e.g., $G{\setminus}U = G[V{\setminus}U]$.

An important family of vertex subsets is given by the neighborhoods of vertices in the graph.
\begin{definition}[Open and closed neighborhood]
    \label{definition:neighborhood}
    The \emph{open neighborhood} of $\bsj \in V$ is the set given by
    \begin{equation}
        \Gamma(\bsj) = \{\bsk \mid \{\bsj,\bsk\} \in E\},
    \end{equation}
    and the \emph{closed neighborhood} of $\bsj$ is the set given by 
    \begin{equation}
        \Gamma[\bsj] = \Gamma(\bsj)\cup\{\bsj\}.
    \end{equation}
    The \emph{degree} $\Delta(\bsj)=\abs{\Gamma(\bsj)}$ of $\bsj$ is the order of its open neighborhood.
\end{definition}
We often refer to the open or closed neighborhood of a vertex $\bsj$ in a subset of vertices $U \subseteq V$ by $\Gamma_U(\bsj) = \Gamma(\bsj) \cap U$. Similarly, $\Gamma_U[\bsj]=\Gamma_U(\bsj)\cup\{\bsj\}$. Note that we do not necessarily assume $\bsj \in U$ for this definition. Accordingly, we refer to the degree in a subset by $\Delta_U(\bsj)=\abs{\Gamma_U(\bsj)}$. We shall also refer to the closed neighborhood of a subset $U \subseteq V$ by $\Gamma[U] = \bigcup_{\bsj \in U}\Gamma[\bsj]$, and similarly for the open neighborhood when there is no ambiguity.

One useful feature of the binary linear structure of commutation relations between Hamiltonian terms is that it allows us to talk about commutation relations between products over subsets of terms. We can thus extract such commutation relations from the frustration graph. In general, we let $h_U=\prod_{\bsj \in U}h_{\bsj}$ denote the product of all operators whose vertices in $G$ are members of a particular subset $U \subseteq V$. Since reordering the operators in this product contributes an overall sign factor to $h_U$, the operator ordering is irrelevant to the commutation relations involving $h_U$ and other products of Hamiltonian terms. We define the operator ordering for specific families of vertex subsets on a case-by-case basis. With these definitions, we have
\begin{equation}
    \scomm{h_{\bsj}}{h_U} = (-1)^{\Delta_U(\bsj)}, \label{equation:degreephase}
\end{equation}
for any $\bsj \in V$. That is, $\langle\bsj,\sum_{\bsk \in U}\bsk\rangle = \Delta_U(\bsj) \pmod 2$ as we expect. We denote the symmetric difference between vertex subsets $U,W \subseteq V$ by $U{\oplus}W=(U{\setminus}W) \cup (W{\setminus}U)$. Applying the constraint that $h_{\bsj}^2 \propto I$ gives
\begin{equation}
    \scomm{h_{\bsj}}{h_{U{\oplus}W}} = \scomm{h_{\bsj}}{h_{U}} \scomm{h_{\bsj}}{h_{W}}, \label{equation:symmetricdifference}
\end{equation}
so that the commutation relation between $h_{\bsj}$ and $h_U$ is only changed by taking the symmetric difference with $W$ if $h_{\bsj}$ and $h_W$ anticommute.

We proceed to define the relevant graph-theoretic structures for this work. An \emph{independent set} $S$ in $G$ is a subset of vertices with no edges between them. For the corresponding operator $h_S$ the ordering of the factors in the product is irrelevant, as these factors commute with one another. The independence number $\alpha(G)$ is the order of the largest independent set in $G$. We denote $\mcs_{G}$ as the collection of all independent sets in $G$ and let $\mcs^{(k)}_{G}$ denote the collection of all independent sets of order $k$ from $G$. A \emph{matching} $M$ in $G$ is a subset of edges such that no two edges in $M$ are incident. A \emph{perfect matching} is a matching such that every vertex in $V$ is incident to exactly one edge in the matching. Clearly, a graph can only have a perfect matching if its order is even. We denote $\mcm_G$ as the collection of all matchings in $G$ and let $\mcm^{(k)}_{G}$ denote the collection of all matchings of $k$ edges from $G$.

The \emph{claw} is the graph consisting of a central vertex neighboring to every vertex in an independent set of order three (see Table~\ref{table:simplicialclaw}). That is, it is the complete bipartite graph $K_{1, 3}$. The vertices in the three-vertex independent set are called the \emph{leaves} of the claw. A graph is \emph{claw-free} if it does not include the claw as an induced subgraph. When we list a subset of vertices that induces a claw, we shall generally order the list by the central vertex followed by the leaves.

\begin{table}[ht!]
    \centering
    \begin{tabular}{cc}
        \hline\hline
        \emph{Forbidden} & \emph{Includes} \\
        Claw $K_{1,3}$ & Simplicial Clique \\
        \hline
        \includegraphics[width=0.215\textwidth]{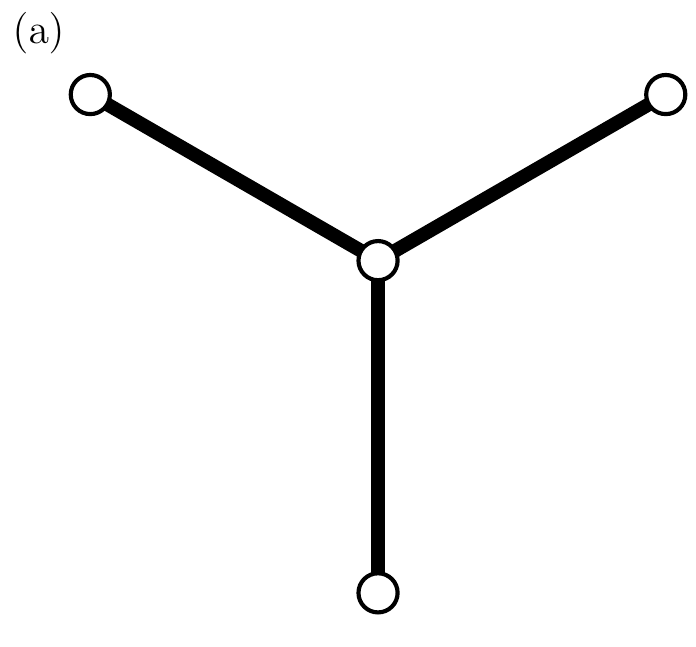} & \includegraphics[width=0.225\textwidth]{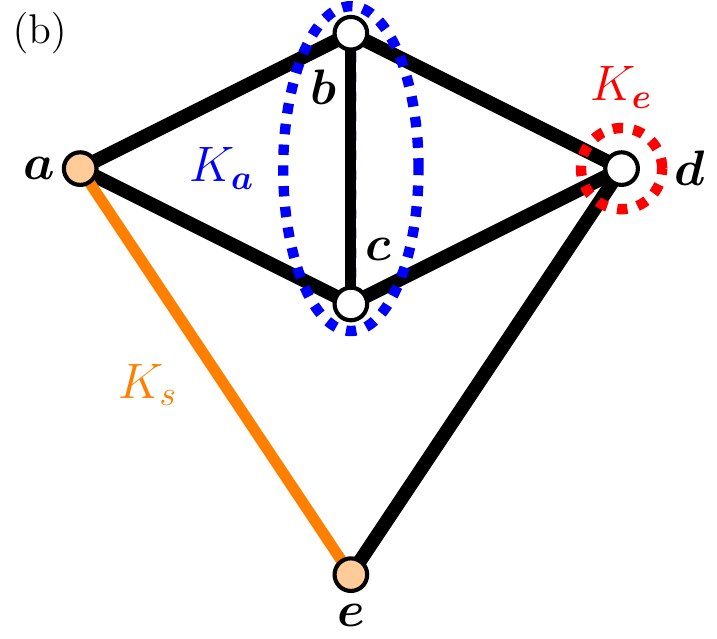}
        \\
        \hline\hline
    \end{tabular}
    \caption{(a) A graph is claw-free if no subset of its vertices induces the claw $K_{1,3}$. (b) A simplicial clique $K_s$ (orange) is a clique such that the neighborhood of each vertex in $K_s$ induces a clique in the graph $G{\setminus}K_s$. The depicted simplicial clique consists of the vertices $\{\bsa,\bse\}$. The neighborhood of the vertex $\bsa$ in $G{\setminus}K_s$ is the clique $K_{\bsa}=\{\bsb,\bsc\}$. Similarly, the neighborhood of $\bse$ in $G{\setminus}K_s$ is the single vertex $K_{\bse}=\bsd$. Crucially, our result captures frustration graphs containing even holes.}
    \label{table:simplicialclaw}
\end{table}

A \emph{path} is a set of distinct vertices $P=\{\bsj_i\}_{i=0}^{\ell}$ with $\bsj_i$ neighboring to $\bsj_{i+1}$ for $i \in \{0,\dots,\ell-1\}$. The vertices $\{\bsj_0,\bsj_{\ell}\}$ are the \emph{endpoints} of the path, and the quantity $\ell$ is the \emph{length} of the path. When $P \subseteq V$ is a subset of vertices such that $E[P]=\{\{\bsj_i,\bsj_{i+1}\}\}_{i=0}^{\ell-1}$, we say $P$ is an \emph{induced path}. That is, there are no edges in $E[P]$ other than those between vertices with consecutive indices in the path, of which there are $\ell$. We refer to the index $i$ as the \emph{distance} from $\bsj_i$ to $\bsj_0$ along $P$ in this case, and we give the vertex labeling for $P$ as
\begin{equation}
    P = \bsj_0\md\bsj_1\md\dots\md\bsj_{\ell}.
    \label{equation:pathlabel}
\end{equation}
For the corresponding operator $h_P$, we similarly order the factors from left to right by their distance from an endpoint in accordance with the labeling. We denote the set of all induced paths in $G$ by $\mcp_{G}$ and the set of all induced paths in $G$ of length $\ell$ by $\mcp_{G}^{(\ell)}$. We also define $P_k$ as the subpath of $P$ induced by the vertices up to $\bsj_k$, i.e., $P_k=\{\bsj_i\}^k_{i=0}$.

A \emph{cycle} is a set of distinct vertices $C=\{\bsj_i\}_{i=0}^{\ell-1}$ with $\bsj_i$ neighboring to $\bsj_{ i+1}$ for $i\in\{0,\dots\ell-1\}$ with index addition taken modulo $\ell$. The quantity $\ell$ is the length of the cycle. When $C \subseteq V$ is a subset of vertices such that $E[C]=\{\{\bsj_i,\bsj_{i+1\pmod\ell}\}\}_{i=0}^{\ell-1}$ for $\ell \geq 4$, we say $C$ is an \emph{induced cycle} or \emph{hole}. In this case, $\ell=\abs{C}$ is the number of vertices and edges of the hole.

An \emph{even hole} is a hole with an even number of vertices and edges. If $C$ is an even hole with length $\ell=2k$, then there are two unique independent sets of size $k>1$ in $C$, which are the coloring classes of $C$. Let these coloring classes be $C_a$ and $C_b$. We give the vertices of $C$ in each coloring class distinct labels, as
\begin{equation}
    C = \bsb_0\md\bsa_0\md\bsb_1\md\bsa_1\md\dots\md\bsb_{k-1}\md\bsa_{k-1}\md\bsb_{0}, \label{equation:evenholelabel}
\end{equation}
where $C_a=\{\bsa_j\}_{j=0}^{k-1}$ and $C_b=\{\bsb_j\}_{j=0}^{k-1}$. Here, we take addition in the indices of these vertices modulo $k$. We choose the factor ordering for the operator corresponding to $C$ as $h_C=h_{C_a}h_{C_b}$, where the ordering within each of the coloring-class factors is again irrelevant as these are independent sets. Furthermore, we are free to exchange the ordering of $h_{C_a}$ and $h_{C_b}$, as these operators commute (we occasionally group factors in $h_P$ this way as well). We denote the set of all even holes in $G$ by $\mcc^{(\text{even})}_{G}$ and the set of all even holes in $G$ of length $\ell$ by $\mcc_{G}^{(\ell)}$. We say two even holes $C,C' \in \mcc^{\text{(even)}}_{G}$ are compatible if and only if $G[C \cup C']$ is a disconnected graph whose components are $G[C]$ and $G[C']$. We let $\mathscr{C}^{\text{(even)}}_{G}$ denote the collection of all subsets of $\mcc^{\text{(even)}}_{G}$ that are pairwise compatible. For a given such subset $\mcx \in \mathscr{C}^{\text{(even)}}_{G}$, i.e., $\mcx \subseteq \mcc^{\text{(even)}}_{G}$, we let
\begin{align}
    \partial \mcx = \bigcup_{C\in\mcx}C
\end{align}
denote the set of vertices in $\mcx$. Let $\abs{\mcx}$ denote the number of even-hole components in $\mcx$, and $\abs{\partial\mcx}$ be the total length of all the elements of $\mcx$.

A \emph{clique} $K$ in $G$ is a subset of vertices such that every pair of vertices in $K$ is neighboring. A simplicial clique $K_s$ is a clique such that, for every vertex $\bsj \in K_s$, $\Gamma(\bsj)$ induces a clique in $G{\setminus}K_s$ (see Table~\ref{table:simplicialclaw}). A graph is \emph{simplicial} if it contains a simplicial clique. We say that a Hamiltonian is simplicial, claw-free (SCF) if its frustration graph is claw-free and contains a simplicial clique. Together, the aforementioned graph structures play an important role in the free-fermion solvability of a Hamiltonian with frustration graph $G$.

\section{Free-Fermion Models}
\label{section:FreeFermionModels}

\subsection{Exact Solutions}
\label{subsection:ExactSolutions}

A free-fermion Hamiltonian has the form
\begin{align}
    H_{\text{f}} &= i \sum_{(j,k) \in E_{\text{f}}}h_{jk}\gamma_{j}\gamma_{k} \label{equation:freefermiondefinition} \\
    &= \frac{i}{2}\bs{\gamma}^{\text{T}}\cdot\mathbf{h}\cdot\bs{\gamma}, \label{equation:freefermionmatrixdefinition}
\end{align}
where we collect the set of \emph{Majorana fermion modes} in the vector $\bs{\gamma}=(\gamma_j)$ and the Hamiltonian coefficients in the \emph{single-particle Hamiltonian} $\mathbf{h}=(h_{ij})$. We denote the set of Majorana indices by $V_{\text{f}}$ and the set of pairs $\{j,k\}$ of distinct elements of $V_{\text{f}}$ for which $h_{jk}$ is non-zero as $E_{\text{f}}$. The graph $R=(V_{\text{f}},E_{\text{f}})$ is the \emph{Majorana hopping graph}, and $\mathbf{h}$ is a weighted skew-adjacency matrix for $R$.

The Majorana modes satisfy the canonical anticommutation relations
\begin{equation}
    \{\gamma_j,\gamma_k\} = \gamma_j\gamma_k+\gamma_k\gamma_j = 2\delta_{jk}I. \label{equation:anticommutationrelation}
\end{equation}
Products of Majorana operators only commute or anticommute and square to $\pm I$. Without loss of generality, we take $\mathbf{h}$ to be antisymmetric, since any symmetric part of $\mathbf{h}$ will only contribute a physically irrelevant identity term to $H_{\text{f}}$ by Eq.~(\ref{equation:anticommutationrelation}).

The canonical anticommutation relations imply that linear combinations of the Majorana modes are preserved under commutation with the Hamiltonian.
\begin{equation}
    \operatorname{ad}_{iH_{\text{f}}}\gamma_j = [iH_{\text{f}},\gamma_j] = 2\left(\mathbf{h}\cdot\bs{\gamma}\right)_{j}, \label{equation:fermiongammacommutation}
\end{equation}
This gives 
\begin{equation}
    e^{iH_{\text{f}}t}\gamma_je^{-iH_{\text{f}}t} = \left(e^{2\mathbf{h} t}\cdot\bs{\gamma}\right)_j,
\end{equation}
and $e^{2\mathbf{h}t}$ is called the \emph{single-particle transition matrix}. Since $\mathbf{h}$ is antisymmetric, $e^{2\mathbf{h}t}$ is an orthogonal matrix in the group $SO(\abs{V_{\text{f}}})$. Thus, conjugation by free-fermion unitary evolution preserves the canonical anticommutation relations.

Similarly, we can find an orthogonal matrix $e^{\mathbf{w}}$ that block-diagonalizes $\mathbf{h}$ as
\begin{equation}
    e^{-\mathbf{w}} \cdot \mathbf{h} \cdot e^{\mathbf{w}} = \bigoplus_{j=1}^{\abs{V_\text{f}}/2}
    \begin{pmatrix}
        0 & -\varepsilon_j \\
        \varepsilon_j & 0
    \end{pmatrix},
    \label{equation:blockdiagonalization}
\end{equation}
if $\abs{V_{\text{f}}}$ is even. If $\abs{V_{\text{f}}}$ is odd, then the tensor sum runs to $\lfloor\abs{V_{\text{f}}}/2\rfloor$, and $\mathbf{h}$ has an additional zero eigenvalue. Choosing $W = e^{\frac{1}{4}\left(\bs{\gamma}^{\text{T}}\cdot\mathbf{w}\cdot\mathbf{\gamma}\right)}$ gives
\begin{equation}
    W^{\dagger}H_{\text{f}}W = -i\sum_{j=1}^{\lfloor\abs{V_{\text{f}}}/2\rfloor}\varepsilon_j \gamma_{2j-1}\gamma_{2j}. \label{equation:freefermiondiagonal}
\end{equation}
It is convenient to pair the Majorana modes to define the \emph{fermionic eigenmodes} $\{\psi_j\}_{j=1}^{\lfloor\abs{V_{\text{f}}}/2\rfloor}$ as
\begin{equation}
    \psi_j = \frac{1}{2}W\left(\gamma_{2j-1}+i\gamma_{2j}\right)W^{\dagger},
\end{equation}
for $j\in\{1,\dots,\lfloor\abs{V_{\text{f}}}/2\rfloor\}$. These operators satisfy the canonical anticommutation relations for fermionic ladder operators
\begin{equation}
    \{\psi_j,\psi_k\} = 0, \quad \{\psi_j,\psi^{\dagger}_k\} = \delta_{jk}I,
\end{equation}
and are defined such that 
\begin{equation}
    H_{\text{f}} = \sum_{j=1}^{\lfloor\abs{V_\text{f}}/2\rfloor}\varepsilon_j[\psi_j,\psi^{\dagger}_j]. \label{equation:freefermioneigen}
\end{equation}
The linear map $\operatorname{ad}_{H_{\text{f}}}$ satisfies the conventional eigenvector relation with respect to the eigenmodes
\begin{equation}
    [H_{\text{f}},\psi_j] = 2\varepsilon_j\psi_j. \label{equation:eigenrelation}
\end{equation}
From Eq.~(\ref{equation:freefermiondiagonal}), we see that the free-fermion Hamiltonian $H_{\text{f}}$ has spectrum given by
\begin{equation}
    \mce_{\mathbf{x}} = \sum_{j=1}^{\lfloor\abs{V_{\text{f}}}/2\rfloor}(-1)^{x_j}\varepsilon_j,
\end{equation}
for $\mathbf{x}\in\{0,1\}^{\times\lfloor\abs{V_{\text{f}}}/2\rfloor}$. The real quantities $\{\varepsilon_j\}_{j=1}^{\lfloor\abs{V_{\text{f}}}/2\rfloor}$, called the \emph{single-particle energies}, are defined to satisfy $\varepsilon_j\geq0$ for all $j$. We can generate eigenstates of $H_{\text{f}}$ by applying $W$ to the mutual eigenstates of the operators $\{-i\gamma_{2j-1}\gamma_{2j}\}_{j=1}^{\lfloor\abs{V_{\text{f}}}/2\rfloor}$.

From the block-diagonal form of $\mathbf{h}$ in Eq.~(\ref{equation:blockdiagonalization}), we have that the single-particle energies are the reciprocals of the non-negative roots of the characteristic polynomial
\begin{align}
    g_{\mathbf{h}}(u) &= \det(\mathbf{I}-iu\mathbf{h}) \\
    &= \sum_{U \subseteq V_{\text{f}}}(-iu)^{\abs{U}}\det\left(\mathbf{h}_{UU}\right). \label{equation:freefermioncharacteristicpolynomial}
\end{align}
Here, $\mathbf{h}_{UU}$ denotes the principal submatrix of $\mathbf{h}$ with rows and columns indexed by the elements of $U \subseteq V_{\text{f}}$. We shall express this polynomial in terms of the graph structures of $R$. Since $\mathbf{h}_{UU}$ is antisymmetric, its determinant is the square of the Pfaffian, i.e.,
\begin{equation}
    \det\left(\mathbf{h}_{UU}\right) = \text{Pf}\left(\mathbf{h}_{UU}\right)^2.
\end{equation}
The Pfaffian is defined to be
\begin{equation}
    \text{Pf}\left(\mathbf{h}_{UU}\right) = \sum_{M\in\mcm^{(\abs{U}/2)}_{R[U]}}(-1)^{\pi(M)}\prod_{\substack{(j,k) \in M \\ j < k}}h_{jk},
\end{equation}
if $\abs{U}$ is even, and zero if $\abs{U}$ is odd. The sign factor $(-1)^{\pi(M)}$ is defined implicitly as the factor incurred upon sorting the individual Majorana-mode factors in $\prod_{\{(j,k) \in M \mid j < k\}}\gamma_j\gamma_k$ such that indices are ascending from left to right. This gives 
\begin{equation}
    \begin{split}
        g_{\mathbf{h}}(u) &= \sum_{\substack{M,M'\in\mcm_R \\ \abs{M}=\abs{M'}}}(-u^2)^{\abs{M}}(-1)^{\pi(M)+\pi(M')} \\
        &\quad\times \left(\prod_{\substack{(j,k) \in M \\ j < k}}h_{jk}\right) \left(\prod_{\substack{(j,k) \in M' \\ j < k}}h_{jk}\right).
    \end{split}
    \label{equation:characteristicpolynomial}
\end{equation}
We return to this explicit expression in the following sections.

Let us now consider nesting commutators by repeated application of Eq.~(\ref{equation:fermiongammacommutation}),
\begin{equation}
    \operatorname{ad}_{iH_\text{f}}^{k}\gamma_j = 2^k(\mathbf{h}^k\cdot\bs{\gamma})_j.
\end{equation}
Note that we now include a factor of $i$ with the Hamiltonian to ensure that the resulting linear map preserves Hermiticity as in Theorem~\ref{theorem:polynomialdivisibility}. We see that $\operatorname{ad}_{iH_\text{f}}$ as a linear map satisfies the same characteristic polynomial as $2\mathbf{h}$, as we expect from Eq.~(\ref{equation:eigenrelation}). Furthermore, we have
\begin{equation}
    \{\operatorname{ad}_{iH_{\text{f}}}^k\gamma_j,\operatorname{ad}_{iH_{\text{f}}}^{\ell}\gamma_j\} =2(-1)^k2^{k+\ell}\left(\mathbf{h}^{k+\ell}\right)_{jj}I.
\end{equation}
This relation is symmetric in $k$ and $\ell$, as we expect, since the diagonal element of $\mathbf{h}^{k+\ell}$ vanishes if $k$ and $\ell$ have opposite parity due to the antisymmetry of $\mathbf{h}$. Thus, $(-1)^k=(-1)^{\ell}$ if the expression does not vanish. Analogs to these relations are indicative of the existence of a free-fermion solution, as we shall see in Section~\ref{section:KrylovSubspaces}.

\subsection{Generalized Jordan-Wigner Solutions}
\label{subsection:GeneralizedJordanWignerSolutions}

Remarkably, the exact solvability of free-fermion models can be leveraged to find exact solutions to spin models, which are not explicitly given in terms of free fermions. If one can find the proper identification of spin and fermionic degrees of freedom, one can treat the spin model as an effective free-fermion model and apply the exact solution method. One way to do this is to identify each term in a given spin model $H$ with a corresponding term in $H_{\text{f}}$ such that commutation relations between terms are preserved. Such a solution is called a \emph{generator-to-generator mapping}. This family of solutions includes the Jordan-Wigner transformation as well as the exact solution to the Kitaev honeycomb model. For this reason, it is also called a \emph{generalized Jordan-Wigner solution}.

The Jordan-Wigner transformation of a 1D spin system on $n$ qubits defines $2n$ effective Majorana modes via
\begin{align}
    \gamma_{2j-1} &= \sigma^z_1\otimes\dots\otimes\sigma^z_{j-1}\otimes\sigma^{x}_j, \\ 
    \gamma_{2j} &= \sigma^z_1\otimes\dots\otimes\sigma^z_{j-1}\otimes\sigma^{y}_j, 
    \label{equation:jordanwignertransformation}
\end{align}
for $j\in\{1,\dots,n\}$. It is well-known for its application to the 1D XY model~\cite{lieb1961two},
\begin{align}
    H_{\text{XY}} &= \sum_{j=1}^{n-1}\left[(1-\delta)\sigma^y_j\sigma^y_{j+1}+(1+\delta)\sigma^x_j\sigma^x_{j+1}\right] \\
    \mapsto &i\sum_{j=1}^{n-1}\left[(1-\delta)\gamma_{2j-1}\gamma_{2j+2}-(1+\delta)\gamma_{2j}\gamma_{2j+1}\right]. \label{equation:1dxy}
\end{align}
This defines an effective single-particle matrix for $H_{\text{XY}}$ and allows us to solve the model in the fermionic picture according to the procedure of the Section~\ref{subsection:ExactSolutions}.

Graph theory allows us to make this procedure systematic using the frustration graph of $H$. For a given free-fermion Hamiltonian $H_{\text{f}}$, the frustration graph has a particular structure due to the canonical anticommutation relations. It is the graph whose vertices are the edges of the fermion hopping graph $R$ with vertices adjacent in the frustration graph if and only if the corresponding edges of $R$ are incident. That is, the frustration graph of $H_{\text{f}}$ is the \emph{line graph} of $R$.
\begin{definition}[Line graph]
    \label{definition:linegraph}
    Given a graph $R=(V,E)$, the \emph{line graph} $L(R)=(E,F)$ of $R$ is the graph whose vertices correspond to the edges of $R$ and
    \begin{equation}
        F = \{\{e,f\} \mid e,f\in E,\ \abs{e \cap f}=1\}.
    \end{equation}
    $R$ is called the \emph{root graph} of $L(R)$.
\end{definition}
Not every graph can be realized as the line graph of another graph. In particular, note that a line graph is always claw-free, as three edges cannot all be incident to a fourth edge without at least two of those edges being incident to each other. In fact, line graphs are characterized by a complete set of nine forbidden subgraphs which includes the claw~\cite{beineke1970characterizations}. Using the notation of Def.~\ref{definition:linegraph}, a vertex $j$ in $R$ is mapped to the clique $K_j=\{\{j,k\}\}_{k\in\Gamma(j)}$ in $L(R)$ under the line graph mapping. Since every edge has two vertices, it is an equivalent characterization of line graphs that the edges of a line graph $L(R)$ can be partitioned into cliques such that every vertex in $L(R)$ is a member of at most two cliques~\cite{krausz1943demonstration}. Furthermore, every such clique $K_j$ in $L(R)$ is \emph{simplicial} since $\Gamma_{L(R){\setminus}K_j}[e] = K_k$ for $e=\{j,k\} \in E$. Thus, line graphs are simplicial and claw-free. A path of $\ell+1$ vertices in $R$ is mapped to an induced path of length $\ell$ in $L(R)$, and a cycle of $\ell$ vertices in $R$ is mapped to a hole of length $\ell$ in $L(R)$. Finally, a matching of size $k$ in $R$ is mapped to an independent set of order $k$ in $L(R)$.

In Ref.~\cite{chapman2020characterization}, it was shown that an injective generator-to-generator mapping from $H$ to a free-fermion Hamiltonian $H_{\text{f}}$ exists if and only if its frustration graph is a line graph. This allows us to associate each vertex $\bsj \in V$ of the frustration graph $G$ with an edge $\varphi(\bsj)=\{\varphi_1(\bsj),\varphi_2(\bsj)\}$ of the fermion hopping graph $R$. We choose an ordering on the vertices $V_{\text{f}}$, and our convention on $\varphi$ is such that $\varphi_1(\bsj)<\varphi_2(\bsj)$. We then define a free-fermion solution to $H$ via this mapping as
\begin{equation}
    h_{\bsj} \mapsto i\wt{b}_{\bsj}\gamma_{\varphi_1(\bsj)}\gamma_{\varphi_2(\bsj)},
\end{equation}
for all $\bsj \in V$. Here, $\wt{b}_{\bsj}\in\mathbb{R}$ is an as yet unspecified coupling strength for the effective free-fermion Hamiltonian. Since the coupling strengths $\{\wt{b}_{\bsj}\}_{\bsj \in V}$ can be varied without changing the commutation relations between terms, we cannot determine them from frustration graph alone. Instead, the values of these coefficients are determined by constraints on products of the spin Hamiltonian terms. The constraint that $h^2_{\bsj}=b^2_{\bsj}I$ guarantees that we must have $\wt{b}_{\bsj}=\pm b_{\bsj}$ for all $\bsj \in V$. This gives
\begin{align}
    h_{\bsj} \mapsto i(-1)^{\tau(\bsj)}\abs{b_{\bsj}}\gamma_{\varphi_1(\bsj)}\gamma_{\varphi_2(\bsj)}, \mbox{\hspace{5mm}} 
    \label{equation:generalizedjordanwigner}
\end{align}
for all $\bsj \in V$ and where $\tau:V\to\{0,1\}$. We fix this remaining sign freedom with the constraints given by restricting $H$ to a mutual eigenspace of its commuting \emph{cycle symmetries}.

Let $C_{\text{f}}=\{j_i\}_{i=0}^{\ell-1} \subseteq V_{\text{f}}$ be a cycle in $R$, and denote $C=\{\bsj_i\}_{i=0}^{\ell-1} \subseteq V$ as the \emph{hole} in $G$ such that $\varphi(\bsj_i)=\{j_i,j_{i+1}\} \in E_{\text{f}}$, with index addition taken modulo $\ell$. (We take set equality in this assumption on $\varphi(\bsj_i)$, so we do not assume anything about the vertex ordering of $j_i$ and $j_{i+1}$ here.) We have
\begin{equation}
    \prod_{i=0}^{\ell-1}\left(\gamma_{j_i}\gamma_{j_{i+1}}\right) = \prod_{j \in C_{\text{f}}}\gamma_j^2 = I.
\end{equation}
Thus, we have that the corresponding \emph{cycle symmetry} $h_C$ commutes with every term in $H$ by the requirement that the generator-to-generator mapping preserve commutation relations. Furthermore, cycle symmetry operators $h_C$ and $h_{C'}$ mutually commute for distinct holes $C,C' \subset V$ because they themselves are products of terms from $H$. While $h_C$ is not proportional to the identity in general, $h_C^2$ is. Under the free-fermion solution, $h_C$ maps to 
\begin{equation}
    \prod_{\bsj \in C} \left(i(-1)^{\tau(\bsj)}\abs{b_{\bsj}}\gamma_{\varphi_1(\bsj)}\gamma_{\varphi_2(\bsj)}\right) = \pm i^{\abs{C}}\prod_{j \in C}\abs{b_{\bsj}}.
\end{equation}
Thus, choosing a sign assignment $\tau$ is equivalent to fixing a mutual set of symmetry eigenvalues, and we have a unique free-fermion solution for every such restriction. It is explained in Ref.~\cite{chapman2020characterization} that this choice of $\tau$ naturally corresponds to an orientation on $R$, and it was shown that there exists an orientation satisfying every set of independent cycle-symmetry eigenvalues. Beyond this choice, there is no additional freedom in the free-fermion mapping.

Finally, it may be the case that there are additional states in the free-fermion model that do not correspond to physical states for the spin model, and we must restrict to a subspace of the free-fermion model to recover the spin-model solution. These additional states correspond to a fixed eigenspace of the parity operator
\begin{equation}
    \text{P} = i^{\frac{1}{2}\abs{V_{\text{f}}}(\abs{V_{\text{f}}}-1)}\prod_{j \in V_{\text{f}}}\gamma_j,
\end{equation}
which is always a symmetry of $H_{\text{f}}$. When $\abs{V_{\text{f}}}$ is even, this operator can be constructed from terms in $H_{\text{f}}$ using edges from a structure called a \emph{T-join} of $R$. A perfect matching is a special case of a $T$-join, and a graph with even order has a $T$-join even if it does not have a perfect matching. However, the corresponding product of terms from $H$ may give the identity (up to cycle symmetries), so we must only consider a fixed-parity restriction of $H_{\text{f}}$ as a solution to the spin model in this case.

Unlike the conventional Jordan-Wigner transform of Eq.~(\ref{equation:jordanwignertransformation}), the generalized Jordan-Wigner solutions given by Eq.~(\ref{equation:generalizedjordanwigner}) treat Majorana quadratics as fundamental, without explicitly defining the individual Majorana modes. We can in fact recover Eq.~(\ref{equation:jordanwignertransformation}) by considering the line-graph mapping on the 1D XY model in Eq.~(\ref{equation:1dxy}). First, note that the frustration graph of $H_{\text{XY}}$ consists of two disconnected path components corresponding to the collections of terms
\begin{align}
    H_1 &= (1-\delta)\sum_{j=1}^{\lfloor\frac{n}{2}\rfloor}\sigma^y_{2j-1}\sigma^y_{2j}+(1+\delta)\sum_{j=1}^{\lfloor\frac{n-1}{2}\rfloor}\sigma^x_{2j}\sigma^x_{2j+1}, \notag \\
    H_2 &= (1+\delta)\sum_{j=1}^{\lfloor\frac{n}{2}\rfloor}\sigma^x_{2j-1}\sigma^x_{2j}+(1-\delta)\sum_{j=1}^{\lfloor\frac{n-1}{2}\rfloor} \sigma^y_{2j}\sigma^y_{2j+1}. \notag
\end{align}
We shall only recover the Jordan-Wigner transformation for $H_1$ as $H_2$ can be handled similarly. Let $P$ be the path component of $G$ corresponding to the terms in $H_1$. The operator $\sigma_1^y\sigma_2^y$ corresponds to an endpoint of $P$, so it is a simplicial clique in $P$. We thus introduce the simplicial mode $\chi=\sigma^{\js}=\sigma^x_1$ (defined formally in Section~\ref{subsection:DisguisedFreeFermionModes} as Eq.~\ref{definition:simplicialmode}), which gives that $P\cup\{\js\}$ is again a path component of $\gs$. Label vertices in $P\cup\{\js\}$ (and corresponding terms in $H_1$) according to their distance from $\bsj_0=\js$ as in Eq.~(\ref{equation:pathlabel}) with $\ell=n-1$. Using the fact that $P\cup\{\js\}$ is a line graph, we apply the generalized Jordan-Wigner solution as
\begin{equation}
    b_{\bsj_k}^{-1}h_{\bsj_k} \mapsto
    \begin{cases}
        i\gamma_0\gamma_1 & k=0, \\
        i\gamma_{2k-1}\gamma_{2k+2} & k>0,\ \text{odd}, \\
        -i\gamma_{2k}\gamma_{2k+1} & k>0,\ \text{even},
    \end{cases}
\end{equation}
where we take $b_{\bsj_0}=1$. Since there are no holes in $\gs$, we are free to choose any sign configuration $\tau$ for the free-fermion terms, so we choose this and the Majorana-mode labeling to be consistent with the Jordan-Wigner transformation. We then take
\begin{align}
    &(-i)^{(k\bmod2)}(\Pi_{m=0}^kb_{\bsj_m})^{-1}h_{P_k} \notag \\ 
    &\quad=
    \begin{cases}
        \sigma^z_1\otimes\dots\otimes\sigma^z_{k}\otimes\sigma^{y}_{k+1} & k>0,\ \text{odd}, \\ 
        \sigma^z_1\otimes\dots\otimes\sigma^z_{k}\otimes\sigma^{x}_{k+1} & k\geq0,\ \text{even},
    \end{cases} \\
    &\quad\mapsto 
    \begin{cases}
        i\gamma_0 \gamma_{2(k+1)} & k>0,\ \text{odd}, \\
        i\gamma_0 \gamma_{2(k+1)-1} & k\geq0,\ \text{even}.
    \end{cases}
\end{align}
We have chosen phases to be consistent with the Jordan-Wigner transformation such that the resulting operators are Hermitian, and we have canceled interior Majorana factors under the free-fermion mapping. Finally, we note that the operators $\{i\gamma_{0}\gamma_{k}\}$ also satisfy the canonical anticommutation relations
\begin{align}
    \{i\gamma_{0}\gamma_{j},i\gamma_{0}\gamma_{k}\} = 2\delta_{jk}I,
\end{align}
so we are free to ignore the factor of $i\gamma_0$ in our definition of the Majorana modes.
This factor does not appear in any quadratic products $h_{P_j}h_{P_k}$ from which we can generate all products of Hamiltonian terms in $H_1$. We can similarly recover the remaining definitions for the Majorana modes from $H_2$.

From this example, we see that the Jordan-Wigner transform can be defined operationally in terms of the specific model under consideration. Moreover, we can interpret individual fermion modes as being localized at the endpoints of path operators in the frustration graph. When holes are present in $G$, the interiors of these paths are defined up to symmetries of the model. An analogous structure holds in the general case.

\subsection{Disguised Free-Fermion Modes}
\label{subsection:DisguisedFreeFermionModes}

When $G$ is not a line graph, there is no direct mapping from terms in $H$ to terms in \emph{any} free-fermion Hamiltonian $H_\text{f}$ of the form in Eq.~(\ref{equation:freefermiondefinition}). However, it may still be possible to find a free-fermion solution to $H$. The essential idea is to treat independent sets, cliques, induced paths, and holes algebraically as though they result from a line graph mapping, even though there is no associated root graph. Our foray into this problem is the machinery of transfer matrices. We define the following set of operators, related to independent sets of the frustration graph $G$.

\begin{definition}[Independent set charges]
    \label{definition:independentsetcharges}
    The independent set charges $\{\qkg{k}{G}\}_{k=0}^{\alpha(G)}$ are defined as the sum over independent sets of order $k$ in $G$ by
    \begin{equation}
        \qkg{k}{G} \coloneqq \sum_{S\in\mcs^{(k)}_G}h_S. \label{equation:independentsetcharges}
    \end{equation}
    Note that $\qkg{0}{G}=I$ and $\qkg{1}{G}=H$.
\end{definition}

It was shown in Ref.~\cite{elman2021free} that the independent set charges commute with each other when $G$ is claw-free. Since $\qkg{1}{G}=H$, the independent set charges are conserved as well.

The independent set charges satisfy a recursion relation due to the fact that no independent set can have more than one vertex in a clique of $G$. For a given clique $K \subseteq V$, we partition terms in $\qkg{k}{G}$ according to independent sets with no vertices in $K$ and those with exactly one vertex in $K$. This gives
\begin{equation}
    \qkg{k}{G} = \qkg{k}{G{\setminus}K}+\sum_{\bsj \in K}h_{\bsj}\qkg{k-1}{G{\setminus}\Gamma[\bsj]}. \label{equation:chargecliquerecursion}
\end{equation}
We next define the transfer operator.
\begin{definition}[Transfer operator~\cite{fendley2019free}]
    \label{definition:transferoperator}
    Let $H$ be a Hamiltonian with frustration graph $G$. The \emph{transfer operator} $T_{G}(u)$ is defined as the generating function of the independent set charges
    \begin{align}
        T_{G}(u) &\coloneqq \sum_{S\in\mcs_{G}}(-u)^{\abs{S}}h_S \\
        &= \sum_{k=0}^{\alpha(G)}(-u)^k\qkg{k}{G},
    \end{align}
    where $u\in \mathbb{R}$ is the \emph{spectral parameter}.
\end{definition}
The transfer operator can be viewed as the operator-valued analogue for the independence polynomial of a graph.
\begin{definition}[Weighted independence polynomial]
    \label{definition:independencepolynomial}
    \begin{equation}
        I_{G}(x) \coloneqq \sum_{S\in\mcs_{G}}x^{\abs{S}}\prod_{\bsj \in S}b_{\bsj}^2.
    \end{equation}
\end{definition}
Both the transfer operator and the independence polynomial satisfy similar recursion relations to Eq.~(\ref{equation:chargecliquerecursion}), i.e.,
\begin{align}
    T_G(u) &= T_{G{\setminus}K}(u)-u\sum_{\bsj \in K}h_{\bsj}T_{G{\setminus}\Gamma[\bsj]}(u), \label{equation:transfermatrixrecursion} \\
    I_G(x) &= I_{G{\setminus}K}(x)+x\sum_{\bsj \in K}b_{\bsj}^2I_{G{\setminus}\Gamma[\bsj]}(x). \label{equation:independencepolynomialrecursion}
\end{align}
A special case of these recursion relations is that for which $K$ consists of a single vertex.

The transfer operator bears an interesting relation to the characteristic polynomial for the free-fermion model.
\begin{definition}[Generalized characteristic polynomial]
    \label{definition:generalizedcharacteristicpolynomial}
    Let $H$ be an SCF Hamiltonian with frustration graph $G$. The \emph{generalized characteristic polynomial} $Z_{G}(-u^2)$ of $G$ is defined as
    \begin{equation}
        Z_{G}(-u^2) \coloneqq T_G(u)T_G(-u).
    \end{equation}
\end{definition}
Strictly speaking, $Z_{G}(-u^2)$ is an operator-valued polynomial, but reduces to an ordinary polynomial when restricted to a symmetric subspace as we explain. In Appendix~\ref{section:GeneralizedCharacteristicPolynomialProof}, we show that when $G$ is claw-free, we have
\begin{align}
    Z_{G}(-u^2) &= \!\!\!\!\sum_{\substack{S,S'\in\mcs_{G} \\ S{\oplus}S'=\partial\mcx \\ \mcx\in\mathscr{C}^{\text{(even)}}_G}}(-u^2)^{\abs{S}}h_Sh_{S'}
    \label{equation:setsymmetricrepresentation} \\
    &= \!\!\!\!\sum_{\substack{S,S'\in\mcs_{G} \\ S{\oplus}S'=\partial\mcx \\ \mcx\in\mathscr{C}^{\text{(even)}}_G}}(-u^2)^{\abs{S}}\left(\prod_{\bsj \in S \cap S'} b_{\bsj}^2\right)\prod_{C\in\mcx}h_C. \label{equation:cyclesymmetries}
\end{align}
The reason for calling $Z_{G}(-u^2)$ the generalized characteristic polynomial becomes clear when we consider the case where $G$ is a line graph $L(R)$. Suppose this is the case. Since the $\{h_C\}_{C\in\mcc_{G}^{(\text{even})}}$ are symmetries of the Hamiltonian, we can restrict to a mutual eigenspace of these symmetries through the orientation $\tau$ of $R$ implicit in the choice of sign configuration for $\mathbf{h}$. Denote the restriction of $Z_{G}(-u^2)$ to this subspace by 
\begin{equation}
    Z_{G,\tau}(-u^2) \coloneqq Z_{G}(-u^2)\Pi_{\tau},
\end{equation}
where $\Pi_{\tau}$ is the projector onto the subspace. Furthermore, denote $M\in\mcm_{R}^{(\abs{S})}$ as the matching of $R$ corresponding to the independent set $S$ in $L(R)$ and, similarly, denote $M'\in\mcm_{R}^{(\abs{S})}$ as the matching corresponding to $S'$. Under the free-fermion solution, the operator $h_Sh_{S'}$ maps to
\begin{equation}
    \begin{split}
        &\left[\prod_{\substack{(j,k) \in M \\ j<k}}(ih_{jk}\gamma_{j}\gamma_{k})\right] \left[\prod_{\substack{(j,k) \in M' \\ j<k}}(ih_{jk}\gamma_{j}\gamma_{k})\right] \\
        &= (-1)^{\pi(M)+\pi(M')}\left(\prod_{\substack{(j,k) \in M \\ j<k}}h_{jk}\right) \left(\prod_{\substack{(j,k) \in M' \\ j<k}}h_{jk}\right),
    \end{split}
    \label{equation:matchingproduct}
\end{equation}
since every Majorana mode is either included zero times or twice in this product by the constraint on $S{\oplus}S'$. The phase factor is calculated by sorting the Majorana modes in each matching individually, giving a factor of $(-1)^{\pi(M)+\pi(M')}$. Squaring the sorted operator gives a factor of $i^{2\abs{M}}(-1)^{\frac{1}{2}(2\abs{M})(2\abs{M}-1)}=1$. Comparing Eq.~(\ref{equation:setsymmetricrepresentation}) and Eq.~(\ref{equation:matchingproduct}) to Eq.~(\ref{equation:characteristicpolynomial}) gives
\begin{equation}
    Z_{L(R), \tau}(-u^2) = g_{\mathbf{h}}(u).
\end{equation}
In the case where $G$ is a general SCF graph, Eq.~(\ref{equation:setsymmetricrepresentation}) and Eq.~(\ref{equation:cyclesymmetries}) still hold, but the $h_C$ are not generally symmetries of the Hamiltonian. To recover a set of cycle-like symmetries, we sum the $h_C$ over subsets, denoted $\avg{C_0}$, of even holes with the same neighborhood, and the solution follows analogously. We denote these generalized cycle symmetries as $\jkg{C_0}{G}$, and we formally define them in Section~\ref{section:ClawFreeGraphs}. We label the mutual eigenspace of these symmetries by $\mcj$.

To finally give the full solution, we describe the fermionic eigenmodes. For each mutual eigenspace $\mcj$ of the generalized cycle symmetries we have a separate set of fermionic eigenmodes given in terms of the transfer operator and a \emph{simplicial mode}.
\begin{definition}[Simplicial mode]
    \label{definition:simplicialmode}
    Let $H$ be an SCF Hamiltonian with frustration graph $G$. A \emph{simplicial mode} with respect to a simplicial clique $K_s$ is a Pauli operator $\chi=\sigma^{\js}$ such that $\js$ is not in $V$, and $\chi$ satisfies 
    \begin{equation}
        \langle\js,\bsk\rangle = \delta_{\bsk \in K_s}.
    \end{equation}
    That is, $\chi$ only anticommutes with terms in $H$ whose vertices in $G$ are in the simplicial clique $K_s$.
\end{definition}
The eigenmodes are then given by the \emph{incognito modes}.
\begin{definition}[Incognito modes]
    With the preceding definitions, define the incognito modes on the $\mcj$-eigenspace by
    \begin{equation}
        \psi_{\mcj,j} \coloneqq N_{\mcj,j}^{-1}\Pi_{\mcj}T(-u_{\mcj,j}) \chi T(u_{\mcj,j}),
    \end{equation}
    where $u_{\mcj,j}$ satisfies $Z_{G,\mcj}(-u_{\mcj,j}^2)=0$.
\end{definition}
We prove that the incognito modes are the fermionic eigenmodes as Theorem~\ref{theorem:freefermionsolution}. In Ref.~\cite{elman2021free}, the corresponding result is proven in the special case that $G$ is even-hole-free and claw-free. It was shown in Ref.~\cite{chudnovsky2021note} that these graphs are always simplicial. In this case, there is only one symmetry sector $\mcj$, for which $\Pi_{\mcj}=I$, so the mapping to the free-fermion model is direct in this sense. Additionally, considering Eq.~(\ref{equation:cyclesymmetries}) when there are no even holes in $G$, we see that the only non-vanishing terms in Eq.~(\ref{equation:setsymmetricrepresentation}) are those for which $S=S'$, and the generalized characteristic polynomial coincides with the weighted independence polynomial
\begin{equation}
    Z_{G}(-u^2) = I_G(-u^2).
\end{equation}
It is a well-known result that the independence polynomial of a claw-free graph has real negative roots~\cite{chudnovsky2007roots}, and the generalization to the vertex-weighted case can be seen from Ref.~\cite{leake2019generalizations}.

Before turning to the proofs of our main results, we review some technical properties of claw-free graphs that are important for our purposes and formalize our definitions.

\section{Claw-Free Graphs}
\label{section:ClawFreeGraphs}

Claw-free graphs generalize line graphs in a natural way, and they have a rich characterization~\cite{chudnovsky2005structure}. In this section we describe some structural relations that claw-free graphs satisfy. We also describe some more technical properties satisfied by claw-free graphs containing a simplicial clique.

\subsection{Neighboring Relations}
\label{subsection:NeighboringRelations}

The following is a well-known fact about claw-free graphs, which we set apart by stating as a lemma and briefly prove for completeness.
\begin{lemma}
    \label{lemma:clawfreesymmetricdifference}
    Let $S$ and $S'$ be independent sets in a claw-free graph $G$. The graph $G[S{\oplus}S']$ induced by the symmetric difference of $S$ and $S'$ is a bipartite graph of maximum degree at most two.
\end{lemma}
\begin{proof}
    Clearly $G[S{\oplus}S']$ is bipartite with coloring classes $S{\setminus}S'$ and $S'{\setminus}S$, which are both independent sets by definition. If any vertex $\bsj \in S{\oplus}S'$ has degree greater than two in $G[S{\oplus}S']$, then $\bsj$, together with any three of its neighbors in $G[S{\oplus}S']$, induce a claw in $G$.
\end{proof}
Lemma~\ref{lemma:clawfreesymmetricdifference} implies that $G[S{\oplus}S']$ is a union of disjoint isolated vertices, induced paths, and even holes (odd holes are not bipartite).

\begin{table*}[t]
    \centering
    \setcellgapes{3pt}
    \makegapedcells
    \begin{tabular}{cccccccc}
        \hline\hline
        \multicolumn{4}{c}{Commuting} & \multicolumn{4}{c}{Anticommuting} \\
        \cmidrule(lr){1-4} \cmidrule(lr){5-8} \\
        \multicolumn{4}{c}{\makecell{(a.i)\vspace{10mm}} \makecell{\includegraphics[width=0.35\textwidth]{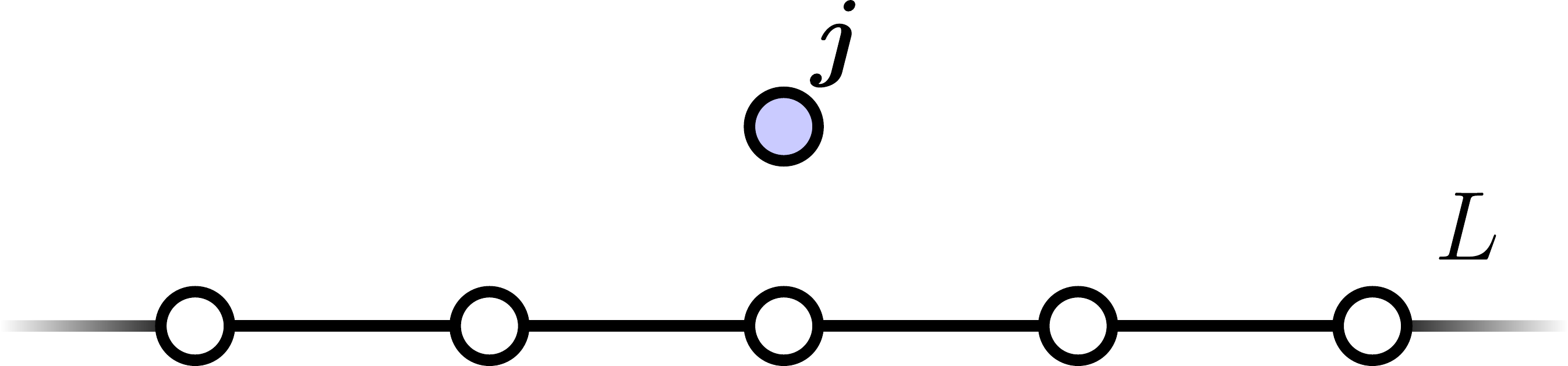}}} & \multicolumn{4}{c}{\makecell{(b.i)\vspace{10mm}} \makecell{\includegraphics[width=0.35\textwidth]{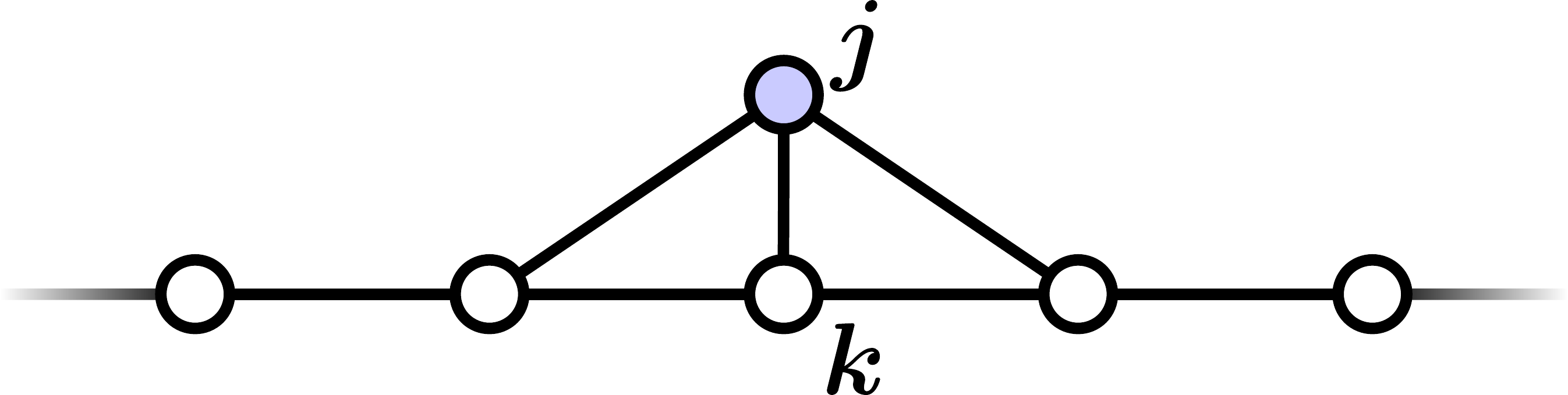}}} \\
        \multicolumn{4}{c}{\makecell{(a.ii)\vspace{10mm}} \makecell{\includegraphics[width=0.35\textwidth]{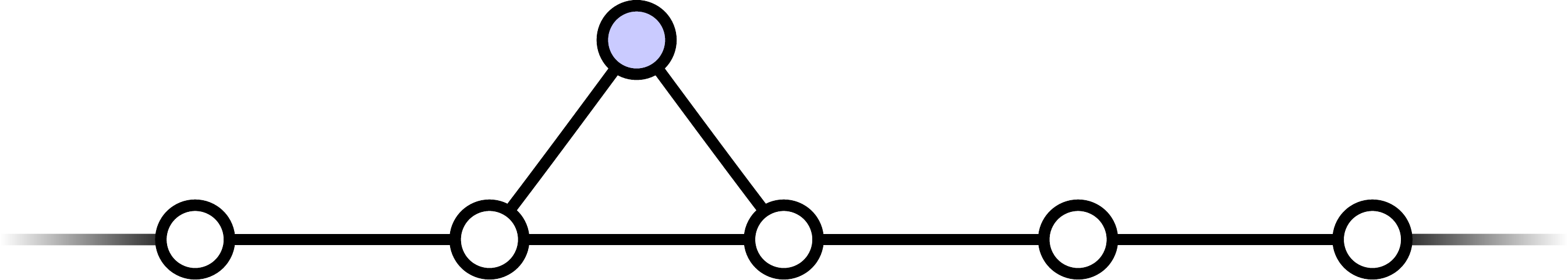}}} & \multicolumn{4}{c}{} \\
        \multicolumn{4}{c}{\makecell{(a.iii)\vspace{10mm}} \makecell{\includegraphics[width=0.35\textwidth]{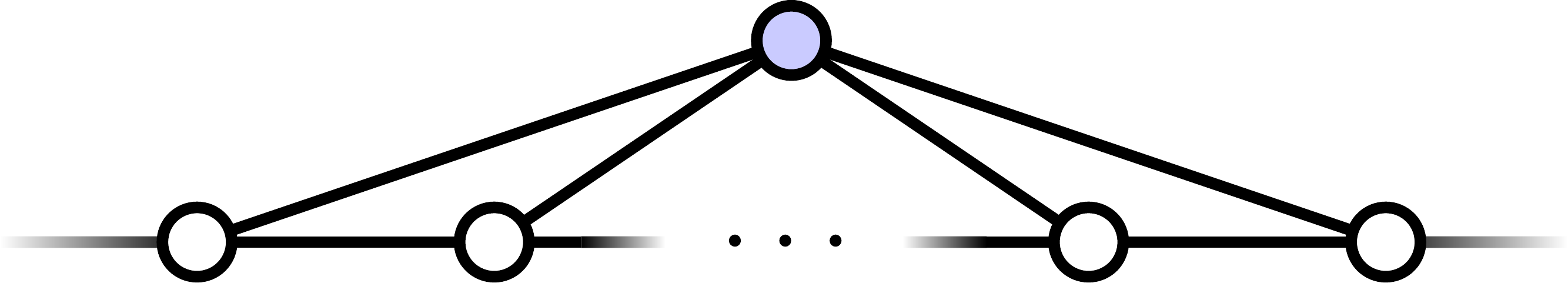}}} & & & & \\
        \multicolumn{4}{c}{} & & & & \\
        \cmidrule(lr){1-4} \cmidrule(lr){5-8}  \\
        \multicolumn{2}{c}{Hole} & \multicolumn{2}{c}{Path} & \multicolumn{2}{c}{Hole} & \multicolumn{2}{c}{Path} \\
        \cmidrule(lr){1-4} \cmidrule(lr){5-8} \\
        \multicolumn{2}{c}{\makecell{(a.iv)\vspace{10mm}} \makecell{\includegraphics[width=0.12\textwidth]{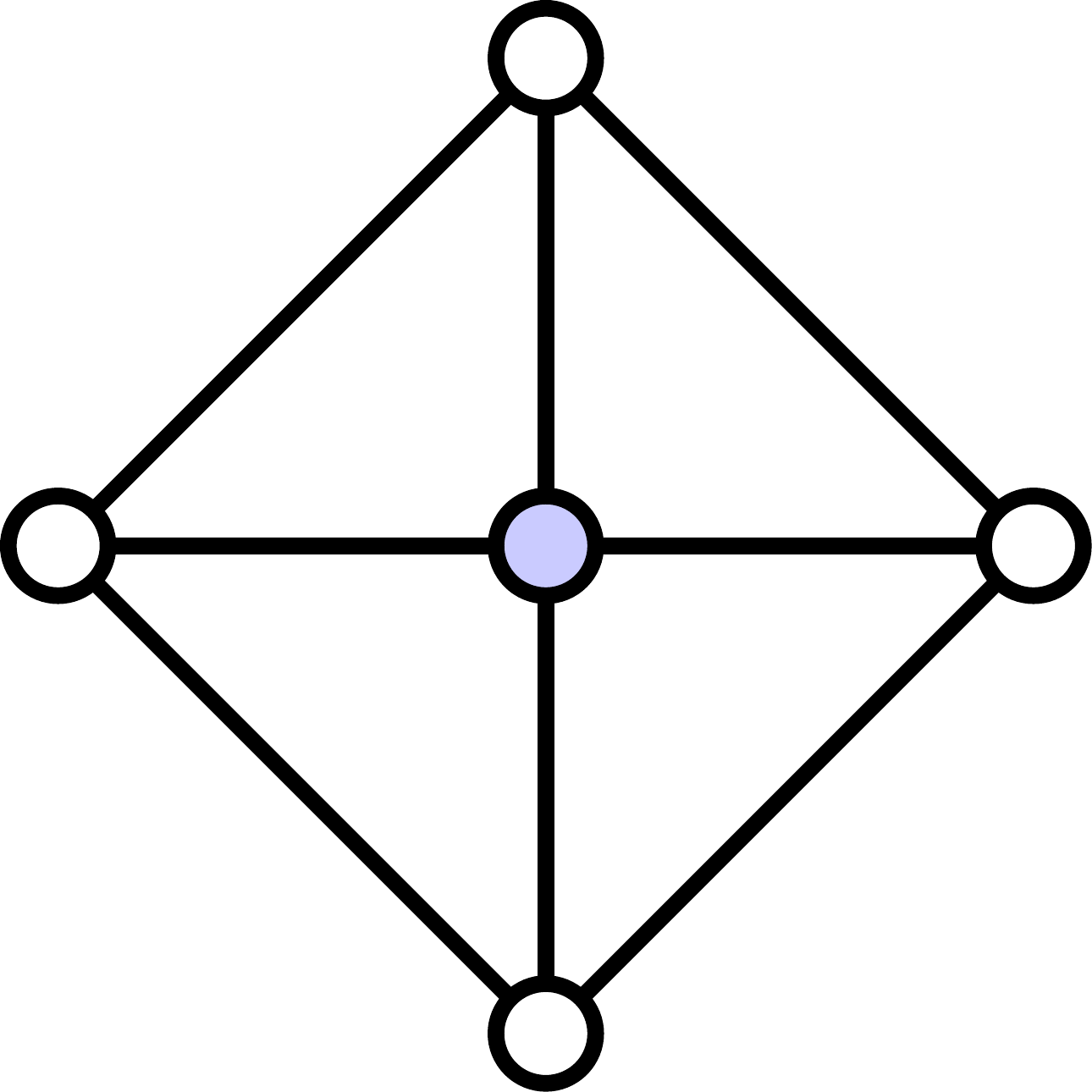}}} & \multicolumn{2}{c}{\makecell{(a.v)\vspace{10mm}} \makecell{\includegraphics[width=0.25\textwidth]{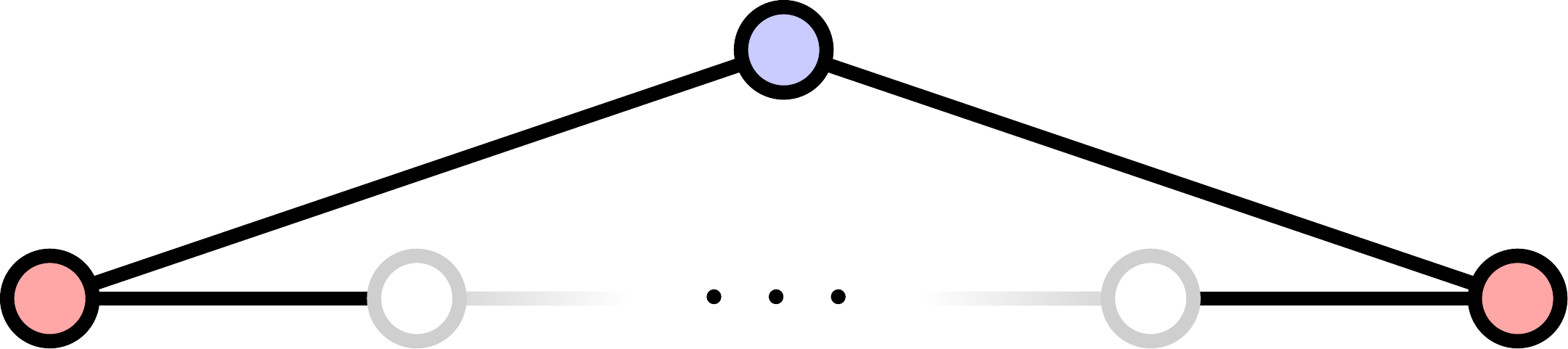}}} & \multicolumn{2}{c}{\makecell{(b.ii)\vspace{10mm}} \makecell{\includegraphics[width=0.12\textwidth]{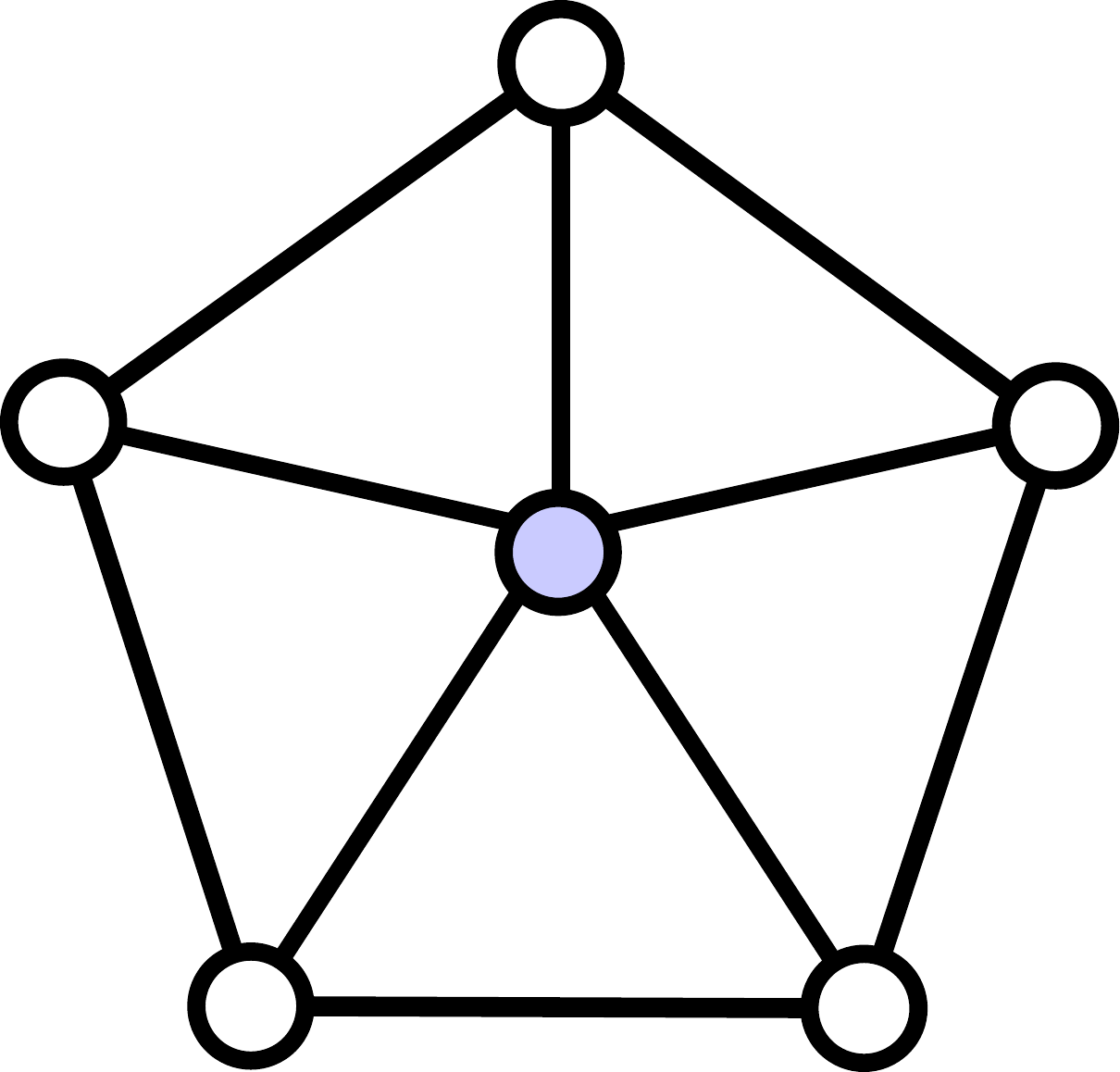}}} & \multicolumn{2}{c}{\makecell{(b.iii)\vspace{10mm}} \makecell{\includegraphics[width=0.25\textwidth]{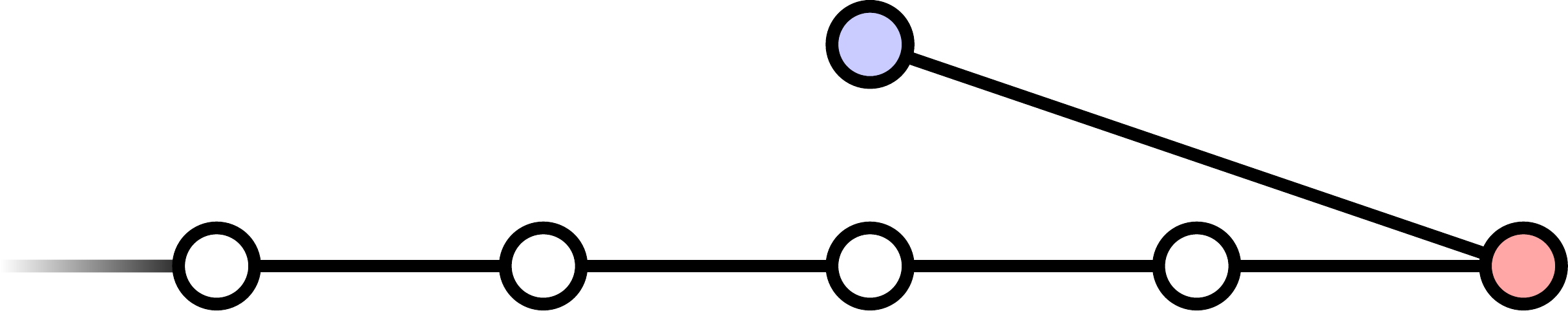}}} \\
        & & & & & & \multicolumn{2}{c}{\makecell{(b.iv)\vspace{10mm}} \makecell{\includegraphics[width=0.25\textwidth]{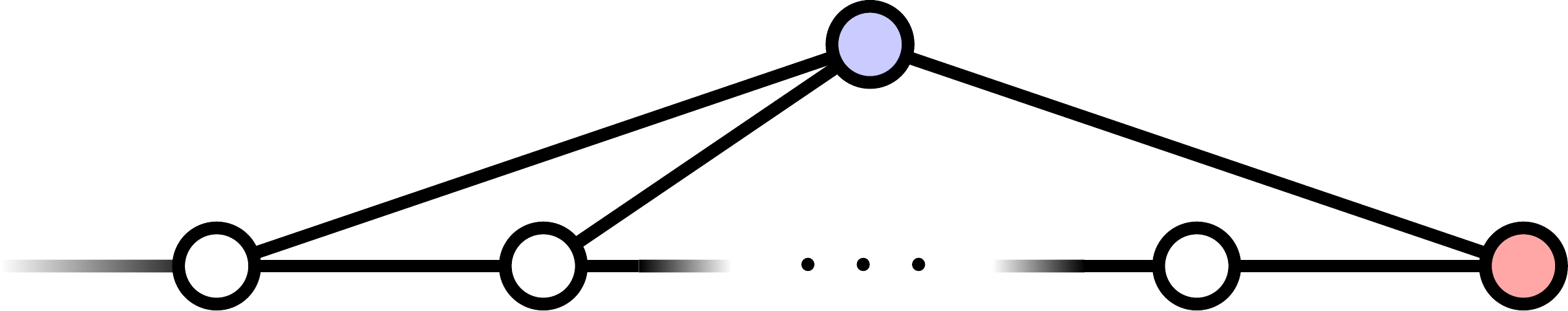}}}    
    \end{tabular}
    \caption{Possible neighboring relations between a hole or induced path $L$ and a vertex $\bsj$ not in $L$ in a claw-free graph. (a) The Hamiltonian term $h_{\bsj}$ commutes with $h_L$ only if $\bsj$ has even-many neighbors in $L$, and $\bsj$ has at most four neighbors in $L$ in this case. If $\bsj$ has two neighbors in $L$, they must be (a.ii) neighboring, or (a.v) the endpoints of an induced path of at least one edge. If $\bsj$ has four neighbors in $L$, they may induce (a.iii) two disjoint edges, a path of length three, or (a.iv) a hole of length four in $L$. (b) The Hamiltonian term $h_{\bsj}$ anticommutes with $h_L$ only if $\bsj$ has at most five neighbors in $L$. If $\bsj$ has three neighbors in $L$, they must (b.i) induce a path of length two in $L$, unless (b.iv) $L$ is an induced path, in which case $\bsj$ can neighbor an endpoint and any pair of neighboring vertices in $L$. (b.iii) The only possibility for $\bsj$ to have one neighbor in $L$ is if $L$ is an induced path, and $\bsj$ is neighboring its endpoint. (b.ii) The only possibility for $\bsj$ to have five neighbors in $L$ is if $L$ is a hole of length five. In case (b.i), we can define a unique additional hole or induced path by the single-vertex deformation $(L{\setminus}\{\bsk\})\cup\{\bsj\}$.}
    \label{table:vertexcyclerelations}
\end{table*}

Induced paths and even holes are \emph{triangle-free}, i.e., they do not contain a clique of three vertices (we define a hole to have length greater than three). As we might expect, the tension between this triangle-free constraint and the claw-free constraint on the entire graph tightly restricts the neighboring relations between these structures and other vertices in the graph.
\begin{lemma}
    \label{lemma:vertexcyclerelations}
    Let $L$ be an induced path or hole in a claw-free graph $G$ and let $\bsj$ be a vertex not in $L$, then all possible neighboring relations between $\bsj$ and $L$ are given in Table~\ref{table:vertexcyclerelations}.
\end{lemma}
Rather than list the cases here, we give their definitions in Table~\ref{table:vertexcyclerelations}. These cases are not necessarily mutually exclusive. For example, cases (a.ii) and (a.iv) coincide for $L$ an induced path of length one.
\begin{proof}
    Any induced subgraph of $L$ must either be bipartite (an even hole, or a disjoint union of induced paths) or an odd hole. Since any bipartite graph of at least five vertices or odd hole of more than five vertices contains an independent set of at least three vertices, $\bsj$ cannot have five or more neighbors in $L$ unless $L$ is a hole of five vertices. This gives case (b.ii). Clearly, there are no claws in $G[\{\bsj\} \cup L]$ if $\bsj$ has no neighbors in $L$, so this gives case (a.i). We consider each additional case according to the number of neighbors to $\bsj$ in $L$.

    Suppose $\bsj$ has exactly one neighbor $\bsk \in L$. If $\bsk$ has two neighbors $\bsu,\bsv \in L$, then $\{\bsk,\bsj,\bsu,\bsv\}$ induces a claw in $G$. Thus, the only possibility is for $\bsk$ to have exactly one neighbor in $L$. This gives case (b.iii), where $L$ is an induced path, and $\bsk$ is an endpoint of $L$. 
    
    More generally, if $\bsk \in L$ is a neighbor to $\bsj$ not in $L$, and $\bsk$ has two neighbors $\bsu,\bsv \in L$, then at least one of $\bsu$ and $\bsv$ must be a neighbor to $\bsj$ as well. If $\bsj$ has exactly two neighbors, $\Gamma_L(\bsj)=\{\bsk_0,\bsk_1\} \subseteq L$, this gives case (a.ii) when at least one of $\bsk_0$ and $\bsk_1$ has two neighbors in $L$. If both of $\bsk_0$ and $\bsk_1$ have exactly one neighbor in $L$, this gives case (a.v).
    
    If $\bsj$ has exactly three neighbors $\Gamma_L(\bsj)=\{\bsk_0,\bsk_1,\bsk_2\} \subseteq L$, then at least two of these vertices must be neighboring in $L$. This gives case (b.i) when all of $\bsk_0$, $\bsk_1$, and $\bsk_2$ have two neighbors in $L$. If at least one of $\bsk_0$, $\bsk_1$, or $\bsk_2$ has one neighbor in $L$, then we have case (b.iv).
    
    If $\bsj$ has four neighbors $\Gamma_L(\bsj)=\{\bsk_0,\bsk_1,\bsk_2,\bsk_3\} \subseteq L$, then we have case (a.iv) if $L$ is a hole of length four. In general, any subset of three vertices in $L$ has an independent set of at least two vertices, since $L$ does not contain triangles. Thus, $\Gamma_L(\bsj)$ cannot contain an isolated vertex. Assuming $\Gamma_L(\bsj)$ not to be a hole, there is a pair of vertices in $\Gamma_L(\bsj)$ with only one neighbor in $\Gamma_L(\bsj)$. Without loss of generality, suppose these vertices are $\{\bsk_0,\bsk_3\}$. If these vertices are neighboring, then $\bsk_1$ and $\bsk_2$ must be neighboring, so as not to be isolated. This gives case (a.iii). If $\bsk_0$ and $\bsk_3$ have the same unique neighbor, say $\bsk_1$, then $\bsk_2$ must also neighbor $\bsk_1$ since it again cannot be isolated in $\Gamma_L(\bsj)$, and we have assumed $\bsk_0$ and $\bsk_3$ each only have one neighbor in $\Gamma_L(\bsj)$. However, this gives that $\bsk_1$ has degree three in $L$ (and accordingly $\{\bsk_1,\bsk_0,\bsk_2,\bsk_3\}$ induces a claw in $G$), so we have that $\bsk_0$ and $\bsk_3$ have distinct unique neighbors in $\Gamma_L(\bsj)$ if they are not neighboring. Suppose $\bsk_1$ is the unique neighbor to $\bsk_0$, and $\bsk_2$ is the unique neighbor to $\bsk_3$. This again gives case (a.iii), completing the proof.
\end{proof}

Lemma~\ref{lemma:vertexcyclerelations} will be important in the forthcoming proofs as it allows us to infer neighboring relations based on partial information. We state this explicitly as two useful corollaries.
\begin{corollary}
    \label{corollary:vertexcycleneighboring}
    If $\bsj$ is a neighbor to $\bsk$ in an even hole $C$, then $\bsj$ is also neighboring to a neighbor of $\bsk$ in $C$.
\end{corollary}
\begin{corollary}
    \label{corollary:vertexcycleanother}
    Let $\bsj$ be such that $\bsk_0\md\bsk_1\md\bsk_2\subseteq\Gamma_L(\bsj)$. If $\bsj$ has an additional neighbor $\bsu \in L$, then $\bsu$ must neighbor at least one of $\bsk_0$ or $\bsk_2$.
\end{corollary}

We make the distinction between the case where $h_{\bsj}$ commutes with $h_L$ and the case where $h_{\bsj}$ anticommutes with $h_L$ in Table~\ref{table:vertexcyclerelations}, as the latter is especially important from a physical perspective. Interestingly, there is only one possibility for $h_{\bsj}$ to anticommute with $h_C$ when $C$ is an even hole, which is (b.i) in Table~\ref{table:vertexcyclerelations}. When this case holds, and $L$ is either an induced path or even hole, there is a unique additional induced path or even hole defined as a \emph{rerouting} of $L$ through $\bsj$.

\begin{definition}[Single-vertex deformation]
    \label{definition:singlevertexdeformation}
    Let $L$ be a hole or an induced path, and let $\bsj$ be a vertex not in $L$ with neighborhood $\Gamma_L(\bsj)=\bsu\md\bsk\md\bsv$ as in case (b.i) of Table~\ref{table:vertexcyclerelations}. The \emph{single-vertex deformation} $L'$ of $L$ by $\bsj$ is defined by
    \begin{equation}
        L' = (L{\setminus}\{\bsk\})\cup\{\bsj\}.
    \end{equation}
    The vertex $\bsk$ is called the \emph{clone} to $\bsj$ in $L$, and we denote this relationship by $\bsj\prec_L\bsk$.
\end{definition}

Note that single-vertex deformations are reversible, i.e., if $\bsj\prec_L\bsk$, then $\bsk\prec_{L'}\bsj$. There is a kind of generalization of a deformation that we need for our proof of Result~\ref{result:inducedpath}.
\begin{definition}[Bubble wand, handle, hoop]
    \label{definition:bubblewandhandlehoop}
    Suppose $\bsj$ neighbors a path $P$ as in case (b.iv) with $\Gamma_P(\bsj)=\{\bsj_i\md\bsj_{i+1},\bsj_{\ell}\}$ for $i<\ell-2$, as labeled in Eq.~(\ref{equation:pathlabel}). We define the \textit{bubble wand} graph to be $ B=P\cup\{\bsj\}$ and define the \emph{handle} of the wand as the path $P_i=\{\bsj_k\}_{k=0}^{i}$, with the \emph{hoop} defined as the hole $C=B{\setminus}P_i$.
\end{definition}
Note that, if we were to allow $i=\ell-2$, then we would have $\bsj\prec_P\bsj_{i+1}$. However, we need to formally distinguish between these cases in our proof. We return to this structure in Section~\ref{section:KrylovSubspaces}.

Returning to the topic, we collect all of the holes or induced paths related by sequences of single-vertex deformations into sets called deformation closures.

\begin{definition}[Deformation closure]
    \label{definition:deformationclosure}
    Let $L_0$ be a hole or an induced path. The \emph{deformation closure} $\avg{L_0}$ of $L_0$ is the set such that $L_0$ is in $\avg{L_0}$ and, for any hole or induced path $L$ in $\avg{L_0}$, every single-vertex deformation of $L$ is in $\avg{L_0}$.
\end{definition}
Note that a given hole or induced path cannot belong to more than one deformation closure. If $L$ is in $\avg{L_0}$ and $\avg{L_0'}$, then $L$ and $L_0$ are related by a deformation, and so are $L$ and $L_0'$. Thus, $L_0$ is related to $L_0'$ by the deformation that first takes $L_0$ to $L$ and then from $L$ to $L_0'$. This therefore gives $\avg{L_0}=\avg{L'_0}$. Additionally, all of the holes in a deformation closure have the same length, so the deformation closures partition the holes in the graph such that all of the holes in a given deformation closure have a fixed length. The induced paths in a given deformation closure have the same length and endpoints, so their deformation closures partition them similarly.

The structures of the deformation closures can be complicated, with certain single-vertex deformations either enabled or prevented by other ones. In particular, this happens when a given vertex $\bst$ is neighboring to exactly one of $\{\bss,\bsa\}$ with $\bss\prec_L\bsa$. In this case, we say that $\bst$ is \emph{dependent} on the deformation by $\bss\prec_L\bsa$. We are interested in the instance where $L$ is an even hole, for which we have the following lemma.
\begin{lemma}
    \label{lemma:evenholecloneneighbor}
    Let $C$ be an even hole with $\abs{C}=2k$ and vertex labeling as defined in Eq.~(\ref{equation:evenholelabel}). If $\bss\prec_C\bsa_0$ and a vertex $\bst$ neighbors exactly one element of $\{\bss,\bsa_0\}$, then either 
    \begin{equation}
        \Gamma_{\left(\{\bss\} \cup C\right)}(\bst) =
        \begin{cases} 
            \bsa_{k-1}\md\bsb_0\md\bsu & \text{(i)}, \\
            \bsb_{k-1}\md\bsa_{k-1}\md\bsb_0\md\bsu & \text{(ii)}, \\
            \bsu\md\bsb_1\md\bsa_1 & \text{(iii)}, \\
            \bsu\md\bsb_1\md\bsa_1\md\bsb_2 & \text{(iv)}, 
        \end{cases}
        \notag
    \end{equation}
    where $\bsu=\bss$ or $\bsu=\bsa_0$. If $k=2$, then cases (ii) and (iv) coincide.
\end{lemma}
\begin{proof}
    Without loss of generality, suppose $\bsu=\bss$, and let $C'=(C{\setminus}\bsa_0)\cup\{\bss\}$ be the single-vertex deformation of $C$ by $\bss$. By Corollary~\ref{corollary:vertexcycleneighboring}, $\bst$ must neighbor at least one vertex in $\Gamma_{C'}(\bss)=\{\bsb_0,\bsb_1\}$. Once again without loss of generality, suppose $\bst$ neighbors $\bsb_0$. Again by Corollary~\ref{corollary:vertexcycleneighboring}, $\bst$ must neighbor at least one vertex in $\Gamma_{C}(\bsb_0) = \{\bsa_{k-1},\bsa_0\}$, so it must neighbor $\bsa_{k-1}$. This gives case (i). If $\bst$ has an additional neighbor in $C'$, then by Corollary~\ref{corollary:vertexcycleanother}, it must be either $\bsb_{k-1}$ or $\bsb_1$. If $k=2$, then $\bsb_{k-1}=\bsb_1$, cases (ii) and (iv) coincide, and this gives that case. If $k>2$, then $\bst$ cannot neighbor $\bsb_1$, as $\{\bsb_1,\bst,\bsa_0,\bsa_1\}$ induces a claw, and $\bst$ cannot have any additional neighbors in $C'$ in this case. Thus, $\bst$ must neighbor $\bsb_{k-1}$, and this gives case (ii). A similar argument applies for the case where $\bsb_1$ is in $\Gamma_{C}(\bst)$ and the case where $\bsu=\bsa_0$. This completes the proof.
\end{proof}

We have the following useful corollaries.
\begin{corollary}
    \label{corollary:evenholecloneexactlyoneneighbor}
    If $\bst$ neighbors exactly one of $\{\bss,\bsa_0\}$ in the setting of Lemma~\ref{lemma:evenholecloneneighbor} with $k>2$, it neighbors exactly one of $\{\bsb_0,\bsb_1\}$.
\end{corollary}
\begin{corollary}
    \label{corollary:evenholeclonemanyneighbors}
    If $\bst$ neighbors at least one of $\{\bss,\bsa_0\}$ and both of $\{\bsb_0,\bsb_1\}$ in the setting of Lemma~\ref{lemma:evenholecloneneighbor} with $k > 2$, it neighbors both of $\{\bss,\bsa_0\}$.
\end{corollary}
These results allow us to collect statements about individual even holes into statements about their deformation closures. The following lemma is a simple example.
\begin{lemma}
    \label{lemma:evenholeneighbor}
    If $\bst$ is in $\Gamma[C_0]$ for an even hole $C_0$, then $\bst$ is in $\Gamma[C]$ for any even hole $C$ in $\avg{C_0}$.
\end{lemma}
\begin{proof}
    It is sufficient to prove that $\bst$ is in $\Gamma[C]$ for a single-vertex deformation $C$ of $C_0$. Thus, let
    \begin{equation}
        C = (C_0{\setminus}\{\bsa_0\})\cup\{\bss\},
    \end{equation}
    with $\bss\prec_{C_0}\bsa_0$. Suppose that $\bst$ is not in $\Gamma[C]$, then $\bst$ is not in $C \cup C_0$ since $C \cup C_0$ is a subset of $\Gamma[C]$, and $\Gamma_{C \cup C_0}(\bst)=\{\bsa_0\}$ since $\bst$ is in $\Gamma[C_0]$. However, this is a contradiction to Lemma~\ref{lemma:evenholecloneneighbor}. Therefore, if $\bst$ is in $\Gamma[C_0]$, then $\bst$ is in $\Gamma[C]$ for any single-vertex deformation $C$ of $C_0$.
    For an even hole $C\in\avg{C_0}$ that is not necessarily a single-vertex deformation of $C_0$, applying this result iteratively to the sequence of deformations from $C_0$ to $C$ completes the proof.
\end{proof}
Lemma~\ref{lemma:evenholeneighbor} shows that $\Gamma[C]=\Gamma[C_0]$ for any $C\in\avg{C_0}$. Conversely, if $\bst$ is not in $\Gamma[C_0]$, then $\bst$ is not in $\Gamma[C]$ for any $C\in\avg{C_0}$. Recall from Section~\ref{section:FrustrationGraphs} that two even holes $C$ and $C'$ are said to be \emph{compatible} if $\bsj$ is not in $\Gamma[C']$ for every $\bsj\in\Gamma[C]$. We thus have the following corollary.
\begin{corollary}
    If an even hole $C$ is compatible with an even hole $C_0$, then $C$ is compatible with every even hole $C'\in\avg{C_0}$.
\end{corollary}

\begin{figure}[ht!]
    \centering
    \includegraphics[width=0.5\textwidth]{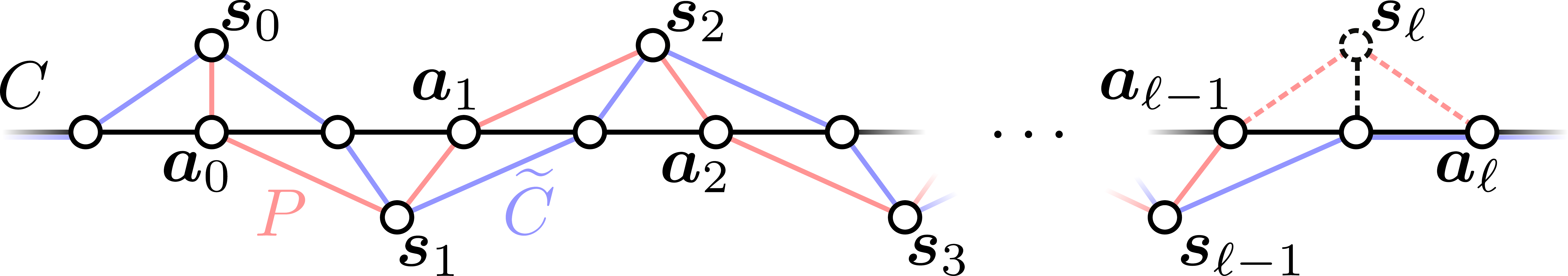}
    \caption{The deformation along a path $P$ when $S$ is an independent set, defined as the path component of $G[S{\oplus}C_a]$ with endpoint $\bss_0$, where $C_a$ is the coloring class to which the clone of $\bss_0$ in $C$ belongs. If $\bss_{\ell}$ is present, then we cannot deform by $\bss_{\ell}$, and we say that $\bss_0$ is tethered to $\bss_{\ell}$. Otherwise, the path $P$ has odd length.}
    \label{figure:deformationsequence}
\end{figure}

Lemma~\ref{lemma:evenholecloneneighbor} allows us to consider sequences of deformations. An important case is that where we deform an even hole $C$ by vertices in an independent set $S=\{\bss_j\}_{j=0}^{\ell-1}$, labeled such that each single-vertex deformation is performed successively in $j$ as shown in Fig.~\ref{figure:deformationsequence}. Denote $C^{(j)}$ as the hole following the deformation by $\bss_j$, with $C^{(-1)}=C$. We have that $\bss_{j+1}$ is not neighboring to $\bss_j$ but is neighboring to its clone $\bsa_j \in C^{(j-1)}$. If we can deform $C^{(j)}$ by $\bss_{j+1}$, then $\bss_{j+1}\prec_{C^{(j)}}\bsa_{j+1}$ by Lemma~\ref{lemma:evenholecloneneighbor}.

There is a unique path associated to the sequence of deformations shown in Fig.~\ref{figure:deformationsequence}. This is the path component $P \subseteq G[S{\oplus}C_a]$ with endpoint $\bss_0$, where the independent set $C_a$ is operationally defined as the coloring class of $C$ to which the clone to $\bss_0$ belongs. (Since $\bss_0$ has only one neighbor in $C_a$, it is the endpoint of a path component of $G[S{\oplus}C_a]$.) We call this sequential deformation by vertices in $P \cap S$ a deformation \emph{along} the path $P$, and we call $P$ the deformation path. We call $\bss_0$ the \emph{initializing vertex} of the deformation along $P$ (or we say $\bss_0$ initializes this deformation). We can continue the deformation until we reach a vertex $\bss_{\ell}$ with only three neighbors in $C$, as shown in Fig.~\ref{figure:deformationsequence}, and we cannot deform by $\bss_{\ell}$ if it is present. Letting $P'$ be the path component of $G[S{\oplus}C_b]$ with $\bss_{\ell}$ as an endpoint, we could similarly deform $C$ along $P'$ until we reach the vertex $\bss_{0}$, so $\bss_{0}$ and $\bss_{\ell}$ can be uniquely associated this way. We say that $\bss_0$ and $\bss_{\ell}$ are \emph{tethered} with respect to $C$. If $\bss_0$ is untethered with respect to $C$, i.e., $\bss_{\ell}$ is not present, then the path $P$ has odd length $2\ell-1$, so is given by
\begin{equation}
    P = \bss_0\md\bsa_0\md\dots\bss_{\ell-1}\md\bsa_{\ell-1}.
\end{equation}
Note that for $\abs{C}=2k$, we have $\ell \leq k$. Otherwise, this contradicts the requirement that $P$ is a path. We can operationally interpret $\ell$ as the number of vertex-clone pairs in the deformation by $P$, and the collection of these pairs corresponds to the unique perfect matching in $P$.

\subsection{Reconfiguration Problems}
\label{subsection:ReconfigurationProblems}

Deformations for holes and induced paths are the subject of a particular \emph{reconfiguration problem} for claw-free graphs. A reconfiguration problem considers whether a graph structure, such as an independent set or shortest path, can be reached from another one by a sequence of allowed moves. We consider the following important reconfiguration move for independent sets.
\begin{definition}[Token sliding~\cite{bonsma2014reconfiguring}]
    Given independent sets $S$ and $S'$ in a claw-free graph $G$, $S$ and $S'$ are related by a \emph{token slide} if there is a pair of neighboring vertices $\bsu$ and $\bsv$ with $\bsu \in S{\setminus}S'$ and $\bsv \in S'{\setminus}S$ such that
    \begin{equation}
        S' = (S{\setminus}\bsu)\cup\{\bsv\}.
    \end{equation} 
\end{definition}
That is, we consider a set of \emph{tokens} placed on the independent set $S$ and ask whether we can obtain $S'$ from $S$ by sliding a token along an edge of $G$. Note that if $S'$ is reachable from $S$ by a token sliding move, then $S$ is similarly reachable from $S'$. We say that $S\leftrightarrow_{\text{TS}}S'$ if $S$ and $S'$ are related by a sequence of token-sliding moves.

Reachability is described by the solution graph $\tsk$, whose vertices correspond to $k$-vertex independent sets in $G$ and are neighboring if they are related by a token slide. Two $k$-vertex independent sets $S$ and $S'$ satisfy $S\leftrightarrow_{\text{TS}}S'$ if $S$ and $S'$ are in the same connected component of $\tsk$. Let $\Xi_k$ be the set of connected components of $\tsk$.

We can consider a connected component of $\tsk$ as a corresponding closure of independent sets, and define the following conserved charges of Theorem~\ref{theorem:conservedcharges} as sums over the appropriate closures.
\begin{definition}[Token-sliding charges]
    \label{definion:tokenslidingcharges}
    The token-sliding charge $\qkg{k,\mu}{G}$ is defined as
    \begin{equation}
        \qkg{k,\mu}{G} \coloneqq \sum_{S\in\mu}h_S, \label{eq:tokenindependence}
    \end{equation}
    where $\mu\in\Xi_k$ is a connected component of $\tsk$. These are related to the independent set charges from Def.~\ref{definition:independentsetcharges} via 
    \begin{equation}
        \sum_{\mu\in\Xi_k}\qkg{k,\mu}{G} = \qkg{k}{G}.
    \end{equation}
    That is, the independent set charge is a sum over the connected components of $\tsk$. If $k=0$, we take the convention that there is only a single component $\mu$ of $\tsk$ with $\qkg{0,\mu}{G}=\qkg{0}{G}=I$.
\end{definition}

Note that if $G$ is itself not connected, then neither is $\tsk$, and we have a token-sliding charge for each component of $\tsk$. However, even when $G$ is connected, $\tsk$ may not be. The case where $G$ is an even hole and $k=\alpha(G)$ is a clear example, since any token slide will take a coloring class of $G$ to a set which is not independent. The $\qkg{k,\mu}{G}$ can thus be thought of as a fine graining of the independent set charges to account for the case where $G$ is not connected or contains a certain kind of even hole. Ref.~\cite{bonsma2014reconfiguring} gives necessary and sufficient conditions for $\tsk$ to be connected. In fact, it is only possible for $\tsk$ to be disconnected with $G$ connected when $k=\alpha(G)$ and when $G$ contains an even hole. This implies Ref.~\cite[Lemma 1]{elman2021free} is already the strongest possible when $H$ has a connected frustration graph. However, when $G$ contains even holes, we may have additional token-sliding charges.

We see that a single-vertex deformation is a special case of a token sliding move on a coloring class of an even hole that preserves the even hole. It is thus natural to define a corresponding operator.
\begin{definition}[Generalized cycle symmetries]
    \label{definition:generalizedcyclesymmetries}
    The generalized cycle symmetries are defined by
    \begin{equation}
        \jkg{C_0}{G} = \sum_{C\in\avg{C_0}}h_C.
    \end{equation}
\end{definition}
Single-vertex deformations and token sliding are both special cases of reconfiguration for regular induced subgraphs. Specifically, single-vertex deformations of even holes correspond to reconfigurations of connected 2-regular induced subgraphs. Token sliding moves on independent sets correspond to reconfigurations of 0-regular induced subgraphs (see, e.g., Ref.~\cite{eto2022reconfiguration} for more details). 

\subsection{Induced Path Trees}
\label{subsection:InducedPathTrees}

Simplicial, claw-free graphs have a hereditary structure. That is, $G[U]$ is an SCF graph for all vertex subsets $U \subseteq V$ of an SCF graph $G$. For a given simplicial clique $K_s$ in $G$, we define
\begin{equation}
    K_{\bsj} = (\Gamma(\bsj){\setminus}K_s)\cup\{\bsj\},
\end{equation}
for all $\bsj \in K_s$. By the definition of a simplicial clique, $K_{\bsj}$ is a clique for all $\bsj \in K_s$. Furthermore, $\Gamma(\bsj){\setminus}K_s$ is itself a simplicial clique in $G{\setminus}K_s$~\cite{chudnovsky2007roots}.

\begin{figure*}[ht!]
    \begin{subfigure}{0.3\textwidth}
         \centering
         \includegraphics[width=0.9\columnwidth]{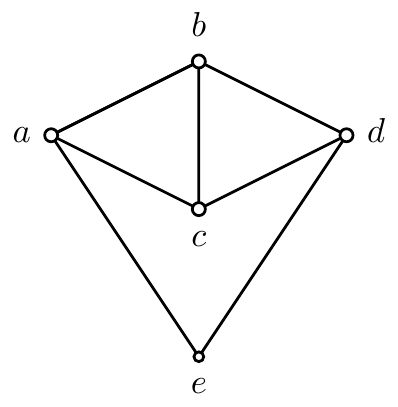}
         \caption{}
         \label{figure:parentgraph}
     \end{subfigure}%
     \begin{subfigure}{0.3\textwidth}
         \centering
         \includegraphics[width=0.9\columnwidth]{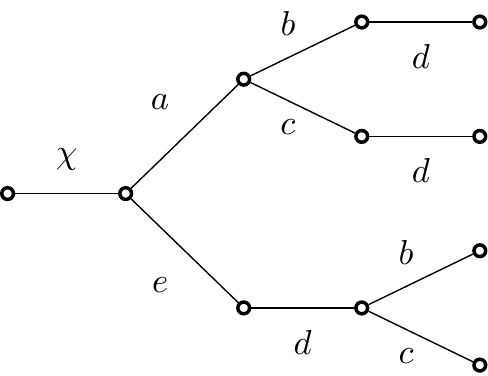}
         \caption{}
         \label{figure:inducedpathtree}
     \end{subfigure}%
     \begin{subfigure}{0.3\textwidth}
         \centering
         \includegraphics[width=0.9\columnwidth]{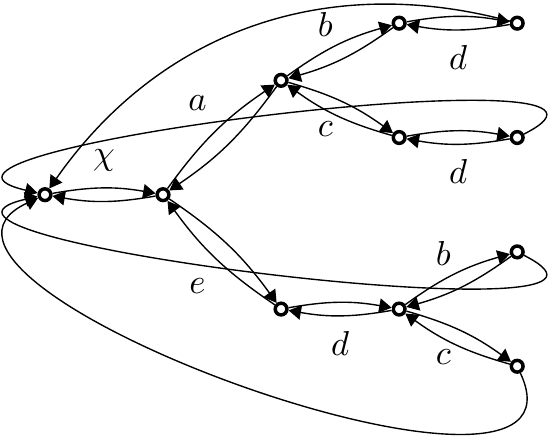}
         \caption{}
         \label{figure:directedhoppinggraph}
     \end{subfigure}
    \caption{The induced path tree: (a) a simplicial, claw-free graph with an even hole; (b) an induced-path tree from a simplicial vertex attached to the simplicial clique $K_s=\{\bsa,\bs{e}\}$ in the graph; (c) the hopping graph of a walk induced by the nested commutators of $\chi$ with the Hamiltonian $H$ for which the graph is the frustration graph.}
    \label{figure:inducedpathtreegraphs}
\end{figure*}

With this in mind, the \emph{induced path tree} is defined as follows.
\begin{definition}[Induced path tree with respect to $\bsj$~\cite{leake2019generalizations}]
    \label{definition:inducedpathtreevertex}
    For $\bsj \in V$, the \emph{induced path tree} $\idt{\bsj}{G}$, of $G$ with respect to $\bsj$ is defined recursively. If $G$ is a tree, then $\idt{\bsj}{G}=G$, and we say that $\bsj$ is the root of $\idt{\bsj}{G}$. Otherwise, we consider the forest of disjoint trees $\idt{\bsk}{(G{\setminus}\Gamma[\bsj])\cup\{\bsk\}}$ with root $\bsk$ for each $\bsk\in\Gamma(\bsj)$. We then define $\idt{\bsj}{G}$ by appending a root vertex corresponding to $\bsj$ and connecting it to the roots of each of these trees.
\end{definition}
As in Ref.~\cite{leake2019generalizations}, we also define an induced path tree with respect to a clique $K$.
\begin{definition}[Induced path tree with respect to $K$~\cite{leake2019generalizations}]
    \label{definition:inducedpathtreeclique}
    Let $K$ be a clique. The induced path tree $\idt{K}{G}$ of $G$ with respect to $K$ is defined as follows. Let $\gs$ be the graph formed by attaching a new vertex $\js$ to $G$ with the property that $(\js,\bsk) \in E(\gs)$ for all $\bsk \in K$. Then $\idt{K}{G}=\idt{\js}{\gs}$.
\end{definition}
Each vertex in $\idt{K}{G}$ can be labeled by the induced path in $G$ given by the sequence of subtree roots in the path from that vertex to $\js$. Note that $\gs$ from Def.~\ref{definition:inducedpathtreeclique} is also simplicial, claw-free when $K=K_s$ is a simplicial clique. Clearly $\gs$ is simplicial since $\js$ is a simplicial vertex. Suppose that $\gs$ contains a claw, then that claw must contain $\js$ since $G$ is claw-free. However, $\Gamma(\js)=K_s$, so there must be some vertex $\bsk$ in $\Gamma(\js)$ that neighbors an independent set of order at least three. Suppose that this is the case, then the set $\Gamma[\bsk]{\setminus}\Gamma[\js]$ must contain a pair of non-neighboring vertices, but this contradicts our assumption that $K_s$ is simplicial. Therefore, we have that $\gs$ is a simplicial, claw-free graph as well. In particular, we shall use the fact that all of the neighboring relations in Table~\ref{table:vertexcyclerelations} hold for induced paths containing $\js$ as an endpoint in $\gs$.

Fig.~\ref{figure:parentgraph} shows a small example of a simplicial, claw-free graph. We have identified the simplicial clique $K_s=\{a,e\}$ and constructed the induced path tree $\idt{K_s}{G}$ in Fig.~\ref{figure:inducedpathtree}. From this, we construct the \emph{directed hopping graph} $\Lambda_{K_s}(G)$, shown in Fig.~\ref{figure:directedhoppinggraph}, and defined as follows.

\begin{definition}[Directed hopping graph]
    \label{definition:directedhoppinggraph}
    The graph $\Lambda_{K_s}(G)$, is related to the graph $\idt{K_s}{G}$ by replacing each edge in $\idt{K_s}{G}$ with a pair of directed arcs and adding a set of arcs corresponding to even holes in $G$. Specifically, there is such an arc from $P$ to $P'$ in $\Lambda_{K_s}(G)$ if there is a vertex $\bsk$ not in $P$ such that $P\cup\{\bsk\}$ is a bubble wand (see Def.~\ref{definition:bubblewandhandlehoop}) with $P'$ the handle, and the hoop is an even hole.
\end{definition}
We shall always consider these graphs with respect to a fixed simplicial clique of $G$, so we drop the explicit $K_s$ dependence in our notation as $\idt{}{\gs}=\idt{K_s}{\gs}$.

\section{Conserved Quantities}
\label{section:ConservedQuantities}

A crucial component to the proof of Theorem~\ref{theorem:freefermionsolution} is the identification of the conserved charges via their graphical structures. In this section we prove the following theorem.
\begin{theorem}[Conserved Charges]
    \label{theorem:conservedcharges}
    Let $H$ be an Hamiltonian with claw-free frustration graph $G$. The operators $\{\qkg{s,\mu}{G}\}_{(s,\mu)}$ and $\{\jkg{C_0}{G}\}_{\avg{C_0}}$ satisfy
    \begin{align}
        [\qkg{s,\mu}{G},\qkg{t,\nu}{G}] &= 0, \tag{i} \\
        [\jkg{C_0}{G},\qkg{s,\mu}{G}] &= 0, \tag{ii} \\
        [\jkg{C_0}{G},\jkg{C'_0}{G}] &= 0. \tag{iii}
    \end{align}
    In particular, these operators are conserved charges of the Hamiltonian.
\end{theorem}
Theorem~\ref{theorem:conservedcharges} alone gives further evidence for the idea proposed in Ref.~\cite{elman2021free} that Hamiltonians with claw-free frustration graphs are integrable, despite possibly not having free-fermion solutions. Our proof strategy is to expand each operator as a sum in the commutator. For each non-vanishing contribution to the sum, we shall show there is a unique additional term to cancel it. This is illustrated in our proof of Theorem~\ref{theorem:conservedcharges}~(i).

\subsection{Independent Set Charges}
\label{subsection:IndependentSetCharges}

In this section we prove Theorem~\ref{theorem:conservedcharges}~(i). 
\begin{proof}[Proof of Theorem~\ref*{theorem:conservedcharges}~$(i)$]
    By expanding, we have
    \begin{equation}
        [\qkg{s,\mu}{G},\qkg{t,\nu}{G}] = \sum_{\substack{S\in\mu \\ S'\in\nu}}[h_S,h_{S'}]. \label{equation:conservedchargescommutator}
    \end{equation}
    Using the fact that $S$ and $S'$ are independent sets, we have $\Delta_{S'}(\bsj)=\Delta_{S'{\setminus}S}(\bsj)$ for $\bsj \in S$, and $\Delta_{S'}(\bsj)=0$ for $\bsj \in S \cap S'$. This gives
    \begin{equation}
        \sum_{\bsj \in S}\Delta_{S'}(\bsj) = \sum_{\bsj \in S{\setminus}S'}\Delta_{S'{\setminus}S}(\bsj) = \abs{E[S{\oplus}S']},
    \end{equation}
    since $S{\oplus}S'=(S{\setminus}S')\cup(S'{\setminus}S)$ and $G[S{\oplus}S']$ is bipartite with coloring classes $S{\setminus}S'$ and $S'{\setminus}S$. Then, by Eq.~(\ref{equation:degreephase}),
    \begin{align}
        h_Sh_{S'} &= (-1)^{\sum_{\bsj \in S}\Delta_{S'}(\bsj)}h_{S'}h_S \\
        &= (-1)^{\abs{E[S{\oplus}S']}}h_{S'}h_S.
    \end{align}
    
    Now consider a fixed term in Eq.~(\ref{equation:conservedchargescommutator}) indexed by $(S,S')$ with $\abs{E[S{\oplus}S']}$ odd. By Lemma~\ref{lemma:clawfreesymmetricdifference}, there there is at least one (and an odd number in general) path component of $G[S{\oplus}S']$ with odd length. Choose such a path $P=\bss'_0\md\bss_0\md\dots\md\bss'_{\ell-1}\md\bss_{\ell-1}$ with $\{\bss_j\}_{j=0}^{\ell-1} = S \cap P$ and $\{\bss'_j\}_{j=0}^{\ell-1} = S' \cap P$, so that $P$ has length $2\ell-1$. Let $\wt{S}=S{\oplus}P$ and $\wt{S}'=S'{\oplus}P$. Note that $\wt{S}$ can be obtained from $S$ by successively sliding $\bss_j$ to $\bss'_j$ for $j\in\{0,\dots,\ell-1\}$. Thus, $S\leftrightarrow_{\text{TS}}\wt{S}$ and $\wt{S}$ is in $\mu$. Similarly, $S'\leftrightarrow_{\text{TS}}\wt{S}'$ and $\wt{S}'$ is in $\nu$. Since $S{\oplus}S'=\wt{S}{\oplus}\wt{S}'$ and $G[\wt{S}{\oplus}\wt{S}']$ has odd size, there is an additional term in Eq.~(\ref{equation:conservedchargescommutator}) indexed by $(\wt{S},\wt{S}')$. Then, we have
    \begin{align}
        &[h_S,h_{S'}]+[h_{\wt{S}},h_{\wt{S}'}] = 2\left(h_{S}h_{S'} + h_{\wt{S}} h_{\wt{S}'}\right) \notag \\
        &\quad= 2h_{S{\setminus}P}\left(h_{S \cap P }h_{S' \cap P}+h_{S' \cap P} h_{S \cap P}\right)h_{S'{\setminus}P} \\
        &\quad= 0.
    \end{align}
    For a collection of such terms in Eq.~(\ref{equation:conservedchargescommutator}) with fixed symmetric difference $S{\oplus}S'$, we fix a path component of $G[S{\oplus}S']$ by which corresponding terms are paired. These terms cancel pairwise, completing the proof.
\end{proof}

\subsection{Generalized Cycle Symmetries and Independent Set Charges}
\label{subsection:GeneralizedCycleSymmetriesIndependentSetCharges}

We now prove Theorem~\ref{theorem:conservedcharges}~(ii). 
\begin{proof}[Proof of Theorem~\ref*{theorem:conservedcharges}~$(ii)$]
    By expanding, we have
    \begin{align}
        [\jkg{C_0}{G},\qkg{s,\mu}{G}] = \sum_{\substack{C\in\avg{C_0} \\ S \in \mu}}[h_C,h_S]. \label{equation:conservedchargecyclecommutator}
    \end{align}
    Assume $h_C$ and $h_S$ anticommute. There is at least one (and an odd number in general) vertex $\bs{s}_0$ in $S$ such that $h_{\bs{s}_0}$ and $h_{C}$ anticommute, as shown in Section~\ref{subsection:NeighboringRelations}, each such vertex initializes a deformation path, but may be tethered. If every such initializing vertex is tethered, then they can be uniquely paired, which contradicts the assumption that $h_S$ and $h_C$ anticommute. Thus, there is at least one untethered path $P=\bss_0\md\bsa_0\md\dots\md\bss_{\ell-1}\md\bsa_{\ell-1}$ with $\bss_0\prec_C\bsa_0$. This is shown in Fig.~\ref{figure:deformationsequence}. Let $\wt{C}=C{\oplus}P$ and $\wt{S}=S{\oplus}P$. We have shown in Section~\ref{subsection:NeighboringRelations} that $\wt{C}$ is in $\avg{C_0}$, and it follows from the proof of Theorem~\ref{theorem:conservedcharges}~(i) that $\wt{S}$ is in $\mu$.

    We have that $h_P$ anticommutes with both $h_S$ and $h_C$, since only $h_{\bs{s}_0}$ anticommutes with $h_C$ and only $h_{\bs{a}_{\ell-1}}$ and $h_S$ anticommute. Thus,
    \begin{align}
        \scomm{h_{\wt{S}}}{h_{\wt{C}}} &= \scomm{h_S}{h_C}\scomm{h_S}{h_P}\scomm{h_P}{h_C}\scomm{h_P}{h_P} \\
        &= -1,
    \end{align}
    by Eq.~(\ref{equation:symmetricdifference}) and the assumption that $h_S$ and $h_C$ anticommute. Then, we have
    \begin{align}
        &[h_C, h_S]+[h_{\widetilde{C}},h_{\widetilde{S}}] = 2\left(h_Ch_S+h_{\widetilde{C}}h_{\widetilde{S}}\right) \notag \\
        &\quad= 2h_{C{\setminus}P}\left(h_{C \cap P}h_{S \cap P}+h_{S \cap P}h_{C \cap P}\right)h_{S{\setminus}P} \\
        &\quad= 0,
    \end{align}
    where the last line follows since $S \cap P$ and $C \cap P$ are the coloring classes of a path of odd length. For a collection of terms in Eq.~(\ref{equation:conservedchargecyclecommutator}) related by a fixed set of untethered paths, we fix a path by which corresponding terms are paired. These terms cancel pairwise, completing the proof.
\end{proof}

\subsection{Generalized Cycle Symmetries}
\label{subsection:GeneralizedCycleSymmetries}

Finally, we prove Theorem~\ref{theorem:conservedcharges}~(iii). By expanding, we have
\begin{equation}
    [\jkg{C_0}{G},\jkg{C'_0}{G}] = \sum_{\substack{C\in\avg{C_0} \\ C'\in\avg{C'_0}}}[h_C,h_{C'}]. \label{equation:conservedcyclecommutator}
\end{equation}
Fix a pair $C$ and $C'$ such that $h_C$ and $h_{C'}$ anticommute. We label the vertices according to Eq.~(\ref{equation:evenholelabel}) as
\begin{align}
    C &= \bsb_0\md\bsa_0\md\bsb_1\md\bsa_1\dots\md\bsb_{k-1}\md\bsa_{k-1}\md\bsb_{0}, \\
    C' &= \bsd_0\md\bsc_0\md\bsd_1\md\bsc_1\dots\md\bsd_{k'-1}\md\bsc_{k'-1}\md\bsd_{0}.
\end{align}
Then $C = C_a \cup C_b$ and $C' = C'_c \cup C_d'$, where $C_a=\{\bsa_j\}_{j=0}^{k-1}$ and $C_b=\{\bsb_j\}_{j=0}^{k-1}$ are the coloring classes of $C$, and $C'_c=\{\bsc_j\}_{j=0}^{k'-1}$ and $C'_d=\{\bsd_j\}_{j=0}^{k'-1}$ are the coloring classes of $C'$. We shall identify a term corresponding to a pair of even holes $\wt{C}$ and $\wt{C}'$ whose contribution cancels the $(C,C')$ term in Eq.~(\ref{equation:conservedcyclecommutator}). We achieve this by deforming $C$ to $\wt{C}$ by a sequence of vertices in $C'$ such that there exists a corresponding reverse deformation from $C'$ to $\wt{C}'$ by vertices in $C$.
This deformation consists of an ordered sequence of entire deformation paths as defined in Section~\ref{section:ClawFreeGraphs}. As this proof is considerably more complicated than the proofs of Theorem~\ref{theorem:conservedcharges}~(i) and Theorem~\ref{theorem:conservedcharges}~(ii), we divide it into several subsections and motivate the proof with an illustrative example.

\subsubsection{Palindromic Path Example}
\label{subsubsection:PalindromicPathExample}

We now motive our proof with an example. We label $C$ and $C'$ such that $\bsc_0\prec_C\bsa_0$ and $\bsc_0$ is the untethered initializing vertex of a deformation path $P$ in $G[C_a{\oplus}C_c']$,
\begin{equation}
    P = \bsc_0\md\bsa_0\md\dots\md\bsc^{(j)}\md\bsa_j\md\dots\md\bsc^{(\ell-1)}\md\bsa_{\ell-1}.
\end{equation}
Since $h_C$ and $h_{C'}$ anticommute, at least one such vertex $\bsc_0$ is guaranteed to exist. We note two subtleties here. First, while we are guaranteed that
\begin{equation}
    \Gamma_C(\bsc^{(j)}) = \bsa_{j-1}\md\bsb_{j}\md\bsa_{j}\md\bsb_{j+1},
\end{equation}
for $j>0$, if $\ell>1$, we cannot assume that $\bsc^{(j)}=\bsc_j$ for $j>0$ (we take $\bsc^{(0)}=\bsc_0$). That is, we do not assume that $\bsc^{(j)}$ and $\bsc^{(j+1)}$ share a neighbor in $C_d'$ for any $j\in\{0,\dots,\ell-2\}$. The second subtlety is that, while $\bsc^{(j)}$ is not in $C$ for all $j\in\{0,\dots,\ell-1\}$, $\bsa_j$ may be in $C \cap C'$ for some $j\in\{0,\dots,\ell-2\}$. $\bsa_{\ell-1}$ cannot be in $C'$, since then $\bsa_{\ell-1}$ would have two neighbors in $C'_c$, which contradicts the assumption that $\bsa_{\ell-1}$ is the endpoint of $P$.

Up to this point, our description has been completely general. We now restrict to the special case in which $\bsa_j$ neighbors both neighbors to $\bsc^{(j)}$ in $C'_d$ for all ${j\in\{0,\dots,\ell-1\}}$. In this case, $\bsa_{\ell-1}\prec_{C'}\bsc^{(\ell-1)}$ by Lemma~\ref{lemma:vertexcyclerelations} and the assumption that $\bsc^{(\ell-1)}$ is the only neighbor to $\bsa_{\ell-1}$ in $C'_c$. By considering Fig.~\ref{figure:deformationsequence}, we see that $P$ is a deformation path for $C'$ as well, since we have assumed $\bsa_j$ has four neighbors in $C'$ for all $j<\ell-1$. Indeed, these two subtleties do not apply in this case. We label the vertices such that $\bsc^{(j)}=\bsc_j$ with
\begin{equation}
    \Gamma_{C'}(\bsa_j) = \bsd_j\md\bsc_j\md\bsd_{j+1}\md\bsc_{j+1},
\end{equation}
for $j<\ell-1$, and
\begin{equation}
    \Gamma_{C'}(\bsa_{\ell-1}) = \bsd_{\ell-1}\md\bsc_{\ell-1}\md\bsd_{\ell}.
\end{equation}
We have that $\bsa_j$ is not in $C'$ for all $j\in\{0,\dots,\ell-1\}$. We refer to such a path as \emph{palindromic}. Let $\wt{C}=C{\oplus}P\in\avg{C_0}$ and $\wt{C}'=C'{\oplus}P\in\avg{C'_0}$. Further, let $\wt{C}_a=C_a{\oplus}P$ and $\wt{C}'_c=C'_c{\oplus}P$. We have
\begin{align}
    &[h_C,h_{C'}]+[h_{\wt{C}},h_{\wt{C}'}]=2(h_Ch_{C'}+h_{\wt{C}}h_{\wt{C}'}) \notag \\
    &\quad= 2h_{C{\setminus}P}(h_{C \cap P}h_{C' \cap P}+h_{C' \cap P}h_{C \cap P})h_{C'{\setminus}P} \\
    &\quad= 0,
\end{align}
where we rearranged the factors $h_{C_a}$, $h_{\wt{C}_a}$, $h_{C'_c}$, and $h_{\wt{C}'_c}$ to the interiors of the respective products (recall our convention that $h_C=h_{C_a}h_{C_b}$) and used the fact that $P$ is a path in $G[C_a{\oplus}C_c']$. Therefore, Eq.~(\ref{equation:conservedcyclecommutator}) holds in this case. The difficulty in the general case arises if there exists a vertex $\bsu$ in $C'_d$ that neighbors $\bsc^{(j)}$ and not $\bsa_j$. In this case, multiple deformation paths will be required to find a solution, and we generalize the parts of the proof accordingly.

\subsubsection{Definitions and Proof Structure}
\label{subsubsection:DefinitionsProofStructure}

We now introduce some definitions concerning deformations of multiple paths and outline our proof strategy.

\begin{definition}[Fixed-pairing-type deformation]
    \label{definition:fixedpairingtypedeformation}
    A fixed-pairing-type deformation $\overrightarrow{\mco}$ of an even hole $C$ by vertices in $C'$ is a sequence $\left(P^{(r)}\right)_{r=0}^m$ of induced paths of odd length in $G[C \cup C']$. Each path $P^{(r)}=\bsj_0^{(r)}\md\bsk_0^{(r)}\md\dots\md\bsj_{\ell_r-1}^{(r)}\md\bsk_{\ell_r-1}^{(r)}$ is a component of $G[C_j{\oplus}C'_{\sigma(j)}]$ where $j\in\{a,b\}$ and $\sigma(j)\in\{c,d\}$ labels a unique coloring class of $C'$ associated to $C_j$ by $\overrightarrow{\mco}$. The vertices of $P^{(r)}$ are labeled such that $\bsj_{s-1}^{(r)}$ is in $C'_{\sigma(j)}{\setminus}C_{j}$ and $\bsk_{s-1}^{(r)}$ is in $C_j{\setminus}C'_{\sigma(j)}$ for all $s\in\{1,\dots,\ell_r\}$ and all $r\in\{0,\dots,m\}$.
    
    The deformation $\overrightarrow{\mco}$ is $(a,c)$ pairing if $\sigma(a)=c$ and $\sigma(b)=d$ for all $r$. Similarly, $\overrightarrow{\mco}$ is $(a,d)$ pairing if $\sigma(a)=d$ and $\sigma(b)=c$ for all $r$. The pairing type of $\overrightarrow{\mco}$ is thus specified by $\sigma$. We additionally describe a path component $P$ of $G[C_j{\oplus}C'_{\sigma(j)}]$ or vertex pair $\{\bsj,\bsk\}$ with $\bsj \in C'_{\sigma(j)}$ and $\bsk \in C_j$ as having the pairing type specified by $\sigma$. We take $P^{(r)}_{2s-1}=\bsj_0^{(r)}\md\dots\md\bsk_{s-1}^{(r)}$ and let
    \begin{equation}
        C^{(r,s)} = C{\oplus}\left(\bigoplus_{g=0}^{r-1}P^{(g)}\right){\oplus}P^{(r)}_{2s-1},
    \end{equation}
    with $C^{(r)}=C^{(r,\ell_r)}$ for all $r$ and $s$. By convention, we take $C^{(r,0)}=C^{(r-1)}$ and $C^{(-1)}=C$.
    
    The deformation $\overrightarrow{\mco}$ is such that $\bsj_{s-1}^{(r)}\prec_{C^{(r,s-1)}}\bsk_{s-1}^{(r)}$ for all $r$ and $s$. Thus, $C^{(r,s)}$ is in $\avg{C}$ and $P^{(r)}$ and $C^{(r-1)}$ satisfy the relationship shown in Fig.~\ref{figure:deformationsequence} for all $r$ and $s$.
    
    We let $\mco=\bigcup_{r=0}^m\{P^{(r)}\}$ denote the set of paths in $\overrightarrow{\mco}$ and we let $\partial\mco=\bigcup_{r=0}^mP^{(r)}$ denote the set of vertices involved in the deformation by $\overrightarrow{\mco}$.
\end{definition}

Our motivation for this definition comes from the palindromic path example. Deforming along a path $P$ in $C \cup C'$ as shown in Fig.~\ref{figure:deformationsequence} gives that the vertices in $P \cap C$ are in a fixed coloring class of $C$, regardless of whether the independent set $P{\setminus}C$ is a subset of a coloring class of $C'$. However, we require that $P{\setminus}C$ is a subset of a coloring class of $C'$ to apply the deformation in reverse.

We now define several important vertex subsets relative to a fixed-pairing-type deformation.
\begin{definition}(Anticommuting and dependent subsets)
    \label{definition:anticommutingdependentsubsets}
    Let $\overrightarrow{\mco}$ be a fixed-pairing-type deformation of $C$ by vertices in $C'$ with labeling as in Def.~\ref{definition:fixedpairingtypedeformation} and let
    \begin{equation}
        W = \{\bsj \in C'{\setminus}C \mid \scomm{h_{\bsj}}{h_C}=-1\},
    \end{equation}
    and
    \begin{equation}
        U^{(r,s)} = \Gamma_{C'}(\bsj_{s-1}^{(r)}){\setminus}\Gamma_{C'}[\bsk_{s-1}^{(r)}],
    \end{equation}
    for $r\in\{0,\dots,m\}$ and $s\in\{1,\dots,\ell_r\}$. The set $W$ consists of vertices in $C'$ whose corresponding term anticommutes with $h_C$. The set $U^{(r,s)}$ consists of vertices in $C'{\setminus}\{\bsk_{s-1}^{(r)}\}$ dependent on the deformation by $\bsj_{s-1}^{(r)}\prec_{C^{(r,s-1)}}\bsj_{s-1}^{(r)}$. Further, let
    \begin{equation}
        U = \bigcup_{r=0}^m\bigcup_{s=1}^{\ell_r}U^{(r,s)},
    \end{equation}
    be the union of all such dependent subsets.
\end{definition}

Note that since $h_C$ and $h_{C'}$ anticommute, $\abs{W}$ is odd. Additionally, note that $|U^{(r,s)}|$ is at most one by Corollary~\ref{corollary:vertexcycleneighboring} with the assumption that $\bsj_{s-1}^{(r)}$ is in $C'$. If $\bsk_{s-1}^{(r)}$ is also in $C'$, then the only element of $U^{(r,s)}$ is the additional neighbor to $\bsj_{s-1}^{(r)}$ in $C'$.

With these definitions, we describe our proof strategy. Let $\overleftarrow{\mco}$ denote the fixed-pairing-type deformation related to $\overrightarrow{\mco}$ by reversing the order of the paths in $\overrightarrow{\mco}$. We first give a sufficient condition for $\overleftarrow{\mco}$ to correspond to a fixed-pairing-type deformation of $C'$ by vertices in $C$. We similarly refer to such a deformation as palindromic. We next give a search process to find a palindromic deformation. This process considers the full set of vertices $W$, defined in Def.~\ref{definition:anticommutingdependentsubsets}, as potential initializing vertices for our desired deformation. For each such initializing vertex $\bsj$ in $W$, we produce a unique fixed-pairing-type deformation $\overrightarrow{\mco}_{\bsj}$. If $\overrightarrow{\mco}_{\bsj}$ is not palindromic, then it is \emph{obstructed} by another vertex $\bsj'$ in $W$. This allows us to define a directed graph, called the \emph{obstruction graph}.
\begin{definition}[Obstruction graph and coloring classes]
    \label{definition:obstructiongraphcoloringclasses}
    The \emph{obstruction graph} $\mcd=(W,D)$ is a directed graph with vertex set $W$ as defined in Def.~\ref{definition:anticommutingdependentsubsets} and with $(\bsj\rightarrow\bsj')$ in $D$ if $\bsj'$ obstructs $\overrightarrow{\mco}_{\bsj}$. The graph $\mcd$ has odd order and is bipartite with the coloring class of $\bsj$ in $W$ given by the pairing type of $\overrightarrow{\mco}_{\bsj}$. The \emph{source set} of $\mcd$ is the coloring class with larger size, and the \emph{obstruction set} is that with smaller size.
\end{definition}

If $\mcd$ is such that no $\overrightarrow{\mco}_{\bsj}$ is palindromic for any $\bsj$ in $W$, then we prescribe a corrective \emph{rerouting} of $\mcd$ to $\mcd^{(0)}$. In this case, we consider the obstruction graph $\mcd^{(0)}$ defined on a vertex subset $W'$ of $W$ with each $\overrightarrow{\mco}_{\bsj}$ replaced with $\overrightarrow{\mco}'_{\bsj}$ for all $\bsj \in W'$. We apply this procedure recursively by updating $\mcd^{(i)}$ to $\mcd^{(i+1)}$ and since the order of the obstruction graph is strictly decreasing, this search process is guaranteed to find a palindromic (i.e. unobstructed) deformation.

Our final step of the proof is to show that the term indexed by $(\wt{C},\wt{C}')$ cancels the $(C,C')$ term. These claims follow from additional properties of the palindromic deformation $\overrightarrow{\mco}$. Before we proceed with the proof, we list several important properties of fixed-pairing-type deformations.

\subsubsection{Fixed-Pairing-Type Deformations}
\label{subsubsection:FixedPairingTypeDeformations}

A number of important properties follow from the assumption that a deformation $\overrightarrow{\mco}$ has fixed-pairing-type. These generally concern the intersections between path components $P$ of $\mco$ and commutation relations between associated operators and the Hamiltonian terms. For the remainder of this section, we assume that $\overrightarrow{\mco}$ is a fixed-pairing-type deformation of $C$ by vertices in $C'$ with labeling as in Def.~\ref{definition:fixedpairingtypedeformation}. Without loss of generality, we assume that $\overrightarrow{\mco}$ is $(a,c)$ pairing.
\begin{lemma}
    \label{lemma:vertexonecoloringclass}
    Let $\bsv$ be a vertex in exactly one coloring class of $C_a,$ $C_b$, $C'_c$, and $C'_d$, then $\bsv$ is contained in at most one path component of $\overrightarrow{\mco}$.
\end{lemma}
\begin{proof}
    Without loss of generality, assume that $\bsv$ is in $C_a{\setminus}(C_b \cup C')$. Clearly, $\bsv$ is not in $P$ if $P\in\mco$ is a component of $G[C_b{\oplus}C'_d]$. By construction, all paths $P$ in $\mco$ that are components of $G[C_a{\oplus}C'_c]$ are disjoint (any repeated paths are removed in the union defining $\mco$), so $\bsv$ is contained in at most one such component $P$.
    
    Thus, for $\bsv$ to be in multiple elements of $\overrightarrow{\mco}$, then $\bsv$ is in $P^{(r)}$ and $P^{(r)}$ is repeated in $\overrightarrow{\mco}$. Assume that this is the case. Let $g>r$ be the smallest index such that $P^{(g)}=P^{(r)}$ for $P^{(g)}\in\overrightarrow{\mco}$. Since $\bsv$ is in $C_a{\setminus}C'_c$, then $\bsv=\bsk_{s-1}^{(r)}$ for some $s$. This gives that $\bsv$ is not in $C^{(r)}$, and thus $\bsv$ is not in $C^{(g-1)}$ by our assumption that $g$ is minimal. We then require that $\bsk=\bsj^{(g)}_{t-1}$ for some $t$ in order for $\bsj_{s-1}^{(g)}\prec_{C^{(g,s-1)}}\bsk_{s-1}^{(g)}$ for all $s$ by our assumption on $\overrightarrow{\mco}$. However, this contradicts the assumption on $\overrightarrow{\mco}$ that $\bsj_{t-1}^{(g)}$ is in $C'$, since $\bsv$ is not in $C'$. Therefore, $\bsv$ is contained in at most one path component of $\overrightarrow{\mco}$, completing the proof.
\end{proof}

By construction, the coloring classes of $C$ are disjoint and similarly for $C'$. Thus, in order to apply Lemma~\ref{lemma:vertexonecoloringclass}, it is sufficient to show that the vertex $\bsv$ is in $C{\oplus}C'$. We have the following corollaries.
\begin{corollary}
    \label{corollary:pathuniqueendpoints}
    Let $P^{(r)}$ be a path in $\overrightarrow{\mco}$. The endpoints $\bsj_0^{(r)}$ and $\bsk_{\ell_r-1}^{(r)}$ are not contained in any other path in $\overrightarrow{\mco}$.
\end{corollary}
\begin{proof}
    Without loss of generality, assume that $\bsj_{0}^{(r)}$ is in $C'_c{\setminus}C_a$. By construction, $\bsj_{0}^{(r)}$ has precisely one neighbor in $C_a$. If $\bsj_{0}^{(r)}$ has an additional neighbor $\bsv$ in $C_a{\setminus}\{\bsk_{0}^{(r)}\}$, then $\bsv$ is not in $C'_c$, since it is neighboring to $\bsj_0^{(r)}$, but then $\bsv$ is also in $P^{(r)}$. This contradicts the assumption that $\bsj_0^{(r)}$ is an endpoint of $P^{(r)}$. Thus, $\bsj_0^{(r)}$ is not in $C$, since every element of $C{\setminus}C_a$ has two neighbors in $C_a$. By Lemma~\ref{lemma:vertexonecoloringclass}, $\bsj_0^{(r)}$ is in exactly one path component of $\overrightarrow{\mco}$. A similar argument holds for $\bsk_{\ell_r-1}^{(r)}$, completing the proof.
\end{proof}
We immediately obtain the following corollary.
\begin{corollary}
    \label{corollary:pathcomponentsdistinct}
    The path components of $\overrightarrow{\mco}$ are distinct subsets.
\end{corollary}
We can therefore relax our distinction between $\mco$ and $\overrightarrow{\mco}$ in our notation. Note that path components of $\overrightarrow{\mco}$ may intersect, but any vertex in the intersection of such components must be in $C \cap C'$. This gives the following corollary.
\begin{corollary}
    \label{corollary:vertextwopathcomponents}
    Every vertex $\bsv$ is contained in at most two path components of $\overrightarrow{\mco}$. If $\bsv$ is in $P^{(r)} \cap P^{(g)}$ for $P^{(r)},P^{(g)}\in\mco$ with $g$ and $r$ distinct, then either $\bsv=\bsj_{s-1}^{(r)}=\bsk_{t-1}^{(g)}$ for some $s$ and $t$ or $\bsv=\bsk_{s-1}^{(r)}=\bsj_{t-1}^{(g)}$ for some $s$ and $t$.
\end{corollary}
\begin{proof}
    Assume that $\bsv$ is in $P^{(r)} \cap P^{(g)} \cap P^{(q)}$ for $P^{(r)},P^{(g)},P^{(q)}\in\mco$ with $r$, $g$, and $q$ distinct. At least two of $\{P^{(r)},P^{(g)},P^{(q)}\}$ are components of either $G[C_a{\oplus}C'_c]$ or $G[C_b{\oplus}C'_d]$ by our assumption that $\overrightarrow{\mco}$ is $(a,c)$ pairing. Without loss of generality, assume that $P^{(r)}$ and $P^{(g)}$ are both components of $G[C_a{\oplus}C'_c]$. However, then $P^{(r)}$ and $P^{(g)}$ are distinct by Corollary~\ref{corollary:pathcomponentsdistinct}, and so their intersection is empty. Therefore, every vertex $\bsv$ is contained in at most two path components of $\overrightarrow{\mco}$.
    
    Now assume that $\bsv$ is in $P^{(r)} \cap P^{(g)}$ for $P^{(r)},P^{(g)}\in\mco$ with $g>r$. By the previous argument, assume without loss of generality that $P^{(r)}$ is a component of $G[C_a{\oplus}C'_c]$ and $P^{(g)}$ is a component of $G[C_b{\oplus}C'_d]$. Since the coloring classes of $C$ and $C'$ are disjoint, then $\bsv$ is in either $C_a \cap C'_d$ or $C_b \cap C'_c$ and $\bsv$ is contained in exactly two coloring classes of $C_a$, $C_b$, $C'_c$, and $C'_d$. If $\bsv$ is in $C_a \cap C'_d$, then $\bsv=\bsk_{s-1}^{(r)}$ for some $s$, since $\bsv$ is in $C_a{\setminus}C'_c$, and $\bsv=\bsj_{t-1}^{(g)}$ for some $t$ since $\bsv$ is in $C'_d{\setminus}C_b$. If $\bsv$ is in $C_b \cap C'_c$, then $\bsv=\bsj_{s-1}^{(r)}$ for some $s$, since $\bsv$ is in $C'_c{\setminus}C_a$, and $\bsv=\bsk_{t-1}^{(g)}$ for some $t$, since $\bsv$ is in $C_b{\setminus}C'_d$. The other cases follow similarly. This completes the proof.
\end{proof}
This corollary shows that a vertex can be introduced and removed at most once in a fixed-pairing-type deformation $\overrightarrow{\mco}$.

We now consider relations involving the vertex subsets of Def.~\ref{definition:anticommutingdependentsubsets}
\begin{corollary}
    \label{corollary:vertexonepathcomponent}
    Every vertex $\bsw$ in $W$ is contained in at most one path component of $\overrightarrow{\mco}$. Further, if a vertex $\bsu$ is in $U^{(r,s)}{\setminus}P^{(r)}$ for some $r$ and $s$, then there is a mutual neighbor $\bsv$ in $C^{(r-1)}$ to $\bsu$, $\bsj_{s-1}^{(r)}$, and $\bsk_{s-1}^{(r)}$. The vertices $\bsu$, $\bsj_{s-1}^{(r)}$, and $\bsv$ are in at most one path component of $\overrightarrow{\mco}$.
\end{corollary}
\begin{proof}
    Clearly, $\bsw$ is not in $C$ if $\bsw$ is in $W$, so $\bsw$ is contained in at most one path component of $\overrightarrow{\mco}$ by Lemma~\ref{lemma:vertexonecoloringclass}. Now consider $\bsu$ in $U^{(r,s)}{\setminus}P^{(r)}$. Assume without loss of generality that $P^{(r)}$ is a component of $G[C_a{\oplus}C'_c]$, then $\bsu$ is in $C'_d$, since it is a neighbor in $C'$ to $\bsj_{s-1}^{(r)}$ in $C'_c{\setminus}C_a$. Thus, $\bsv$ is neighboring to $\bsu \in C'_d$, $\bsk_{s-1}^{(r)}$ in $C_a$, and $\bsj_{s-1}^{(r)}$ in $C'_c$. Therefore $\bsv$ is in $C_b{\setminus}C'$, and so $\bsv$ is contained in at most one path component of $\overrightarrow{\mco}$. Similarly, $\bsj_{s-1}^{(r)}$ is neighboring to $\bsv$ in $C_b$, $\bsu$ in $C'_d$, and $\bsk_{s-1}^{(r)}$ in $C_a$. Therefore $\bsj_{s-1}^{(r)}$ is in $C'_c{\setminus}C$, and so $\bsj_{s-1}^{(r)}$ is contained in at most one path component of $\overrightarrow{\mco}$. Finally, $\bsu$ is neighboring to $\bsv$ in $C_b$, and if $\bsu$ is in $C_a$, then $\bsu$ is in $P^{(r)}$, which contradicts our assumption. Therefore, $\bsu$ is in $C_a{\setminus}C'$, and so $\bsu$ is contained in at most one path component of $\overrightarrow{\mco}$. This completes the proof.
\end{proof}

We now prove statements concerning commutation relations.
\begin{lemma}[{restate=[name=restatement]PairingRelations}]
    \label{lemma:pairingrelations}
    Let $P^{(r)}$ be a path in $\overrightarrow{\mco}$ and let $\bsu$ be a vertex in $C'$.
    \begin{enumerate}
        \item[(a)] If $h_{\bsu}$ and $h_{P^{(r)}}$ anticommute, then $\bsu$ is in $U^{(r,s)}$ for some $s$ if and only if $\bsu$ and $\bsj_0^{(r)}$ are distinct and $\Delta_{C^{(r)}}(\bsu)=\Delta_{C^{(r-1)}}(\bsu)+1$.
        \begin{enumerate}
            \item[(i)] If $\Delta_{C^{(r)}}(\bsu)=3$ with $\bsu\prec_{C^{(r)}}\bsv$, then $\bsv$ is in $C$ and the pairing type of $\{\bsu,\bsv\}$ is the same as that of $\overrightarrow{\mco}$.
            \item[(ii)] $\Delta_{C^{(r)}}(\bsu)=4$ if and only if $\Delta_C(\bsu)=3$ and the pairing type of $\bsu\prec_{C}\bsv$ is opposite to that of $\overrightarrow{\mco}$.
        \end{enumerate}
        \item[(b)] If $h_{\bsu}$ and $h_{P^{(r)}}$ commute and $\bsu$ is in $U^{(r,s)}$ for some $s$, then $\Delta_{C^{(r)}}(\bsu)=4$ and there is a vertex $\bsu'$ in $U^{(r,s')}$ with $s'<s$ such that $h_{\bsu'}$ and $h_{P^{(r)}}$ anticommute. If $P$ is the unique component containing $\bsu'$ with the same pairing type as $\overrightarrow{\mco}$, then $\bsu$ is in $P$.
    \end{enumerate}
\end{lemma}

\begin{figure}[ht!]
    \begin{subfigure}{\columnwidth}
         \centering
         \includegraphics[width=0.7\columnwidth]{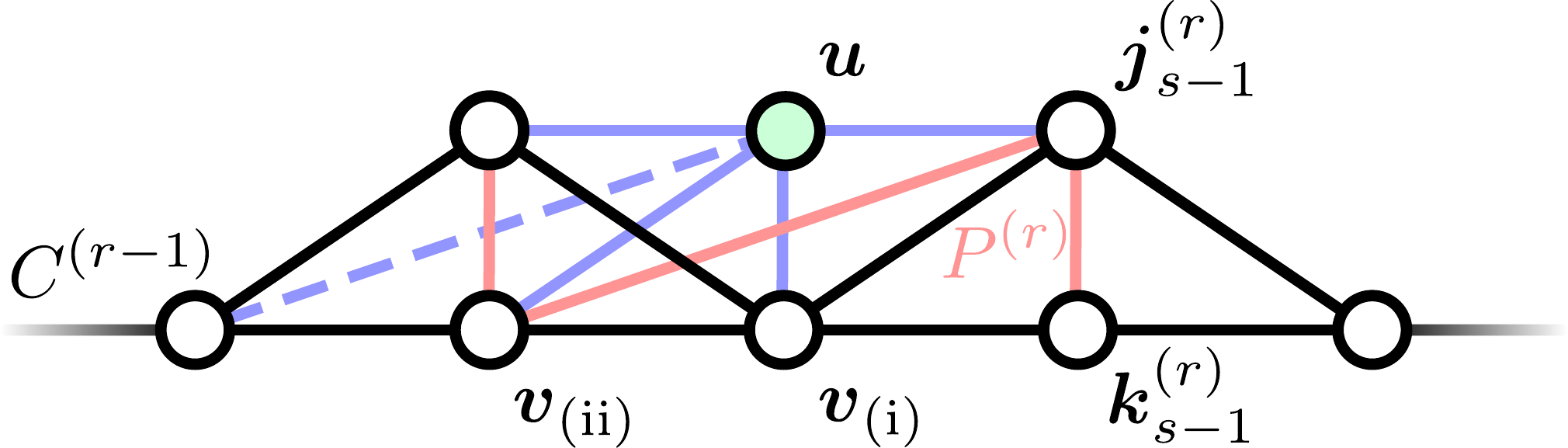}
         \caption{}
         \label{figure:paringrelationshipa}
     \end{subfigure}
     \\
     \begin{subfigure}{\columnwidth}
         \centering
         \includegraphics[width=0.95\columnwidth]{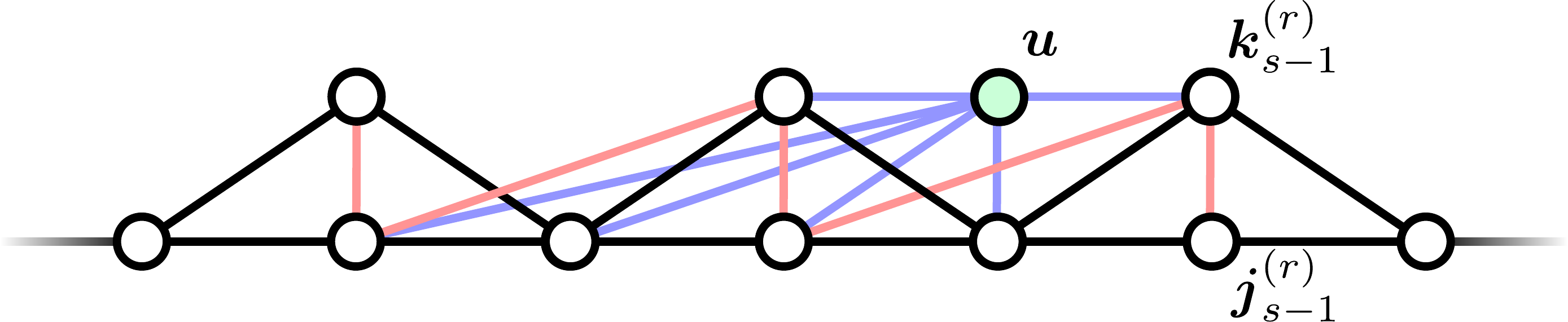}
         \caption{}
         \label{figure:paringrelationshipb}
     \end{subfigure}
    \caption{An illustration of Lemma~\ref{lemma:pairingrelations}. (a) If $h_{\bsu}$ and $h_{P^{(r)}}$ anticommute, then $\bsu$ has either two or three neighbors in $C^{(r-1)}$ (the dashed edge may or may not be present). Deforming by $P^{(r)}$ increases $\Delta_C(\bsu)$ by one. If $\bsu$ has three neighbors in $C^{(r)}$, then $\bsu\prec_{C^{(r)}}\bsv_{\text{(i)}}$ with the same pairing type as $\overrightarrow{\mco}$. In this case, we can deform by $\bsu$. If $\bsu$ has three neighbors in $C^{(r-1)}$, then $\bsu\prec_{C^{(r-1)}}\bsv_{\text{(ii)}}$, and in fact, $\bsu\prec_{C}\bsv_{\text{(ii)}}$ with the opposite pairing type to $\overrightarrow{\mco}$. (b) If $h_{\bsu}$ and $h_{P^{(r)}}$ commute and $\bsu$ is in $U^{(r,s)}$, then $\Delta_{C^{(r)}}(\bsu)=4$.}
    \label{figure:paringrelations}
\end{figure}
We prove Lemma~\ref{lemma:pairingrelations} in Appendix~\ref{section:PairingRelationsProof} by enumerating all possible neighboring relations for $\bsu$ under the assumptions. An illustration of Lemma~\ref{lemma:pairingrelations} is given in Fig.~\ref{figure:paringrelations}. We also prove the following corollary in Appendix~\ref{section:PairingRelationsProof}.
\begin{corollary}[{restate=[name=restatement]VertexAnticommutingDistinctPaths}]
    \label{corollary:vertexanticommutingdistinctpaths}
    Let $\overrightarrow{\mco}$ and $\overrightarrow{\mco}'$ be possibly non-distinct deformations of the same fixed-pairing type. There is no vertex $\bsu$ in $U^{(r,s)} \cap U^{(g,t)}$ such that $h_{\bsu}$ anticommutes with $h_{P^{(r)}}$ and $h_{P^{(g)}}$ for distinct path components $P^{(r)}$ in $\overrightarrow{\mco}$ and $P^{(g)}$ in $\overrightarrow{\mco}'$ for some $s$ and $t$.
\end{corollary}
The following corollary follows immediately from Lemma~\ref{lemma:pairingrelations} and Corollary~\ref{corollary:vertexanticommutingdistinctpaths}.
\begin{corollary}
    \label{corollary:vertexanticommutingdistinctpathsendpoint}
    If $h_{\bsu}$ anticommutes with $h_{P^{(r)}}$ and $h_{P^{(g)}}$ for distinct path components $P^{(r)}$ in $\overrightarrow{\mco}$ and $P^{(g)}$ in $\overrightarrow{\mco}'$ in the setting of Corollary~\ref{corollary:vertexanticommutingdistinctpaths}, then $\bsu$ is in $\{\bsj_0^{(g)},\bsj_0^{(r)}\}$.
\end{corollary}

We proceed by applying these statements to show that a palindromic deformation exists.

\subsubsection{Solution Criterion}
\label{subsubsection:SolutionCriterion}

We now give a sufficient criterion for a fixed-pairing-type deformation $\overrightarrow{\mco}$ of $C$ by vertices in $C'$ to be palindromic, i.e., $\overleftarrow{\mco}$ is a fixed-pairing-type deformation of $C'$ by vertices in $C$. This is the case if every vertex in $C'$ that is dependent on some single-vertex deformation in $\overrightarrow{\mco}$ is in $C^{(m)}$.

\begin{lemma}[{restate=[name=restatement]SolutionCriterion}]
    \label{lemma:solutioncriterion}
    Let $\overrightarrow{\mco}$ be a fixed-pairing-type deformation of $C$ by vertices in $C'$ with labeling as in Def.~\ref{definition:fixedpairingtypedeformation} and let $U$ be a vertex subset as in Def.~\ref{definition:anticommutingdependentsubsets}. If $U$ is a subset of $C^{(m)}$, then $\overrightarrow{\mco}$ is palindromic.
\end{lemma}
We prove Lemma~\ref{lemma:solutioncriterion} in Appendix~\ref{section:SolutionCriterionProof}.

\subsubsection{Search Process}
\label{subsubsection:SearchProcess}

We now describe a search process to find a palindromic deformation $\overrightarrow{\mco}$ using the criterion of Lemma~\ref{lemma:solutioncriterion}. As a subroutine, we describe an iterative obstruction search process, called $\mathsf{OS}$, with the following signature
\begin{equation}
    (\overrightarrow{\mco}_{f},V_f,U_f) \coloneqq \mathsf{OS}(\overrightarrow{\mco}_{i},\bsj_{i},W_i).
\end{equation}
This function is used to initialize the obstruction graph $\mcd=(W,D)$ and perform the rerouting step. Here, $\overrightarrow{\mco}_i$ is an initial fixed-pairing-type deformation, and $\bsj_i$ is a vertex in $C'$ such that $h_{\bsj_{i}}$ anticommutes with $h_{\wt{C}_i}$ if $\wt{C}_i\in\avg{C}$ is the even hole given by deforming $C$ by $\overrightarrow{\mco}_i$, and $\bsj_i\prec_{\wt{C}_i}\bsk$ has the same pairing type as $\overrightarrow{\mco}_i$. We assume that $W_i$ is a subset of $W$, with $W$ as defined in Def.~\ref{definition:anticommutingdependentsubsets}, and $\bsj_i$ is not in $W_i$. The output deformation $\overrightarrow{\mco}_f$ has the same pairing type as $\overrightarrow{\mco}_i$, and $V_f$ and $U_f$ are sets of at most one vertex. If $V_f$ is non-empty, then the only element $\bsj_f$ of $V_f$ obstructs $\overrightarrow{\mco}_f$. If $U_f$ is non-empty, then the only element $\bsu_f$ of $U_f$ is a vertex in $C'$ required for the obstruction graph update. Note that $V_f$, $U_f$, $\overrightarrow{\mco}_i$, $\overrightarrow{\mco}_f$, and $W_i$ may be empty. The pairing types of $\overrightarrow{\mco}_i$ and $\overrightarrow{\mco}_f$ are however well defined.

We define the obstruction search process $\mathsf{OS}$ as follows. Initialize the iteration variables as $\bsu_{t}=\bsj_i$ and $\overrightarrow{\mco}_t=\overrightarrow{\mco}_i$. At each iteration, $\overrightarrow{\mco}_t$ is updated followed by $\bsu_t$. Let $P$ be the unique component of $G[C_j{\oplus}C'_{\sigma(j)}]$ to which $\bsu_t \in C'_{\sigma(j)}$ is a member, where $\sigma$ gives the pairing type of $\overrightarrow{\mco}_i$. If $P \cap W_i$ is non-empty, suppose that $\bsj$ is the vertex in $W_i$ with the smallest distance to $\bsu_t$ along $P$ (we shall show that this vertex is unique). In this case, the process terminates with $\overrightarrow{\mco}_f=\overrightarrow{\mco}_t$, $V_f=\{\bsj\}$, and $U_f=\{\bsu_t\}$. If $P \cap W_i$ is non-empty, then $P$ is a path and we update $\overrightarrow{\mco}_t$ by concatenating $P$ as the last path in the deformation.

In order to update $\bsu_t$, we let
\begin{equation}
    \mcl_{r,s} = \left(\sum_{g=0}^{r-1}\ell_g\right)+s,
\end{equation}
denote the number of single-vertex deformations required to reach $C^{(r,s)}$. Additionally, we define $U^{(r,s)}(\overrightarrow{\mco}_t)$ as the set in Def.~\ref{definition:singlevertexdeformation} corresponding to the set of paths $\mco_t$. Recall that $U^{(r,s)}(\overrightarrow{\mco}_t)$ contains at most one element for a given $r$ and $s$. $\bsu_t$ is updated to the only member $\bsu$ in $U^{(r,s)}(\overrightarrow{\mco}_t)$ with the maximum value of $\mcl_{r,s}$ such that
\begin{enumerate}
    \item[(i)] $U^{(r,s)}(\overrightarrow{\mco}_t){\setminus}C^{(\abs{\mco_t}-1)}$ is non-empty.
    \item[(ii)] $h_{\bsu}$ and $h_{P^{(r)}}$ anticommute.
\end{enumerate}
If no such vertex exists, then the process terminates with $\overrightarrow{\mco}_f=\overrightarrow{\mco}_t$, $V_f=\varnothing$, and $U_f=\varnothing$. In this case, we show that $\overrightarrow{\mco}_f$ satisfies the criterion of Lemma~\ref{lemma:solutioncriterion} and that $\overrightarrow{\mco}=\overrightarrow{\mco}_f$ is palindromic.

Using this function, we initialize the obstruction graph on the set of vertices $W$, by defining 
\begin{equation}
    (\overrightarrow{\mco}_{\bsj},V_{\bsj},U_{\bsj}) \coloneqq \mathsf{OS}(\overrightarrow{\mco}^{(-2)}_{\bsj},\bsj,W{\setminus}\{\bsj\}), \label{equation:initializingstep}
\end{equation}
for all $\bsj \in W$. Here, $\overrightarrow{\mco}^{(-2)}_{\bsj}$ is the empty sequence and is defined to have the same pairing type as $\bsj\prec_C\bsk$. If $V_{\bsj}$ is non-empty, then the only member $\bsj'$ of $V_{\bsj}$ obstructs $\overrightarrow{\mco}_{\bsj}$ and we include the arc $(\bsj\rightarrow\bsj')$ in $D$. By construction, every vertex in $W$ has at most one outgoing arc in $\mcd$ and there are no self-loops. We shall show that $\mcd$ is a bipartite directed graph with odd order and coloring classes given by the pairing types.

We now specify the procedure to update $\mcd^{(i)}$ to $\mcd^{(i+1)}$ with the convention that $\mcd^{(-1)}=\mcd$ and all related objects notated similarly. In each iteration where there is no unobstructed deformation, we update $W^{(i)}$ to $W^{(i+1)}$ by removing one vertex from the source set and one vertex from the obstruction set of $\mcd^{(i)}$. Thus, $\mcd^{(i+1)}$ is bipartite with odd order and whose source and obstruction sets are subsets of the source and obstruction sets of $\mcd^{(i)}$, respectively. $(\overrightarrow{\mco}^{(i)}_{\bsj},V^{(i)}_{\bsj},U^{(i)}_{\bsj})$ is updated to $ (\overrightarrow{\mco}^{(i+1)}_{\bsj},V^{(i+1)}_{\bsj},U^{(i+1)}_{\bsj})$ and $D^{(i+1)}$ is such that $(\bsj\rightarrow\bsj')$ is in $D^{(i+1)}$ for all $\bsj,\bsj' \in W^{(i+1)}$ if and only if $\bsj'$ obstructs $\overrightarrow{\mco}^{(i+1)}_{\bsj}$.

In particular, suppose that there is no vertex $\bsj$ in $W^{(i)}$ for which $\overrightarrow{\mco}^{(i)}_{\bsj}$ is unobstructed. Thus, each vertex of $\mcd^{(i)}$ has exactly one outgoing arc and there must be at least one vertex $\bsj''$ in the obstruction set of $\mcd^{(i)}$ with at least two incoming arcs. We show that every vertex has at most two incoming arcs, and so $\bsj''$ has exactly two incoming arcs. Let $\bsj$ and $\bsj'$ be the vertices in the source set whose outgoing arcs are incoming to $\bsj''$, i.e., $(\bsj\rightarrow\bsj'')$ and $(\bsj'\rightarrow\bsj'')$ are in $D^{(i)}$. We additionally show that $\bsj''$ is the only member of exactly one of $U^{(i)}_{\bsj}$ or $U^{(i)}_{\bsj'}$. Suppose that $\bsj''$ is in $U^{(i)}_{\bsj'}$ and let $\bsu^{(i)}_{\bsj,f}$ be the only member of $U_{\bsj}^{(i)}$.
Further, let $W^{(i+1)}=W^{(i)}{\setminus}\{\bsj',\bsj''\}$ and
\begin{equation}
    (\overrightarrow{\mco}^{(i+1)}_{\bsv},V^{(i+1)}_{\bsv},U^{(i+1)}_{\bsv}) = (\overrightarrow{\mco}^{(i)}_{\bsv},V^{(i)}_{\bsv},U^{(i)}_{\bsv}), \label{equation:searchprocesstrivialupdate}
\end{equation} 
for all $\bsv \in W^{(i+1)}{\setminus}\{\bsj\}$. We then set
\begin{align}
    &(\overrightarrow{\mco}^{(i+1)}_{\bsj},V^{(i+1)}_{\bsj},U^{(i+1)}_{\bsj}) \notag \\
    &\quad= \mathsf{OS}(\overrightarrow{\mco}^{(i)}_{\bsj'}\|\overrightarrow{\mco}^{(i)}_{\bsj},\bsu^{(i)}_{\bsj,f},W^{(i+1)}{\setminus}\{\bsj\}), \label{equation:searchprocessupdate}
\end{align}
where $\|$ denotes the concatenation of sequences. This completes the description of the search process.

\subsubsection{Correctness of Search Process}
\label{subsubsection:CorrectnessSearchProcess}

We show that the search process returns a valid palindromic deformation inductively. The base case is given as the following lemma.
\begin{lemma}[{restate=[name=restatement]SearchProcessBaseCase}]
    \label{lemma:searchprocessbasecase}
    The following statements hold for $\{\overrightarrow{\mco}_{\bsj}\}_{\bsj \in W}$ and $\mcd$ at the initialization of the search process.
    \begin{enumerate}
        \item[(i)] The deformation $\overrightarrow{\mco}_{\bsj}$ is a possibly empty fixed-pairing-type deformation.
        \item[(ii)] The obstruction graph $\mcd$ is bipartite with coloring classes given by the pairing types.
        \item[(iii)] The deformations $\{\overrightarrow{\mco}_{\bsj}\}_{\bsj \in W}$ of a given pairing type are pairwise disjoint as sets of induced paths.
        \item [(iv)] Every vertex in the obstruction graph $\mcd$ has at most two incoming arcs. If $(\bsj\rightarrow\bsj'')$ and $(\bsj'\rightarrow\bsj'')$ are in $D$, then $\bsj''$ is the only member of exactly one of $U_{\bsj}$ or $U_{\bsj'}$.
    \end{enumerate}
\end{lemma}
We prove Lemma~\ref{lemma:searchprocessbasecase} in Appendix~\ref{section:SearchProcessBaseCaseProof}. The induction is given as the following lemma.
\begin{lemma}[{restate=[name=restatement]SearchProcessInduction}]
    \label{lemma:searchprocessinduction}
    The statements of Lemma~\ref{lemma:searchprocessbasecase} hold for $\{\overrightarrow{\mco}^{(i)}_{\bsj}\}_{\bsj \in W^{(i)}}$ and $\mcd^{(i)}$ for all steps $i$ in the search process.  
\end{lemma}
We prove Lemma~\ref{lemma:searchprocessinduction} in Appendix~\ref{section:SearchProcessInductionProof}.
This process therefore returns an unobstructed deformation $\overrightarrow{\mco}$, since, at each step of the search process, exactly one source vertex and one obstruction vertex are removed. 
In the worst case, the search proceeds until there are no more obstruction vertices, in which case, the unobstructed deformation is trivial. We complete the proof of correctness with the following lemma.
\begin{lemma}
    Any unobstructed deformation satisfies the criterion of Lemma~\ref{lemma:solutioncriterion} and is palindromic.
\end{lemma}
\begin{proof}
    If a deformation is unobstructed, then there is no vertex $\bsu$ in $U^{(r,s)}{\setminus}C^{(m)}$ such that $h_{\bsu}$ and $h_{P^{(r)}}$ anticommute. If there is a vertex $\bsu$ in $U^{(r,s)}{\setminus}C^{(m)}$ such that $h_{\bsu}$ and $h_{P^{(r)}}$ commute, then by Lemma~\ref{lemma:pairingrelations}, there is another vertex $\bsu'$ in $U^{(r,s')}$ with $s'<s$ such that $h_{\bsu'}$ and $h_{P^{(r)}}$ anticommute, and $\bsu'$ and $\bsu$ are members of the same path component $P$ with the same pairing type as $\overrightarrow{\mco}$. Then $\bsu'$ is not in $C^{(r)}$, since $h_{\bsu'}$ and $h_{P^{(r)}}$ anticommute. Therefore, $\bsu'$ is not in $C^{(r-1)}$ and $\bsu'$ is not in $P^{(r)}$, since $\bsu'$ is in $U^{(r,s')}$ and $h_{\bsu'}$ and $h_{P^{(r)}}$ anticommute. Then $P$ must be in $\overrightarrow{\mco}$ for $\bsu'$ to be in $C^{(m)}$. Since $\bsu'$ is in at most one such path component by Corollary~\ref{corollary:vertexonepathcomponent}, then $P=P^{(g)}$ must be in $\overrightarrow{\mco}$ with $g>r$ for $\bsu'$ to be in $C^{(m)}$. By deforming by $P^{(g)}$, we have that $\bsu$ is in $C^{(g)}$, since $\bsu$ and $\bsu'$ are members of the same coloring class in $C'$ and $\bsu$ is in $U^{(r,s)}$. If $\bsu$ is not in $C^{(m)}$, then $\bsu$ is a member of an additional path component of $\overrightarrow{\mco}$. By Corollary~\ref{corollary:vertexonepathcomponent}, this path component must be $P^{(r)}$, however, then $\bsu=\bsk_{s-2}^{(r)}$ and $\bsu=\bsj_{t-1}^{(g)}$ by Corollary~\ref{corollary:vertextwopathcomponents}. Since $\bsu$ is contained in no additional path components of $\overrightarrow{\mco}$, then $\bsu$ is in $C^{(m)}$. Therefore, there is no vertex $\bsu$ in $U^{(r,s)}{\setminus}C^{(m)}$ for any $r$ and $s$, and so $U$ is a subset of $C^{(m)}$. This completes the proof.
\end{proof}

\subsubsection{Proof of Cancellation}
\label{subsubsection:ProofCancellation}

We now prove Theorem~\ref{theorem:conservedcharges}~(iii). 
\begin{proof}[Proof of Theorem~\ref*{theorem:conservedcharges}~$(iii)$]
    Let
    \begin{equation}
        h_{\mco} = \prod_{r=0}^mh_{P^{(r)}},
    \end{equation}
    where the product taken in ascending order in $r$ and 
    \begin{equation}
        h_{P^{(r)}} =  
        \begin{cases}
            h_{\left(P^{(r)} \cap C_a\right)}h_{\left(P^{(r)} \cap C'_c\right)} & P^{(r)} \subseteq C_a{\oplus}C'_c, \\
            h_{\left(P^{(r)} \cap C_b\right)}h_{\left(P^{(r)} \cap C'_d\right)} & P^{(r)} \subseteq C_b{\oplus}C'_d.
        \end{cases}
    \end{equation}
    We have
    \begin{align}
        h_Ch_{C'} &= \scomm{h_C}{h_{C'}}h_{C'}h_C, \\
        h_C h_{\mco}h_{\mco}h_{C'} &= \scomm{h_C}{h_{C'}}h_{C'}h^{\dagger}_{\mco}h^{\dagger}_{\mco}h_{C}, \\
        h_{\wt{C}}h_{\wt{C}'} &= \scomm{h_C}{h_{C'}}h_{\wt{C}'}h_{\wt{C}},
    \end{align}
    where all scale factors cancel upon conjugating by $h_{\mco}^2=(h^{\dagger}_{\mco})^2$. We use the fact that independent set factors such as $h_{C_a}$ and $h_{C_b}$ commute to multiply their corresponding factors in $h_{P^{(r)}}$ for each $r$. This is possible as $h_{C^{(r)}}$ is an even hole at each step of the deformation. Note that the first and last line are scalar commutators, so we have
    \begin{equation}
        \scomm{h_{\wt{C}}}{h_{\wt{C}'}} = \scomm{h_{C}}{h_{C'}} = -1.
    \end{equation}
    
    To show that the $(C,C')$ and $(\wt{C},\wt{C}')$ terms cancel, we show that $h_{\mco}$ is anti-hermitian. Note that the search process is such that an odd number of vertices in the obstruction graph are included in $\mco$, and that every such vertex is included in at most one path from $\overrightarrow{\mco}$. Further, note that $h_{P^{(r)}}$ and $h_{C^{(r-1)}}$ anticommute by construction. Since $C^{(r-1)}=C{\oplus}\bigoplus_{j=0}^{r-1}P^{(j)}$, then $h_{P^{(r)}}$ has the opposite commutation relation to $\prod_{j=0}^{r-1}h_{P^{(j)}}$ as it does to $h_{C}$. The commutation relation between $h_{P^{(r)}}$ and $h_{C}$ is the number of vertices from $W$ contained in $P^{(r)}$. Further, each factor $h_{P^{(r)}}$ is anti-hermitian. Thus
    \begin{align}
        h_{\mco}^{\dagger} &= \prod_{r=0}^mh^{\dagger}_{P^{(m-r)}} \\
        &= \prod_{r=0}^m(-h_{P^{(m-r)}}) \\
        &= -\prod_{r=0}^mh_{P^{(r)}} \\
        &= -h_{\mco}.
    \end{align}
    The third line follows by commuting each factor of $h_{P^{(r)}}$ through $\prod_{j=0}^{r-1}h_{P^{(j)}}$, which incurs an additional factor of $-1$ for every path with an even number of vertices in $W$. This factor gives an overall factor of $-1$ for each path $P^{(r)}$ containing an odd number of vertices in $W$. Since the total number of vertices in $W$ contained in any path in $\mco$ is odd, there must be an odd number of paths containing an odd number of vertices in $W$.
    
    By setting $h_{\mco}h_{\mco}^{\dagger}=B^2I$, we have
    \begin{align}
        h_Ch_{C'} &= -B^{-2}h_Ch_{\mco}h_{\mco}h_{C'} \\
        &= -h_{\wt{C}}h_{\wt{C}'}.
    \end{align}    
    Therefore, for every collection of terms in Eq.~(\ref{equation:conservedcyclecommutator}) related by a common collection of fixed-pairing-type multiple-path deformations, we fix a deformation to cancel terms pairwise, completing the proof.
\end{proof}

\section{Exact Solutions}
\label{section:ExactSolutions}

With the identification of the conserved charges, we proceed to prove the remainder of Result~\ref{result:freefermionsolution}. In particular, we show that within each symmetric subspace of the generalized-even-hole operators, a Hamiltonian with a simplicial, claw-free frustration graph exhibits a free-fermion solution.
\begin{theorem}[Free-Fermion Solution]
    \label{theorem:freefermionsolution}
    Let $H$ be an SCF Hamiltonian with frustration graph $G$. There exists a set of mutually commuting cycle symmetries $\{\jkg{C_0}{G}\}_{\avg{C_0}}$ for $H$, such that
    \begin{equation}
        H = \sum_{\mcj}\left(\sum_{j=1}^{\alpha(G)}\varepsilon_{\mcj,j}[\psi_{\mcj,j},\psi_{\mcj,j}^{\dagger}]\right)\Pi_{\mcj}, \notag
    \end{equation}
    with
    \begin{equation}
        \psi_{\mcj,j} = N_{\mcj,j}^{-1}\Pi_{\mcj}T(-u_{\mcj,j}) \chi T(u_{\mcj,j}), \notag
    \end{equation}
    where $\{\Pi_{\mcj}\}_{\mcj}$ is a complete set of projectors onto the mutual eigenspaces of $\{\jkg{C_0}{G}\}_{\avg{C_0}}$, and $u_{\mcj,j}$ satisfies $Z_{G}(-u_{\mcj,j}^2)\Pi_{\mcj}=0$.
    Furthermore, the projectors $\Pi_{\mcj}$ satisfy
    \begin{equation}
        [\Pi_{\mcj},\psi_{\mcj',j}] = 0, \notag
    \end{equation}
    for all $\mcj$, $\mcj'$, and $j$. The single-particle energies in the subspace labeled by $\mcj$ are given by $\varepsilon_{\mcj,j}=1/u_{\mcj,j}$.
\end{theorem}
The proof closely follows the analysis given in Refs.~\cite{elman2021free, fendley2019free} with slight modifications and generalizations where necessary. We first prove the following lemma relating the transfer operators to the generalized characteristic polynomial of the frustration graph $G$.
\begin{lemma}[{restate=[name=restatement]GeneralizedCharacteristicPolynomial}]
    \label{lemma:generalizedcharacteristicpolynomial}
    Let $H$ be a Hamiltonian with claw-free frustration graph $G$. The generalized characteristic polynomial $Z_{G}(-u^2)$ is given by
    \begin{equation}
        Z_{G}(-u^2) = \!\!\!\!\sum_{\mcx\in\mathscr{C}^{\text{(even)}}_{G}}(-u^2)^{\abs{\partial\mcx}/2}2^{\abs{\mcx}}I_{G{\setminus}\Gamma[\mcx]}(-u^2)\prod_{C\in\mcx}h_C. \notag
    \end{equation}
\end{lemma}
We prove Lemma~\ref{lemma:generalizedcharacteristicpolynomial} in Appendix~\ref{section:GeneralizedCharacteristicPolynomialProof}. We also require the following lemma.
\begin{lemma}[{restate=[name=restatement]FundamentalIdentity}]
    \label{lemma:fundamentalidentity}
    Let $H$ be an SCF Hamiltonian with frustration graph $G$. Further, let $K_s$ be a simplicial clique and let $\chi$ be a simplicial mode with respect to $K_s$. Then
    \begin{equation}
        \begin{split}
            T_{G}(u) & \left(1+u\sum_{\bsj \in K_s}h_{\bsj}\right) \chi T_{G}(-u) \\
            &= Z_{G}(-u^2)\left(1-u\sum_{\bsj\in K_s}h_{\bsj}\right)\chi. \notag
        \end{split}
    \end{equation}
\end{lemma}
We prove Lemma~\ref{lemma:fundamentalidentity} in Appendix~\ref{section:FundamentalIdentityProof}. Lemma~\ref{lemma:fundamentalidentity} immediately shows that the incognito modes satisfy the eigenmode condition for the Hamiltonian over the subspace specified by $\mcj$.
\begin{lemma}
    \label{lemma:commutationrelation}
    Let $H$ be an SCF Hamiltonian with frustration graph $G$. The single-particle energies $\{\varepsilon_{\mcj,j}\}_{\mcj,j}$ and incognito modes $\{\psi_{\mcj,j}\}_{\mcj,j}$ satisfy
    \begin{equation}
        [H,\psi_{\mcj,\pm j}] = \pm2\varepsilon_{\mcj,j}\psi_{\mcj,\pm j}. \notag
    \end{equation}
\end{lemma}
\begin{proof} 
    We have
    \begin{align}
        [H,\psi_{\mcj,\pm j}] &= N_{\mcj,j}^{-1}\Pi_{\mcj}T(\mp u_{\mcj,j})[H,\chi]T(\pm u_{\mcj,j}) \notag \\
        &= 2N_{\mcj,j}^{-1}\Pi_{\mcj}T(\mp u_{\mcj,j})\left(\sum_{\bsj \in K_s}h_{\bsj}\chi\right)T(\pm u_{\mcj,j}). \notag
    \end{align}
    Applying Lemma~\ref{lemma:fundamentalidentity}, together with the definition of $u_{\mcj,j}$ that $\Pi_{\mcj}Z_{G}(-u_{\mcj,j}^2)=0$ gives
    \begin{align}
        [H,\psi_{\mcj,\pm j}] &= \pm\frac{2}{u_{\mcj,j}N_{\mcj,j}}\Pi_{\mcj}T(\mp u_{\mcj,j}) \chi T(\pm u_{\mcj,j}) \notag \\
        &= \pm2\varepsilon_{\mcj,j}\psi_{\mcj,\pm j}, \notag
    \end{align}
    completing the proof.
\end{proof}
We now show that the incognito modes $\{\psi_{\mcj,\pm j}\}_{\mcj,j}$ obey the canonical anticommutation relations.
\begin{lemma}[{restate=[name=restatement]LadderAnticommutationRelations}]
    \label{lemma:ladderanticommutationrelations}
    Let $H$ be an SCF Hamiltonian with frustration graph $G$. The incognito modes $\{\psi_{\mcj,\pm j}\}_{\mcj,j}$ satisfy the following anticommutation relations.
    \begin{equation}
        \{\psi_{\mcj,+j},\psi_{\mcj',-k}\} = \delta_{\mcj,\mcj'}\delta_{jk}\Pi_{\mcj}. \notag
    \end{equation}
\end{lemma}
We prove Lemma~\ref{lemma:ladderanticommutationrelations} in Appendix~\ref{section:LadderAnticommutationRelationsProof}. 
Finally, we show that we can write the Hamiltonian $H$ as a free-fermion Hamiltonian in the eigenmode basis.
\begin{lemma}[{restate=[name=restatement]DiagonalRelation}]
    \label{lemma:diagonalrelation}
    Let $H$ be an SCF Hamiltonian with frustration graph $G$. The single-particle energies $\{\varepsilon_{\mcj,j}\}_{\mcj,j}$ and incognito modes $\{\psi_{\mcj,j}\}_{\mcj,j}$ satisfy
    \begin{equation}
        H = \sum_{\mcj}\left(\sum_{j=1}^{\alpha(G)}\varepsilon_{\mcj,j}[\psi_{\mcj,+j},\psi_{\mcj,-j}]\right)\Pi_{\mcj}. \notag
    \end{equation}
\end{lemma}
We prove Lemma~\ref{lemma:diagonalrelation} in Appendix~\ref{section:DiagonalRelationProof}. Combining Lemma~\ref{lemma:commutationrelation}, Lemma~\ref{lemma:ladderanticommutationrelations}, and Lemma~\ref{lemma:diagonalrelation} proves Theorem~\ref{theorem:freefermionsolution}.

\section{Krylov Subspaces}
\label{section:KrylovSubspaces}

In this section, inspired by a well-known polynomial divisibility result of Godsil~\cite{godsil1993algebraic}, we connect the frustration graph formalism to the emerging body of work on operator Krylov subspaces~\cite{parker2019universal, rabinovici2021operator, moudgalya2022hilbert, moudgalya2022symmetries, caputa2022geometry}. By constructing a set of Krylov subspaces associated with path operators in the Hamiltonian, we are able to present an alternative solution method for SCF Hamiltonians, which we expect to provide a strategy for applying graph theory to more general models. For our purposes, the alternative method will also provide a better physical intuition for the free-fermion modes in the Pauli basis where $H$ is defined. The technical content of the method is captured by the following theorem.
\begin{theorem}[Polynomial Divisibility]
    \label{theorem:polynomialdivisibility}
    There exists a real matrix $\mathbf{A}$ with elements indexed by induced paths in $\gs$ such that
    \begin{equation}
        \Pi_{\mcj}\adh{\chi}{k} = (-2i)^k\sum_{P\in\mcp}\left(\mathbf{A}_{G,\mcj}^k\right)_{\{\js\},P}h_P, \notag
    \end{equation}
    over each mutual eigenspace of the $\{\jkg{C_0}{G}\}$, where $\chi$ is a simplicial mode corresponding to the vertex $\js$ not in $V$, and we define $\adh{\chi}{}=[iH,\chi]$. The matrix $\mathbf{A}_{G,\mcj}$ denotes the weighted adjacency matrix of a directed bipartite graph.
\end{theorem}

\subsection{Implications of Polynomial Divisibility}
\label{subsection:ImplicationsPolynomialDivisibility}

Before continuing to the proof, let us elaborate on some implications of Theorem~\ref{theorem:polynomialdivisibility}. The theorem appears very holographic in the sense that commutation with the Hamiltonian only changes path operators at the endpoints. This is entirely due to the fact that the frustration graph is claw-free. We make this precise by showing that the theorem implies a set of fermion modes given by repeated commutators with the Hamiltonian.
\begin{corollary}
    \label{corollary:modifiedanticommutationrelation}
    The operators generated by repeated commutation with $H$ satisfy
    \begin{equation}
        \Pi_{\mcj}\{\adh{\chi}{j},\adh{\chi}{k}\}
        = 2\left(\mathbf{M}_{G,\mcj}\right)_{jk}\Pi_{\mcj}, \notag
    \end{equation}
    where the matrix $\mathbf{M}_{G,\mcj}$ is real symmetric.
\end{corollary}
Let us restrict to a particular subspace labeled by $\mcj$ implicitly, allowing us to drop this label as well as the factor of $\Pi_{\mcj}$ in what follows. We additionally drop the label $G$ unless specifying a particular induced subgraph.
\begin{proof}
    Our proof is by induction on $j$. Note that, for all $k$,
    \begin{align}
        \{\adh{\chi}{0},\adh{\chi}{k}\} &= \{\chi,\adh{\chi}{k}\} \notag \\
        &= (-i)^k2^{k+1}\left(\mathbf{A}^k\right)_{\{\js\},\{\js\}}I,
    \end{align}
    by applying Theorem~\ref{theorem:polynomialdivisibility} to $\adh{\chi}{k}$ together with the fact that the only operator in the sum that does not anticommute with $\chi$ is itself. Thus, the corollary holds for $j=0$ with
    \begin{align}
        M_{0,k} = (-2i)^k\left(\mathbf{A}^k\right)_{\{\js\},\{\js\}}.
    \end{align}
    Now assume that
    \begin{equation}
        \{\adh{\chi}{\ell},\adh{\chi}{k}\} = 2M_{\ell,k}I,
    \end{equation}
    for all $k$, some matrix $\mathbf{M}$, and all $\ell\in\{0,\dots,j-1\}$. Then
    \begin{equation}
        \begin{split}
            &[\adh{\chi}{j},\adh{\chi}{k}] = [\adh{\left(\adh{\chi}{j-1}\right)}{},\adh{\chi}{k}] \\
            &\quad= [iH, [\adh{\chi}{j-1},\adh{\chi}{k}]]-[\adh{\chi}{j-1},\adh{\chi}{k+1}] \\ 
            &\quad= -2M_{j-1, k+1}I-2\adh{\chi}{k}\adh{\chi}{j},
        \end{split}
        \label{equation:jacobistep}
    \end{equation}
    where the second line follows by the Jacobi identity and the third line follows from applying the inductive hypothesis with the identity 
    \begin{align}
        [A,B] = \{A,B\}-2BA, \label{equation:commutatoranticommutator}
    \end{align}
    and canceling terms. Applying Eq.~(\ref{equation:commutatoranticommutator}) again to Eq.~(\ref{equation:jacobistep}) gives
    \begin{equation}
        \{\adh{\chi}{j},\adh{\chi}{k}\} = -2M_{j-1,k+1}I.
    \end{equation}
    This shows the first part of the corollary. By solving the recursion relation for $\mathbf{M}$, we have
    \begin{align}
        \{\adh{\chi}{j},\adh{\chi}{k}\} &= 2(-i)^jM_{0,j+k}I \notag \\
        &= (-i)^{j+k}2^{j+k+1}\left(\mathbf{A}^{j+k}\right)_{\{\js\},\{\js\}}I. \notag
    \end{align}
    This gives
    \begin{equation}
        M_{jk} = (-2i)^{j+k}\left(\mathbf{A}^{j+k}\right)_{\{\js\},\{\js\}},
    \end{equation}
    so $\mathbf{M}$ is real symmetric.
\end{proof}

While we may repeatedly take commutators with $H$, the fact that the set of induced paths of $G$ is finite implies that there is a minimal rank $r$ such that $\adh{\chi}{r}$ is a linear combination of the elements of $\{\adh{\chi}{j}\}_{j=0}^{r-1}$.
Suppose we have
\begin{align}
    \adh{\chi}{r} = \sum_{k=0}^{r-1}v_k\adh{\chi}{k}.
\end{align}
This gives
\begin{align}
    \{\adh{\chi}{j},\adh{\chi}{r}\} &= \sum_{k=0}^{r-1}v_k\{\adh{\chi}{j},\adh{\chi}{k}\} \\
    &= 2\left(\mathbf{M}\cdot\mathbf{v}\right)_{j}I.
\end{align}
Thus, it suffices to consider only the spanning set given by the elements of $\{\adh{\chi}{j}\}_{j=0}^{r-1}$, and we take $\mathbf{M}$ to be an $r \times r$ matrix. The value of this rank is typically much lower than the number of induced paths in $G$. By the definition of $\mathbf{M}$, we have
\begin{align}
    \left(\mathbf{v}^{\text{T}}\cdot\mathbf{M}\cdot\mathbf{v}\right) &= \frac{1}{2d}\sum_{j,k=0}^{r-1}v_jv_k\tr(\{\adh{\chi}{j},\adh{\chi}{k}\}) \\
    &= \frac{1}{d}\tr\left[\left(\sum_{j=0}^{r-1}v_j\adh{\chi}{j}\right)^2\right],
\end{align}
where $d$ is the Hilbert-space dimension. Since this is the trace of the square of a non-zero Hermitian matrix, it is positive. Thus, $\mathbf{M}$ is positive definite. This allows us to define a set of physical Majorana modes $\{\gamma_j\}_{j=0}^{r-1}$ by diagonalizing $\mathbf{M}$ as
\begin{equation}
    \mathbf{M} = \mathbf{U}^{\text{T}}\mathbf{D}\mathbf{U},
\end{equation}
where $\mathbf{U}$ is an orthogonal matrix and $\mathbf{D}$ is a diagonal matrix with positive elements $D_{jj}=\lambda_j$ along the main diagonal. We now define
\begin{equation}
    \gamma_j = i^{(j\bmod2)}\lambda_j^{-1/2}\sum_{k=0}^{r-1}U_{jk}\adh{\chi}{k}.
\end{equation}
These operators are Hermitian and satisfy the canonical anticommutation relations for fermions
\begin{equation}
    \{\gamma_j,\gamma_k\} = 2\delta_{jk}I.
\end{equation}
As a result, they are trace-orthogonal, i.e.,
\begin{equation}
    d^{-1}\tr(\gamma_j\gamma_k) = \frac{1}{2d}\tr(\{\gamma_j,\gamma_k\}) = \delta_{jk}.
\end{equation}
We can write the elements of an effective single-particle Hamiltonian by
\begin{equation}
    h_{jk} = -\frac{i}{2d}\tr[\gamma_j\operatorname{ad}_H{(\gamma_k)}{}]
\end{equation}
Formally, $\mathbf{h}$ is given by the inverse transformation relating $\{\gamma_j\}_{j=0}^{r-1}$ to $\{\adh{\chi}{j}\}_{j=0}^{r-1}$ applied to the companion matrix for $\operatorname{ad}_{iH}$ on the cyclic subspace generated by $\chi$.

\subsection{Matrix Elements and Pauli Realization}
\label{subsection:MatrixElementsPauliRealization}

While the only specific property of $\mathbf{A}$ that we have relied on is the fact that $\mathbf{A}$ is the weighted adjacency matrix of a directed bipartite graph, it will be helpful to propose a specific form here. In Def.~\ref{definition:directedhoppinggraph} we defined a directed bipartite graph, by replacing the edges of an induced path tree with directed arcs and adding arcs corresponding to generalized cycles. Choosing an induced path tree with respect to a simplicial clique $K$ (Def.~\ref{definition:inducedpathtreeclique}), we can define the elements $A_{PP'}$ for $(P,P') \in E_{\Lambda(G)}$ as
\begin{equation}
    A_{PP'} =
    \begin{cases}
        1 & \abs{P'}>\abs{P}, \\
        b_{\bsj}^2 & P=P{\setminus}\{\bsj\}, \\
        \frac{\mcj_{\avg{C_0}}}{N_{P,P',\avg{C_0}}} & \text{hoop in $\avg{C_0}$}, \\
        0 & \text{otherwise}.
    \end{cases}
    \label{eq:decoratedtree}
\end{equation}
The normalization factor $N_{P,P',\avg{C_0}}$ is chosen such that $\frac{1}{2}\sum_{P'}N_{P,P',\avg{C_0}}^{-1}=1$. Each of the non-zero cases corresponds to a particular non-vanishing contribution from an additional application of $\operatorname{ad}_{iH}$ in a fixed mutual subspace of the generalized cycle symmetries. In the first two cases, the induced path can transition to an adjacent induced path in the induced path tree. In the last case, the induced path wraps around an even hole, and this even-hole part contributes a factor of the generalized cycle eigenvalue in the given subspace. There is an additional factor of two due to the fact that the path can wrap around the hole in either direction.

To prove Theorem~\ref{theorem:polynomialdivisibility}, we shall also consider a particular Pauli representation of $\gs$, and we shall prove particular properties of $\gs$ using this representation. Since the result will be a property of $\gs$ alone, it will not depend on the representation, and we can conclude that it holds for all representations of $\gs$. We choose the representation to have the property that
\begin{equation}
    \frac{1}{d}\tr(h^{\dagger}_Ph_{P'}) = \delta_{PP'},
\end{equation}
for any pair of induced paths $P,P'\in\mcp_{\gs}$. For an explicit instance of such a representation, we can take the \emph{fiducial bosonization} from Ref.~\cite{chapman2022free}. In this representation, we assign a qubit to each edge $e = \{\bsj,\bsk\} \in E_{G}$ of the frustration graph. Without loss of generality, we choose one of the terms from $h_{\bsj}$ and $h_{\bsk}$ to act on this qubit as $\sigma_e^z$, and we let the other term act as $\sigma_e^x$ (e.g. $h_{\bsj}$ acts as $\sigma_e^z$ and $h_{\bsk}$ acts as $\sigma_e^x$). Additionally the only terms acting on the qubit corresponding to edge $e$ are $h_{\bsj}$ and $h_{\bsk}$. Thus, $h_P$ is the only induced-path operator acting as $\sigma^y$ only on the qubits corresponding to the edges in $E[P]$, so it satisfies the property.

\subsection{Deformation Closure of an Induced Path}
\label{subsection:DeformationClosureInducedPath}

We now prove the following lemma regarding paths within the same deformation closure.
\begin{lemma}
    \label{lemma:deformationclosurepathsrelation}
    Let $P=\js\md\bsj_1\md\dots\md\bsj_{\ell}$ be an induced path with $\ell\geq2$. Then
    \begin{equation}
        \left(\mathbf{A}^k\right)_{\{\js\},P} = \left(\mathbf{A}^k\right)_{\{\js\},\wt{P}}, \notag
    \end{equation}
    for all $\wt{P}\in\avg{P}$ and all $k\geq0$.
\end{lemma}
\begin{proof}
    It is sufficient to prove the lemma for $\wt{P}$ given by a single-vertex deformation of $P$. By convention, we take $\bsj_0=\js$, and let
    \begin{equation}
        \wt{P} = (P{\setminus}\{\bsj_i\})\cup\{\bsk\},
    \end{equation} 
    for $\bsk \prec_{P} \bsj_i$, with $i\in\{1,\dots,\ell-1\}$. We consider $\left(\mathbf{A}^k\right)_{P,P'}$ as the weighted sum of walks from $P$ to $P'$ on $\Lambda(\gs)$ in $k$ steps. Let $P_m=\js\md\bsj_1\md\dots\md\bsj_m$ for $m\leq\ell$, with $P_{\ell}=P$, and define
    \begin{equation}
        \widetilde{P}_m =
        \begin{cases}
            P_m & m < i, \\
            (P_m{\setminus}\{\bsj_i\})\cup\{\bsk\} & m \geq i.
        \end{cases}
    \end{equation}
    We first show that, if a weighted arc $P_r \to P_s$ is present in $\Lambda(\gs)$ for $r,s\in\{i,i+1,\dots,\ell\}$, then there is a corresponding arc $\widetilde{P}_r \to \widetilde{P}_s$ with the same weight present as well. Suppose $P_r$ and $P_s$ are neighboring in $\idt{}{\gs}$. If $s>r$, then $P_r \subset P_s$ with $P_s{\setminus}P_r=\{\bsj_s\}$. Since $r \geq i$, then $s \geq i+1$, so $\wt{P}_r \subset \wt{P}_s$ with $\wt{P}_s{\setminus}\wt{P}_r=\{\bsj_s\}$. Thus, there is an arc $P_r \to P_s$ with weight $1$ present in $\Lambda(\gs)$ in this case. If $s<r$, then $P_s \subset P_r$ with $P_r{\setminus}P_s=\{\bsj_r\}$. Since $s \geq i$, we have $r \geq i+1$, so $\wt{P}_r \subset \wt{P}_s$ with $\wt{P}_s{\setminus}\wt{P}_r=\{\bsj_r\}$. Thus, there is an arc $P_r \to P_s$ with weight $b_{\bsj_r}^2$ present in $\Lambda(\gs)$ in this case. If $\{P_r,P_s\}$ is not in $E_{\idt{}{\gs}}$, then there is a vertex $\bss$ in $\Gamma(\bsj_r){\setminus}\Gamma[\bsj_{r-1}]$ such that $P_r \cup \{\bss\} = P_s \cup C$ with $C\in\avg{C_0}$. Restricting to this case, if $s=i$, then we have that $\bss$ and $\bsk$ are neighboring. If this is not the case, then $\{\bsj_{i+1},\bss,\bsk,\bsj_{i+2}\}$ induces a claw in $G$. Thus, we have $\Gamma_{\widetilde{P}_r}(\bss)=\{\bsj_r,\bsk,\bsj_{i+1}\}$ and so $\widetilde{P}_r \cup \{\bss\} = \widetilde{P}_s \cup C$ with $C\in\avg{C_0}$. If $s>i$, then $\widetilde{P}_r \cup \{\bss\} = \widetilde{P}_s \cup C$ with $C \in \avg{C_0}$. Thus, there is an arc $\widetilde{P}_r\to\widetilde{P}_s$ with the same weight as that from $P_r \to P_s$.

    Now, consider the weighted sum of walks from $\{\js\}$ to $P$ on $\Lambda(\gs)$ in $k$ steps. Since every walk in this sum ends at $P$, there is a step $m \leq k$, that is the last step in which the arc $P_{i-1} \to P_i$ is traversed. After this, no arcs $P_{r} \to P_{s}$ can be traversed where $s<i$. Otherwise, the walk would traverse the arc $P_{i-1} \to P_i$ again. Thus, all arcs $P_{r} \to P_{s}$ traversed after step $m$ have ${r,s\in\{i,i+1,\dots,\ell\}}$. Since $A_{P_{i-1},P_i}=A_{P_{i-1},\wt{P}_i}=1$, we substitute $A_{P_{i-1},P_i} \to A_{P_{i-1},\wt{P}_i}$ at step $m$ and ${A_{P_r,P_s} \to A_{\wt{P}_r,\wt{P}_s}}$ thereafter. This gives a walk with an equal weight that ends at $\wt{P}$, and since this substitution can be performed for each term in the sum, we have
    \begin{equation}
        \left(\mathbf{A}^k\right)_{\{\js\} P} = \left(\mathbf{A}^k\right)_{\{\js\},\wt{P}},
    \end{equation}
    completing the proof.
\end{proof}

We note that a single-vertex deformation of a path $P$ is a special case of a bubble wand with the hoop being a cycle of length three. Thus, we expect there to be a similar result, regarding bubble wands when the hoops are even holes and describe symmetries of the model, which we shall see in the following section.

\subsection{Proof of Polynomial Divisibility}
\label{subsection:PolynomialDivisibilityTheoremProof}

We now prove Theorem~\ref{theorem:polynomialdivisibility}.

\begin{proof}[Proof of Theorem~\ref{theorem:polynomialdivisibility}]
    Our proof is by induction on the power $k$ of $\operatorname{ad}_{iH}$. Clearly, we have
    \begin{equation}
        \adh{\chi}{0} = \chi = \sum_{P\in\mcp}\left(\mathbf{A}^0\right)_{\{\js\},P}h_P, \\
    \end{equation}
    and
    \begin{equation}
        \adh{\chi}{} = -2i\sum_{\bsj \in K_s} \chi h_{\bsj} = -2i\sum_{P\in\mcp} \left(\mathbf{A}\right)_{\{\js\},P}h_P,
    \end{equation}
    so the theorem holds for $k\in\{0,1\}$. Now assume it is true for all powers $m<k$, i.e.,
    \begin{equation}
        \adh{\chi}{m}=\sum_{P\in\mcp}(\mathbf{A}^m)_{\{\js\},P}h_P,
    \end{equation}
    for $m<k$. Now take
    \begin{align}
        \adh{\chi}{k} &= [iH,\adh{\chi}{k-1}] \\
        &= -\frac{1}{2}(-2i)^{k}\sum_{P\in\mcp}(\mathbf{A}^{k-1})_{\{\js\},P}[iH,h_P] \\
        &= (-2i)^{k}\!\!\!\!\sum_{\substack{P\in\mcp \\ P=\js\md\dots\md\bsj_{\ell}}}\!\!\!\!(\mathbf{A}^{k-1})_{\{\js\},P}b^2_{\bsj_{\ell}} h_{P{\setminus}\bsj_{\ell}} \\
        &\quad+ \sum_{P\in\mcp}(\mathbf{A}^{k-1})_{\{\js\},P}\!\!\!\!\sum_{\substack{\bsj \notin P \\ \scomm{h_{\bsj}}{h_P}=-1}}\!\!\!\!h_Ph_{\bsj},
    \end{align}
    where the third line follows by expanding the commutator and collecting terms according to whether $\bsj$ is in $P$. Note that in the former case, the only operator $h_{\bsj}$ with $\bsj$ in $P$ that anticommutes with $h_P$ is $h_{\bsj_{\ell}}$. By Lemma~\ref{lemma:vertexcyclerelations}, there are three cases whereby $h_{\bsj}$ can anticommute with $h_P$ for $\bsj$ not in $P$. In case (b.i), where $\bsj\prec_{P}\bsk$, we have a unique additional term corresponding to $\wt{P}=(P{\setminus}\bsk)\cup\{\bsj\}$ and $\bsk$ not in $\wt{P}$, which cancels the term corresponding to $P$ and $\bsj$ as
    \begin{align}
        &\left(\mathbf{A}^{k-1}\right)_{\{\js\},P}h_Ph_{\bsj}+\left(\mathbf{A}^{k-1}\right)_{\{\js\}, \wt{P}}h_{\wt{P}}h_{\bsk} \notag \\
        &\quad= \left[\left(\mathbf{A}^{k-1}\right)_{\{\js\},P}- \left(\mathbf{A}^{k-1}\right)_{\{\js\}, P} \right] h_P h_{\bsj} \\
        &\quad= 0.
    \end{align}
    In case (b.iv), we have that $\Gamma_{P}(\bsj)=\{\bsj_s,\bsj_{s+1},\bsj_{\ell}\}$ for ${0 \leq s < \ell-2}$ (if $s=\ell-2$, then we again have case (b.i)). Let $P'=\js\md\bsj_1\md\dots\md\bsj_{s}$, then we have that ${C=(P{\setminus}P')\cup\{\bsj\}}$ gives a hole in $G$. Next, let ${P^*=(P\cup\{\bsj\}){\setminus}\{\bsj_{s+1}\}}$ be the path traversing $C$ in the opposite direction to $P$. We shall show that 
    \begin{equation}
        \left(\mathbf{A}^m\right)_{\{\js\},P} = \left(\mathbf{A}^m\right)_{\{\js\},P^*}, \label{equation:pathcircle}
    \end{equation}
    for all powers $m<k$. We have
    \begin{equation}
        \left(\mathbf{A}^m\right)_{\{\js\},P} = \!\!\!\!\sum_{\mathbf{P}\in\mcp_{G}^{\times m-1}}\!\!\!\!A_{\{\js\},P^{(1)}}A_{P^{(1)},P^{(2)}} \dots A_{P^{(m-1)},P}. \notag
    \end{equation}
    Note that the summand is only non-zero when ${\mathbf{P}=\{P^{(1)},\dots,P^{(m-1)}\}}$ is a walk on the directed graph with weighted adjacency matrix $\mathbf{A}$. Let us now group terms in this sum by walks which pass $P'$ for the last time at step $g$. This gives
    \begin{equation}
        \begin{split}
            \!\left(\mathbf{A}^m\right)_{\{\js\},P} &= \sum_{g=1}^m\sum_{\mathbf{P}\in\mcp_{G}^{\times g-1}}A_{\{\js\},P^{(1)}} \dots A_{P^{(g-1)},P'}\! \\
            &\quad\times \left(\mathbf{A}_{G{\setminus}\Gamma[P']}^{m-g}\right)_{\{\bsj_{s+1}\},(P{\setminus}P')},
        \end{split}
        \label{equation:pathgroupterms}
    \end{equation}
    where $\mathbf{A}_{G{\setminus}\Gamma[P']}$ is the weighted adjacency matrix corresponding to the graph $G{\setminus}\Gamma[P']$ in the natural way. 
    It is a result of Ref.~\cite{chudnovsky2007roots} that $h_{\bsj_{s+1}}$ is a simplicial mode for this graph, and so this matrix is well defined.
    
    First, note that, strictly speaking, we cannot have $g=m$ by the requirement that $s<\ell-2$ and the fact that the length of a path can increase by at most one at any step in the walk.
    This implies that the largest $g$ can be is $m-3$, but we shall include all values of $g$ that are not obviously forbidden in Eq.~(\ref{equation:pathgroupterms}) with the understanding that there are no terms in the sum when $g$ is too large. Next, we shall see that once the walk passes through $P'$ for the last time, it must immediately pass through $P'\cup\{\bsj_{s+1}\}$, or it will eventually have to return to $P'$ again to proceed to $P$. Since $\bsj$ neighbors $\bsj_s$ and $\bsj_{s+1}$, if the walk passes through any path containing $\bsj$ from this point, it will again return to $P'$ by the definition of $\mathbf{A}$. Finally, if the walk passes through a path $P_r=\js\md\bsj_1\md\dots\md\bsj_r$ for $r \leq s$, then it must pass through $P'$ again.
    Thus, our requirement that the walk passes through $P'$ for the last time at step $g$ gives Eq.~(\ref{equation:pathgroupterms}). To prove Eq.~(\ref{equation:pathcircle}), it is sufficient to prove
    \begin{equation}
        \!\left(\mathbf{A}_{G{\setminus}\{\bsj,\bsk\}}^m\right)_{\{\bsj\},(P{\setminus}\{\js\})} \!=\! \left(\mathbf{A}_{G{\setminus}\{\bsj, \bsk\}}^{m}\right)_{\{\bsk\},(P^*{\setminus}\{\js\})}\!,
        \label{equation:patheffectivecircle}
    \end{equation}
    for $\bsj,\bsk \in K_s$ with the simplicial mode corresponding to vertex $\js$, and $P^*=(P \cup\{\bsk\}){\setminus}\{\bsj\}$ with $(P{\setminus}\{\js\})\cup\{\bsk \}=(P^*{\setminus}\{\js\})\cup\{\bsj\}\in\mcc_G$. We thus proceed to prove Eq.~(\ref{equation:patheffectivecircle}).
    
    Letting $H'=H-h_{\bsj}-h_{\bsk}$, we apply the inductive hypothesis and assume the fiducial bosonization of Ref.~\cite{chapman2022free} to obtain
    \begin{equation}
        \!\left(\mathbf{A}_{G{\setminus}\{\bsj,\bsk\}}^m\right)_{\{\bsj\},(P{\setminus}\{\js\})} \!=\! \frac{(-2)^{-m}}{d\prod_{\bsu \in P}b_{\bsu}^2}\tr\left[h^{\dagger}_P \operatorname{ad}_{H'}^m\left(\chi h_{\bsj}\right)\right]\!, \notag
    \end{equation}
    where
    \begin{align}
        h^{\dagger}_P &= \left(\prod_{\bsu \in C} b_{\bsu}^2\right)^{-1}\left(h_C^{\dagger}h_C\right)h^{\dagger}_{P}, \\
        h^{\dagger}_P &= \frac{1}{b_{\bsk}^2}h^{\dagger}_C(h_{\bsk}\chi).
    \end{align}
    Here, $C=P\cup\{\bsk\}=P^*\cup\{\bsj\}$. Note that the factors in $h_C$ have cyclic ordering, so that the corresponding factors in $h^{\dagger}_P$ are canceled. This gives
    \begin{align}
        &\left(\mathbf{A}_{G{\setminus}\{\bsj,\bsk\}}^{m}\right)_{\{\bsj\},(P{\setminus}\{\js\})} \notag \\
        &\quad= \frac{(-2)^{-m}}{d\prod_{\bsu \in C}b_{\bsu}^2}\tr\left[h_C^{\dagger} \left(h_{\bsk}\chi\right)\operatorname{ad}_{H'}^m\left(\chi h_{\bsj}\right)\right] \\
        &\quad= \frac{2^{-m}}{d\prod_{\bsu \in C}b_{\bsu}^2}\tr\left\{\left(h_{\bsj} \chi \right)\operatorname{ad}_{H'}^m\left[h_C^{\dagger}\left(\chi h_{\bsk}\right)\right]\right\},
    \end{align}
    where we have applied the identity
    \begin{equation}
        \tr\left(X\adh{Y}{}\right) = -\tr\left(Y\adh{X}{}\right), 
    \end{equation}
    $m$ times in succession. We next sum over all deformations of $P$ with the coefficient corresponding to $C$ to obtain
    \begin{align}
     &\sum_{\substack{P\in\avg{P} \\ C=P\cup\{\bsk\}}}\left(\prod_{\bsu \in C} b_{\bsu}^2\right)\left(\mathbf{A}_{G{\setminus}\{\bsj, \bsk\}}^{m}\right)_{\{\bsj\},(P{\setminus}\{\js\})} \notag \\ 
        &\quad= \frac{2^{-m}}{d}\tr\left\{\left(h_{\bsj}\chi\right)\operatorname{ad}_{H'}^m\left[\jkg{C}{G}\left(\chi h_{\bsk}\right)\right]\right\} \notag \\
        &\quad= \frac{2^{-m}}{d}\tr\left\{\left(h_{\bsj}\chi\right)\jkg{C}{G}\operatorname{ad}_{H'}^m\left[\left(\chi h_{\bsk}\right)\right]\right\} \notag \\
        &\quad= \frac{2^{-m}}{d}\sum_{\substack{P\in\avg{P} \\ C=P\cup\{\bsk\}}}\tr\left\{\left(h_{\bsj}\chi\right)h^{\dagger}_C\operatorname{ad}_{H'}^m\left(\chi h_{\bsk}\right)\right\} \notag \\
        &\quad= (-1)^{m+\abs{C}}\!\!\!\!\sum_{\substack{P\in\avg{P} \\ C=P\cup\{\bsk\}}}\!\!\!\!\left(\prod_{\bsu \in C}b_{\bsu}^2\right)\left(\mathbf{A}_{G{\setminus}\{\bsj,\bsk\}}^{m}\right)_{\{\bsk\},(P^*{\setminus}\{\js\})}\!\!. \notag
    \end{align}
    Note that $m+\abs{C}$ must be even in order for the walk to reach either $P{\setminus}\{\js\}$ or $P^*{\setminus}\{\js\}$. Applying Lemma~\ref{lemma:deformationclosurepathsrelation}, we have that the matrix amplitude in the sum is a constant, and so we obtain Eq.~(\ref{equation:patheffectivecircle}). By pairing corresponding walks to $P$ and $P^*$ with the same first $g-1$ steps in Eq.~(\ref{equation:pathgroupterms}), we have Eq.~(\ref{equation:pathcircle}). Therefore, in case (b.iv), the contributions from odd holes cancel, and the contributions from even holes add. This completes the proof.
\end{proof}

\section{Application of Results}
\label{section:ApplicationResults}

We now apply our formalism to a two dimensional model. The system is supported on a two-dimensional square lattice, with five qubits located on the links of the lattice, such that there is a spin at each of the positions $(j+\frac{\alpha}{6},k)$ and $(j,k+\frac{\alpha}{6})$ for $\alpha=\{1,2,3,4,5\}$. Due to the symmetry of the lattice we collect the terms along each link of the lattice in the Hamiltonian by their coupling strength and label each term as $h_\mu$ for $\mu\in\{a,b,c,d,e,f,g,h\}$. The full Hamiltonian is given by
\begin{equation}
    \begin{split}
        H = &\sum_{j,k}\left[a\left(\sigma^y_{j+\frac{2}{6},k}\sigma^x_{j+\frac{1}{6},k}+\sigma^y_{j,k+\frac{2}{6}}\sigma^x_{j,k+\frac{1}{6}}\right)\right. \\
        &\quad+ b\left(\sigma^x_{j+\frac{2}{6},k}\sigma^y_{j-\frac{1}{6},k}+\sigma^x_{j,k+\frac{2}{6}}\sigma^y_{j+\frac{1}{6},k}\right) \\
        &\quad+ c\left(\sigma^z_{j+\frac{2}{6},k}\sigma^y_{j+\frac{3}{6},k}+\sigma^z_{j,k+\frac{2}{6}}\sigma^y_{j,k+\frac{3}{6}}\right) \\
        &\quad+ d\left(\sigma^z_{j+\frac{2}{6},k}\sigma^z_{j+\frac{3}{6},k}+\sigma^z_{j,k+\frac{2}{6}}\sigma^z_{j,k+\frac{3}{6}}\right) \\
        &\quad+ e\left(\sigma^x_{j+\frac{3}{6},k}\sigma^z_{j+\frac{4}{6},k}+\sigma^x_{j,k+\frac{3}{6}}\sigma^z_{j,k+\frac{4}{6}}\right) \\
        &\quad+ f\left(\sigma^y_{j+\frac{3}{6},k}\sigma^z_{j+\frac{4}{6},k}+\sigma^y_{j,k+\frac{3}{6}}\sigma^z_{j,k+\frac{4}{6}}\right) \\
        &\quad+ g\left(\sigma^x_{j+\frac{4}{6},k}\sigma^y_{j+1,k+\frac{1}{6}}+\sigma^x_{j,k+\frac{4}{6}}\sigma^y_{j-\frac{1}{6},k}\right) \\
        &\quad+ \left.h\left(\sigma^y_{j+\frac{4}{6},k}\sigma^x_{j+\frac{5}{6},k}+\sigma^y_{j,k+\frac{4}{6}}\sigma^x_{j,k+\frac{5}{6}}\right)\right].
    \end{split}
    \label{equation:applicationhamiltonian}
\end{equation}

\begin{figure}[ht!]
    \centering
    \includegraphics[width=0.45\textwidth]{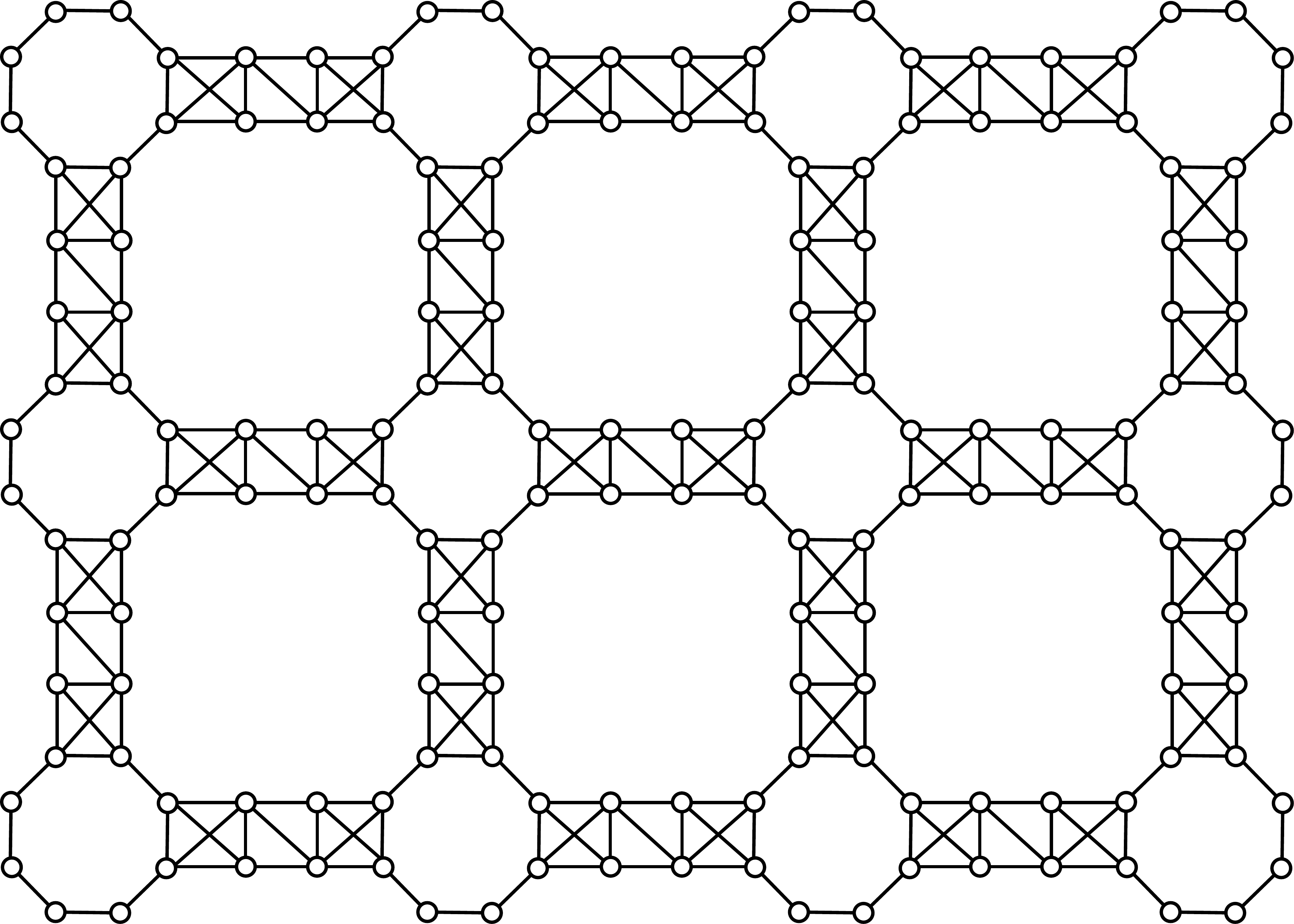}
    \caption{The frustration graph for the two-dimensional model. The frustration graph is simplicial and claw-free. The model is neither a line graph nor (even-hole, claw)-free and so exists beyond the scope of Refs.~\cite{chapman2020characterization, elman2021free}}
    \label{figure:squareoctagon}
\end{figure}

The frustration graph of the model is shown in Fig.~\ref{figure:squareoctagon}. The frustration graph is not a line graph, thus admitting no obvious map to free fermions~\cite{chapman2020characterization}. Also, since the graph is two-dimensional the model necessarily contains even holes, and thus falls beyond the scope of Ref.~\cite{elman2021free}. Nevertheless, the frustration graph is claw-free and contains a simplicial clique, thus admitting a free-fermion solution of the form Eq.~(\ref{equation:firstresult}). The model was constructed by first designing a two-dimensional, claw-free, simplicial graph which is not a line graph and finding a qubitization of the model.
As we have shown, the free-fermion solution does not depend on this particular realization. We stress that, while the solution can be found by mapping the model to a line graph through a local unitary transformation, the model was not constructed to have this property.

For small system sizes it is possible to construct the full independence polynomial and generalized characteristic polynomial of the graph, from which the single-particle energies and fermionic modes can be extracted. However, in the thermodynamic limit this is not practical. It is perhaps more informative to use our knowledge of the existence of such a solution to construct a unitary (and therefore spectrum-preserving) transformation which maps the model from its present form to one whose frustration graph is a line graph.

The unit cell of the model contains sixteen Hamiltonian terms acting on ten qubits; however, there is a symmetry between the horizontal and vertical links of the lattice. It is sufficient therefore to consider the transformation on a single arm of the graph containing eight vertices (Hamiltonian terms) applied to five qubits; we thus denote the Hamiltonian on each of the arms of the graph as
\begin{equation}
    H_{\text{arm}}=h_a+h_b+h_c+h_d+h_e+h_f+h_g+h_h, \notag
\end{equation}
with each term associated to the Pauli realization given in Eq.~(\ref{equation:applicationhamiltonian}) in the natural way.

\begin{figure}[ht!]
    \begin{subfigure}{0.45\textwidth}
         \centering
         \includegraphics[width=0.9\columnwidth]{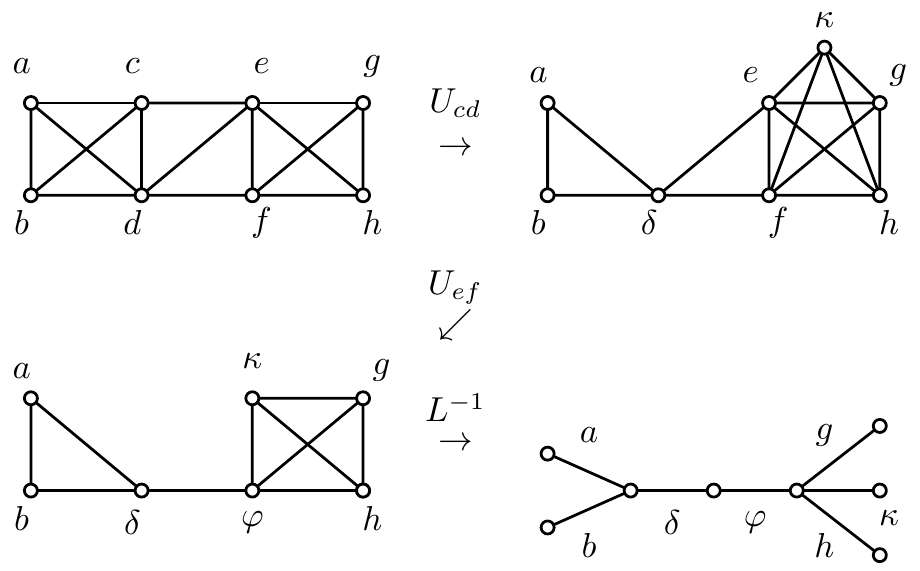}
         \caption{}
         \label{figure:squareoctagongadget}
     \end{subfigure}
     \begin{subfigure}{0.45\textwidth}
         \centering
         \includegraphics[width=0.9\columnwidth]{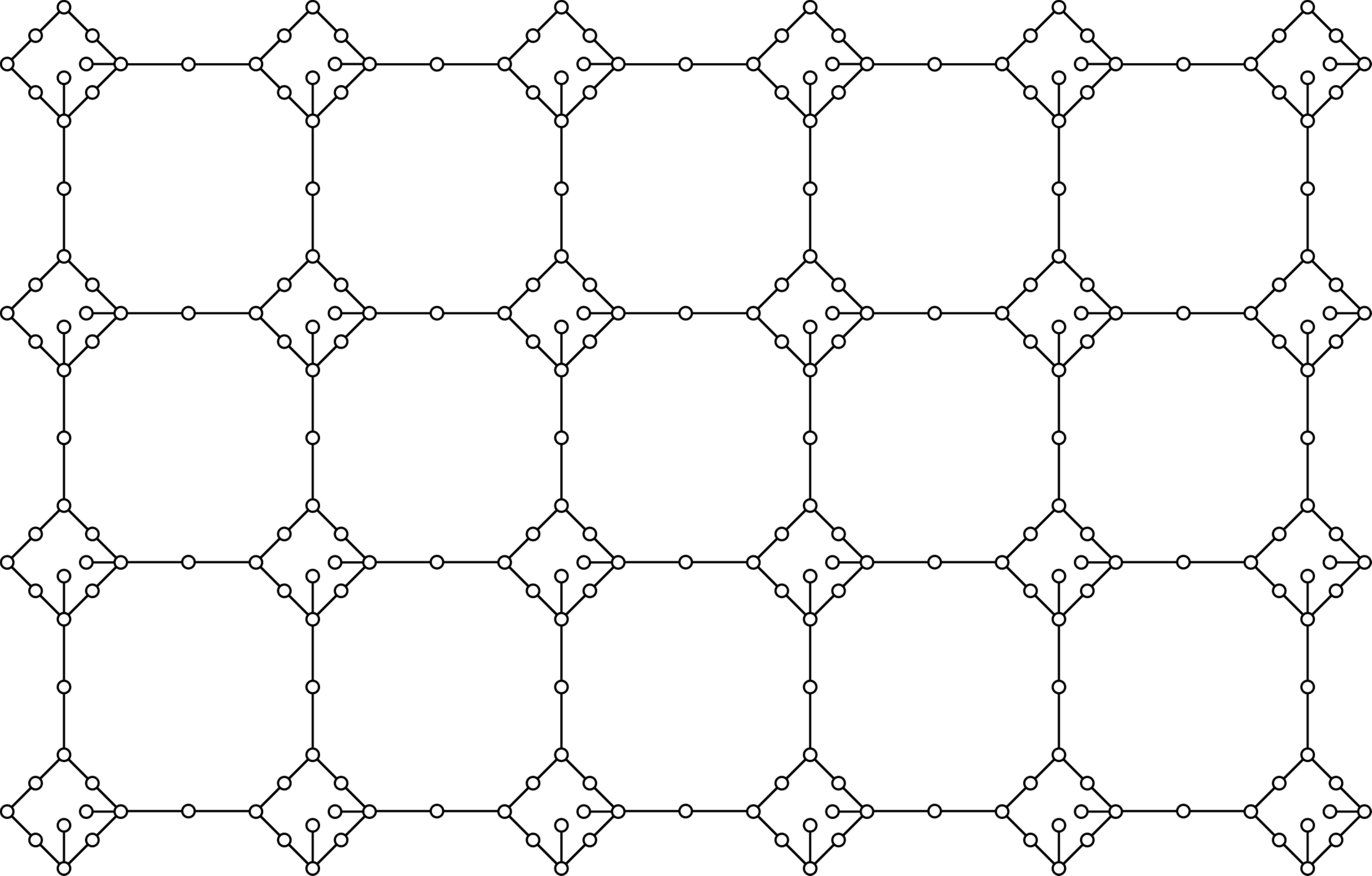}
         \caption{}
         \label{figure:squareoctagonrootgraph}
     \end{subfigure}
     \caption{A 2-dimensional simplicial claw-free frustration graph. (a) A graphical depiction of the unitary circuit to map the model as given to a model whose frustration graph is a line graph and so admits a Jordan-Wigner solution. In the last step, we show the generalized Jordan-Wigner solution applied to a single unit cell, though it is not always the case that such a solution extends globally. (b) The generalized Jordan-Wigner transformation on the entire model.}
\end{figure}
     
The transformation is depicted graphically in Fig.~\ref{figure:squareoctagongadget}. In each step the Hamiltonian is conjugated by a unitary generated by the product of a pair of terms, choosing the generating angle to remove one of the terms. This will generally introduce new terms in the Hamiltonian with interaction strengths depending on the rotation angle of our previous steps. By choosing our rotations appropriately, we can iterate this procedure until the frustration graph is a line graph. While our result guarantees that such a unitary circuit exists, in general cases it may be difficult to find in practice.

We begin by applying a unitary rotation to contract the edge between the vertices $c$ and $d$. This rotation is generated by the product of the Hamiltonian terms $h_c$ and $h_d$ in each of the halved unit cells so that
\begin{align}
    U_{cd} &= \prod_{j,k}e^{\theta h_ch_d} \notag \\
    &= \prod_{j,k}\exp\left(-i\theta\sigma^x_{j+\tfrac{3}{6},k}\right)\exp\left(-i\theta\sigma^x_{j,k+\tfrac{3}{6}}\right), \notag
\end{align}
with $\theta=\frac{1}{2}\arctan\left(\frac{c}{d}\right)$ and chosen such that ${U_{cd}(h_c+h_d)U_{cd}^\dagger=h_\delta}$ where $h_\delta = \frac{1}{c}\sqrt{c^2+d^2}h_d$. That is, the Hamiltonian term associated with the vertex $c$ is removed from the model. Note that vertices $a$, $b$, and $e$ each neighbor both $c$ and $d$, thus commuting with the unitary $U_{cd}$. However, vertex $f$ neighbors only $d$, so this rotation introduces an additional term to the Hamiltonian
\begin{equation}
    h_\kappa = 
    \begin{cases}
        f\sin(2\theta)\sigma^z_{j+\tfrac{3}{6},k}\sigma_{j+\tfrac{4}{6},k}^z &\ \text{on horizontal links}, \\
        f\sin(2\theta)\sigma^z_{j,k+\tfrac{3}{6}}\sigma_{j,k+\tfrac{4}{6}}^z &\quad \text{on vertical links},
    \end{cases} 
    \notag
\end{equation}
for each arm in the lattice via $U_{cd}h_fU_{cd}^\dagger=\cos(2\theta)h_f+h_\kappa$. The frustration graph of the rotated model is shown in the top right of Fig.~\ref{figure:squareoctagongadget}. This graph is one of the forbidden subgraphs of a line graph and so another rotation needs to be applied in order to find the hopping graph. We now apply a rotation to contract the edge between vertices $e$ and $f$. This rotation is given by the unitary matrix
\begin{align}
    U_{ef} &= \prod_{j,k}e^{\phi h_eh_f} \notag \\
    &= \prod_{j,k}\exp\left(-i\phi\sigma^z_{j+\tfrac{3}{6},k}\right)\exp\left(-i\phi\sigma^z_{j,k+\tfrac{3}{6}}\right), \notag
\end{align}
with $\phi=-\frac{1}{2}\arctan\left(\frac{f\cos(2\theta)}{e}\right)$ and chosen such that $U_{ef}(h_e+\cos(2\theta)h_f)U_{ef}^\dagger=h_\varphi$. We now see that the frustration graph shown in the bottom left of Fig.~\ref{figure:squareoctagongadget} is a line graph. Finding the root graph, as shown in the bottom right of Fig.~\ref{figure:squareoctagongadget}, gives the hopping graph of an arm of the lattice. Though it is not always the case that a local transformation will extend globally, in this case we have transformed the entire Hamiltonian into a line-graph free-fermion model with a local unitary circuit.

\begin{figure}[ht!]
    \begin{subfigure}{0.45\columnwidth}
         \centering
         \includegraphics[width=0.95\columnwidth]{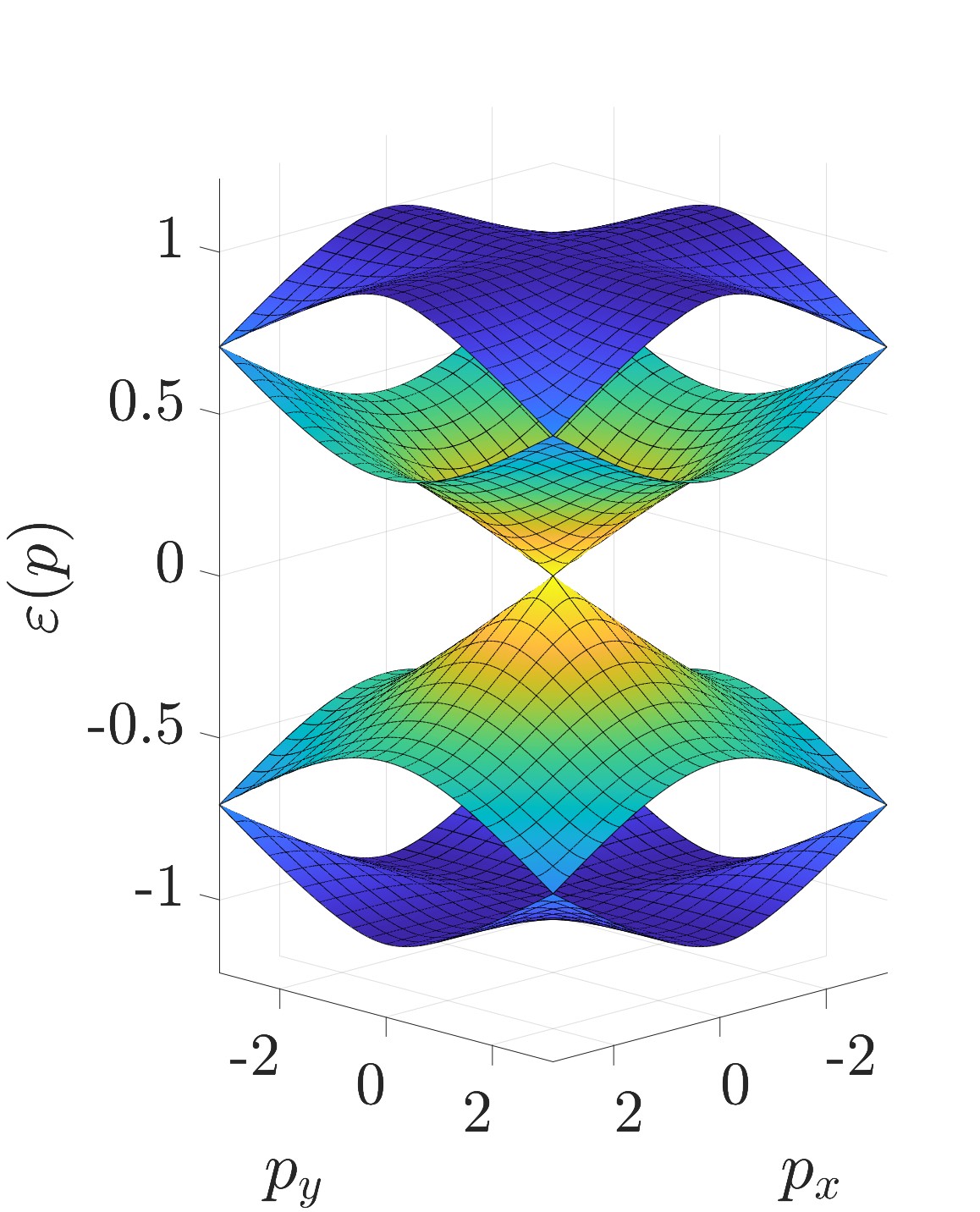}
         \caption{}
         \label{figure:squareoctagondispersionoff}
     \end{subfigure}
     \begin{subfigure}{0.45\columnwidth}
         \centering
         \includegraphics[width=0.95\columnwidth]{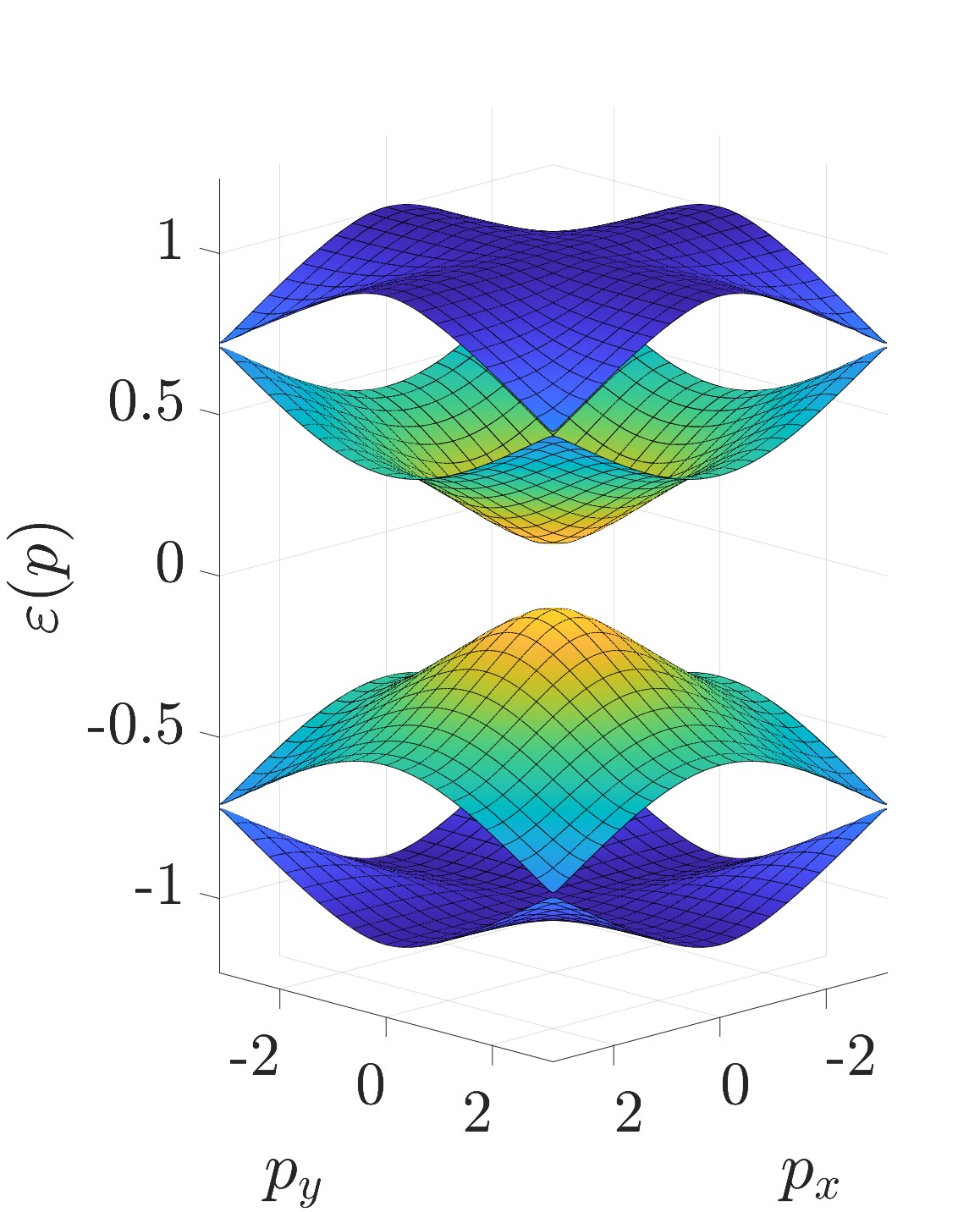}
         \caption{}
         \label{figure:squareoctagondispersionon}
     \end{subfigure}
     \caption{Dispersion relation $\varepsilon(p)$ for the two-dimensional model studied here with $\kappa$ (a) turned off and (b) turned on. When $h_{\kappa}$ is turned off, the model reverts to the same phase as the Kitaev honeycomb model as can be seen from the Dirac cone shape~\cite{kitaev2006anyons}. When $h_{\kappa}$ is turned on, a gap opens up and the other bands hybridize as in the periodic model.}
     \label{figure:squareoctagondispersion}
\end{figure}

With the Majorana hopping graph identified, it is then possible to numerically construct and plot the dispersion relation for the model. The dispersion relation for the model is shown in Fig.~\ref{figure:squareoctagondispersion}. Note that when $h_\kappa$ term is turned off (that is, when $f\sin\theta=0$), the dispersion relation has a conical shape and is critical at $p=(0,0)$ with critical exponent $z_c=1$. Thus, when $\kappa=0$, the model is in the same phase as the Kitaev honeycomb model~\cite{kitaev2006anyons}. On the contrary, as $\kappa$ is increased, the gap opens and the other bands hybridize. Thus, we have constructed a gapped, two-dimensional phase which can be realized by a two-local Hamiltonian. Models with properties such as these could prove useful for applications in error correction or general condensed matter physics.

\section{Discussion}
\label{section:Discussion}

We have developed a graph-theoretic framework for free-fermion solvability, which unifies the results of the previous work. Our results show that the absence of a claw and the presence of a simplicial clique in the frustration graph of a Hamiltonian are sufficient to prove that the model has a generic free-fermion solution. Further, this condition can be efficiently determined via the algorithm of Ref~\cite{chudnovsky2012growing}. Our key insight in this proof is the identification of a family of mutually commuting Hamiltonian symmetries --- the generalized cycle symmetries (see Section~\ref{section:ConservedQuantities}). The identification of these symmetries has also developed further the connection, previously established in Ref.~\cite{elman2021free}, between the absence of a claw in the frustration graph of a quantum many-body model and the integrability of the model. In the case of simplicial, claw-free graphs, we have shown that finding the single-particle energies of the model reduces to diagonalizing the generalized cycle operators of the model. These operators are, in general, not Pauli operators (see Section~\ref{subsection:ExampleApplication}). Nevertheless, the diagonalization of these operators does not necessarily represent an intractable barrier to obtaining the energies of these models.

While the sufficient condition of an SCF frustration graph encompasses the previous work in this area, SCF Hamiltonians do not span the entirety of the free-fermionic Hamiltonian space. For one thing we know that there exist models which are free for specific coefficients; for another there are models whose spectrum is the union of a super linearly increasing number of free-fermionic spectra. In Section~\ref{section:ApplicationResults}, we considered a two-dimensional model that is equivalent to a Jordan-Wigner solvable free-fermion model under a constant-depth circuit. One could imagine a class of models that can be mapped to a Jordan-Wigner solvable model via a circuit whose depth grows exponentially. Thus, it is clear that a characterization of free-fermion models could be developed based on the depth of the circuit needed to transform the model's frustration graph into a graph for which a Jordan-Wigner transformation to free fermions applies. Such a characterization could theoretically help to categorize those models which do not have a gapless phase, and thus fall beyond the scope of the characterization by critical exponents and quantum field theories~\cite{osborne2021quantum, osborne2023conformal}.

Since the claw is a forbidden induced subgraph of both previous graph-theoretic characterizations, it seems to be a natural criterion for a free-fermion solution to apply. However, it is perhaps mysterious that we require the presence of a simplicial clique as well. Indeed, the graph families of both previous characterizations are guaranteed to also contain a simplicial clique. This was recently shown in Ref.~\cite{chudnovsky2021note} for a family generalizing (even-hole, claw)-free graphs, as inspired by this precise question. Intuitively, we can view a simplicial clique as a kind of fermionic boundary mode, and its existence in a claw-free graph implies a recursive structure that forces the model to be fermionic in some sense. We capture this intuition by connecting these models to a polynomial divisibility result for the independence polynomial of a graph given by Godsil~\cite{godsil1993algebraic} and later generalized to the multivariate setting in Ref.~\cite{leake2019generalizations}. While our characterization does not capture all free-fermion solutions to spin models (there are well-known examples of non-generic solutions), we expect that it will be difficult to remove the simplicial clique assumption in the generic case. 

Regarding quantum circuits, there is also the question of whether this class of free-fermion models also extends the class of free-fermionic quantum gates. It is known that the matchgates~\cite{knill2001fermionic, terhal2002classical, bravyi2006universal, brod2011extending, hebenstreit2019all} represent a non-universal set of quantum gates. It is also known that a quantum circuit constructed entirely from matchgates results in Gaussian states which are ground states of a Jordan-Wigner-type Hamiltonians~\cite{hebenstreit2019all}. This then raises the question of whether there is an overlapping, but not universal, gate set which produces the SCF Hamiltonians developed here. Also related to quantum circuits is the notion of simulability. It has been suggested~\cite{matos2021quantifying} that free-fermion states can be more efficiently prepared using quantum optimization algorithms than their interacting counterparts. It would be interesting to know whether the ground states (or thermal states) of SCF models are more efficiently prepared using quantum optimization algorithms. It is clear then, that while the current work extends the class of known free-fermion-solvable models, there are avenues for further work regarding both free-fermion solutions to many-body physics, and using the mathematical framework of graph theory to probe our understanding of physics.

Certainly, it appears as though free-fermion solvable models are more abundant than was previously known, so an obvious question to ask would be: just how abundant are free-fermionic Hamiltonians? Even if the class of Hamiltonians was restricted to those currently known to have a generic solution it remains unknown what the likelihood is of a random Hamiltonian admitting an SCF free-fermion solution. There are also unanswered questions pertaining to whether free-fermionic systems can be leveraged for purposes within the field of quantum information. A strong link between the Jordan-Wigner-type free-fermion models of Ref.~\cite{chapman2020characterization} and error correcting codes for quantum computing has been established and studied by Ref.~\cite{chapman2022free}. However it remains to be seen whether such a link can be extended to include the (even-hole, claw)-free models developed in Ref.~\cite{elman2021free} or the SCF models developed in the current work. The ability to leverage free-fermionic solutions might also appear useful in the development of more compact fermion-to-qubit mappings~\cite{derby2021acompact, derby2021bcompact}, which could have a plethora of uses in the fields of quantum chemistry, condensed matter physics, and high-energy physics.

Another example of future work that could be probed within the mathematical framework of graph theory pertains to the question of qudit many-body models and \emph{free parafermions}. In particular, it would be interesting to understand whether there exists a graph-theoretic characterization of free parafermions in the same way that we have developed for fermions. Some examples of free-parafermionic models have been studied in isolation~\cite{fendley2014free, alcaraz2017energy, alcaraz2018anomalous, yao2021parafermionization}, while families of spin chains have also been constructed~\cite{alcaraz2020free, alcaraz2020integrable, alcaraz2021free, Alcaraz2021Powerful}. However, as of yet, there is no systematic identification mechanism akin to the simplicial, claw-free characterization of free fermions. We conjecture that such a characterization is possible.

\section*{Acknowledgements}

We thank Simon Benjamin, Paul Fendley, Steven Flammia, Tyler Helmuth, Alex Little, and Yuan Miao for helpful discussions. AC acknowledges support from EPSRC under agreement EP/T001062/1, and from EU H2020-FETFLAG-03-2018 under grant agreement no. 820495 (AQTION). SJE was supported by the ARC Centre of Excellence for Engineered Quantum Systems (EQUS, CE170100009) and in part with funding from the Defense Advanced Research Projects Agency under the Quantum Benchmarking (QB) program under award no. HR00112230007,  HR001121S0026, and HR001122C0074 contracts. RLM was supported by the QuantERA ERA-NET Cofund in Quantum Technologies implemented within the European Union's Horizon 2020 Programme (QuantAlgo project), EPSRC grants EP/L021005/1, EP/R043957/1, and EP/T001062/1, and the ARC Centre of Excellence for Quantum Computation and Communication Technology (CQC2T), project number CE170100012. The views, opinions and/or findings expressed are those of the authors and should not be interpreted as representing the official views or policies of the Department of Defense or the U.S. Government. No new data were created during this study.

\onecolumngrid

\appendix

\section{Proof of Lemma~\ref*{lemma:pairingrelations}}
\label{section:PairingRelationsProof}

\PairingRelations*

We shall assume without loss of generality that $\overrightarrow{\mco}$ is $(a,c)$ pairing and that $P^{(r)}$ is a connected component of $G[C_a{\oplus}C'_c]$. We shall prove statements~(a) and (b) separately. 

\subsection{Proof of Statement~(a)}

\begin{proof}[Proof of Lemma~\ref*{lemma:pairingrelations}~$(a)$]
    We first assume that $\bsu \in C'$ is a vertex such that $h_{\bsu}$ and $h_{P^{(r)}}$ anticommute. If $\bsu$ is in $U^{(r,s)}$ for some $s\in\{1,\dots,\ell_r\}$, then $\bsu$ is distinct from $\bsj_0^{(r)}$. We assume that $\bsu$ distinct from $\bsj_0^{(r)}$ and shall prove that $\bsu$ is in $U^{(r,s)}$ for at least one value of $s$. If $\bsu$ is in $P^{(r)}$, then $\bsu=\bsk_{\ell_r-1}^{(r)}$. However, $\bsk_{\ell_r-1}^{(r)}$ is in $C{\setminus}C'$ by the proof of Corollary~\ref{corollary:pathuniqueendpoints} and our assumption on $\overrightarrow{\mco}$. Therefore $\bsu$ is not in $P^{(r)}$.
    
    Suppose that $\bsu$ is in $P^{(r)}$ and that $\bsu$ is not in $U^{(r,s)}$ for any value of $s$ by way of contradiction. If $\bsu$ neighbors neither or both of the vertices in the set $\{\bsj_{s-1}^{(r)},\bsk_{s-1}^{(r)}\}$ for every $s$, then $h_{\bsu}$ and $h_{P^{(r)}}$ commute. Thus, suppose $\bsu$ is in $\Gamma_{C'}(\bsk_{s-1}^{(r)}){\setminus}\Gamma_{C'}[\bsj_{s-1}^{(r)}]$ for some $s$. If $s=\ell_r$, then this contradicts our assumption that $\bsk_{\ell_r-1}^{(r)}$ is the endpoint of $P^{(r)}$ by the proof of Corollary~\ref{corollary:pathuniqueendpoints} and by Corollary~\ref{corollary:vertexcycleneighboring}. Thus, $s<\ell_r$, and so $\ell_r>1$.
    
    We have that $\bsj_{s}^{(r)}$ is in $\Gamma_{C'}(\bsu)$, otherwise $\{\bsk_{s-1}^{(r)},\bsj_{s-1}^{(r)},\bsj_{s}^{(r)},\bsu\}$ induces a claw. Since $\bsj_{s}^{(r)}$ is in $C'_c$, then $\bsu$ is in $C'_d$. If $\bsk_{s}^{(r)}$ is not in $\Gamma(\bsu)$, then $\bsu$ is in $U^{(r,s+1)}$, and we assume that $\bsk_{s}^{(r)}$ is in $\Gamma(\bsu)$. There is an additional neighbor $\bsv \in C'_c$ to $\bsu$, which is distinct from $\bsj_{s-1}^{(r)}$. $\bsv$ is not in $\{\bsk_{s-1}^{(r)},\bsk_{s}^{(r)}\}$, since both vertices in this set are neighboring to $\bsj_{s}^{(r)} \in C'_c$. Further, $\bsv$ is not neighboring to $\bsk_{s-1}^{(r)}$, otherwise the set $\{\bsk_{s-1}^{(r)},\bsj_{s-1}^{(r)},\bsj_{s}^{(r)},\bsv\}$ induces a claw. Thus, $\bsv$ is in $\Gamma_{C'}(\bsk_{s}^{(r)})$, otherwise the set $\{\bsu,\bsk_{s-1}^{(r)},\bsk_{s}^{(r)},\bsv\}$ induces a claw, and so $\bsv=\bsj_{s+1}^{(r)}$ and $\ell_r>2$. Then $\Gamma_{P^{(r)}}(\bsu)=\bsk_{s-1}^{(r)}\md\bsj_{s}^{(r)}\md\bsk_{s}^{(r)}\md\bsj_{s+1}^{(r)}$ as $\bsu$ has no additional neighbors in $P^{(r)}$ by Lemma~\ref{lemma:vertexcyclerelations}. However, then $\bsu$ is in $U^{(r,s+2)}$ and $h_{\bsu}$ and $h_{P^{(r)}}$ commute. Therefore, $\bsu$ is in $U^{(r,s)}$ for at least one value of $s\in \{1,\dots,\ell_r\}$ if and only if $\bsu$ is distinct from $\bsj_0^{(r)}$. Now, by applying Lemma~\ref{lemma:vertexcyclerelations}, we have that $\Delta_{C^{(r)}}(\bsu)=\Delta_{C^{(r-1)}}(\bsu)+1$.

    If $\bsu$ is neighboring to $P^{(r)}$ as in Table~\ref{table:vertexcyclerelations} case (b.i), then $\bsu$ is in $U^{(r,s)}$ for $s>1$ and $\Gamma_{P^{(r)}}(\bsu)=\bsj_{s-2}^{(r)}\md\bsk_{s-2}^{(r)}\md\bsj_{s-1}^{(r)}$. If $\bsu$ is neighboring to $P^{(r)}$ as in Table~\ref{table:vertexcyclerelations} case (b.iii), then $\bsu$ is in $U^{(r,1)}$ and $\Gamma_{P^{(r)}}(\bsu)=\{\bsj_0^{(r)}\}$. Finally, if $\bsu$ is neighboring to $P^{(r)}$ as in Table~\ref{table:vertexcyclerelations} case (b.iv), then $\bsu$ is in $U^{(r,1)}$ and $\Gamma_{P^{(r)}}(\bsu)=\{\bsj_0^{(r)},\bsj_{t-1}^{(r)}\md\bsk_{t-1}^{(r)}\}$ for $1<t\leq\ell_r$ or $\Gamma_{P^{(r)}}(\bsu)=\{\bsj_0^{(r)},\bsk_{t-2}^{(r)}\md\bsj_{t-1}^{(r)}\}$ for $2<t\leq\ell_r$. In these cases, we have $\Delta_{C^{(r)}}(\bsu)=\Delta_{C^{(r-1)}}(\bsu)+1$. If $\bsu$ is neighboring to $P^{(r)}$ as in Table~\ref{table:vertexcyclerelations} case (b.iv) with $\Gamma_{{P^{(r)}}}(\bsu)=\{\bsk_{s-2}^{(r)}\md\bsj_{s-1}^{(r)},\bsk_{\ell_r-1}^{(r)}\}$ for $1<s<\ell_r$, then this contradicts assumption that $\bsk_{\ell_r-1}^{(r)}$ is an endpoint of $P^{(r)}$. This proves that $\Delta_{C^{(r)}}(\bsu)=\Delta_{C^{(r-1)}}(\bsu)+1$. This result allows us to prove Corollary~\ref{corollary:vertexanticommutingdistinctpaths}.
    
    Suppose $\bsu$ is in $U^{(r,s)}$ and $h_{\bsu}$ and $h_{P^{(r)}}$ anticommute. Further, suppose that $\Delta_{C^{(r)}}(\bsu)=3$. By Lemma~\ref{lemma:evenholecloneneighbor}, $\bsu$ neighbors a mutual neighbor $\bsv$ to $\bsj_{s-1}^{(r)}$ and $\bsk_{s-1}^{(r)}$ in $C^{(r,s-1)}$. Then $\bsv$ is not in $P^{(r)}$ since it neighbors both $\bsj_{s-1}^{(r)}$ and $\bsk_{s-1}^{(r)}$. Thus, $\bsv$ is in $C^{(r)}$ and $\bsv$ is in $C^{(r-1)}$. Similarly, let $\bsv'$ be the mutual neighbor to $\bsj_{s-1}^{(r)}$ and $\bsk_{s-1}^{(r)}$ in $C^{(r,s-1)}$ other than $\bsv$. By the same argument, $\bsv'$ is in $C^{(r)}$ and $\bsv'$ is in $C^{(r-1)}$. Let $\bsw$ be the additional neighbor to $\bsv$ in $C^{(r)}$ other than $\bsj_{s-1}^{(r)}$. Since $\bsu$ has three neighbors in $C^{(r)}$, then either $\Gamma_{C^{(r)}}(\bsu)=\bsw\md\bsv\md\bsj_{s-1}^{(r)}$ or $\Gamma_{C^{(r)}}(\bsu)=\bsv\md\bsj_{s-1}^{(r)}\md\bsv'$ by Lemma~\ref{lemma:vertexcyclerelations}. However, if $\bsv'$ is a neighbor to $\bsu$, then since $\Delta_{C^{(r)}}(\bsu)=\Delta_{C^{(r-1)}}(\bsu)+1$, we have that $\Gamma_{C^{(r-1)}}(\bsu)=\{\bsv,\bsv'\}$. This contradicts Lemma~\ref{lemma:vertexcyclerelations} since $\bsv$ and $\bsv'$ are not neighboring ($\{\bsv,\bsw,\bsk^{(r)}_{s-1},\bsu\}$ induces a claw). Thus, $\Gamma_{C^{(r)}}(\bsu)=\bsw\md\bsv\md\bsj_{s-1}^{(r)}$, so $\bsu\prec_{C^{(r)}}\bsv$. If $\bsv$ is in $C'$, then $\{\bsu,\bsv,\bsj_{s-1}^{(r)}\}$ induces a triangle in $C'$, so $\bsv$ is in $C{\setminus}C'$. Since $\bsv$ is a neighbor to $\bsk_{s-1}^{(r)} \in C_a$, then $\bsv$ is in $C_b$. Further, $\bsu$ is in $C'_d$ since $\bsu$ is in $U^{(r,s)}$, and so $\bsu\prec_{C^{(r)}}\bsv$ is $(a,c)$ pairing. This proves statement (a.i).
    
    Now, suppose $\bsu$ is in $U^{(r,s)}$, $h_{\bsu}$ and $h_{P^{(r)}}$ anticommute, and $\Delta_{C^{(r-1)}}(\bsu)=3$. Again, let $\bsv$ be the mutual neighbor to $\bsj_{s-1}^{(r)}$ and $\bsk_{s-1}^{(r)}$ in $C^{(r,s-1)}$ that is neighboring to $\bsu$. Then $\bsv$ is in $C^{(r-1)}$, $C^{(r)}$, and $C{\setminus}C'$. If $\Delta_{C}(\bsu)\neq3$, then there is a path $P^{(g)} \in \mco$ with $g<r$ whose corresponding operator anticommutes with $h_{\bsu}$. If $h_{\bsu}$ commutes with every such $h_{P^{(g)}}$, and $\Delta_C(\bsu)\neq3$, then $\Delta_{C^{(r-1)}}(\bsu)$ and $\Delta_C(\bsu)$ differ by an even number. However, $\bsu$ can have no odd degree in $C$ other than $3$ by Lemma~\ref{lemma:vertexcyclerelations}. Now suppose that $g<r$ is the largest index such that $h_{P^{(g)}}$ and $h_{\bsu}$ anticommute. By Corollary~\ref{corollary:pathcomponentsdistinct}, $P^{(g)}$ and $P^{(r)}$ are distinct. By Corollary~\ref{corollary:vertexanticommutingdistinctpaths} and Corollary~\ref{corollary:vertexanticommutingdistinctpathsendpoint} for $\overrightarrow{\mco}=\overrightarrow{\mco}'$, we have that $\bsu=\bsj_0^{(g)}$, and so $\Delta_{C^{(g)}}(\bsu)=2$. Therefore, $\Delta_{C^{(r-1)}}(\bsu)$ is even since we have assumed there is no path $P^{(q)}$ whose corresponding operator anticommutes with $h_{\bsu}$ for $g<q<r$. This contradicts our assumption that $\Delta_{C^{(r-1)}}(\bsu)=3$, and so $\Delta_{C}(\bsu)=3$. Now suppose that $\Delta_{C}(\bsu)=3$ and $g<r$ is the largest index such that $h_{P^{(g)}}$ and $h_{\bsu}$ anticommute. By Corollary~\ref{corollary:vertexanticommutingdistinctpathsendpoint}, $\bsu=\bsj_0^{(g)}$. Then $\bsu$ is in $C^{(g)}$ and thus $\bsu$ is in $C^{(r-1)}$ by Corollary~\ref{corollary:vertexonepathcomponent}. This contradicts our assumption that $h_{\bsu}$ and $h_{P^{(r)}}$ anticommute. Thus, if $\Delta_{C}(\bsu)=3$, then $h_{\bsu}$ commutes with all operators $h_{P^{(g)}}$ with $g<r$. Thus, $\Delta_{C^{(r)}}(\bsu)=4$, since $\bsu$ can have no odd degree in $C^{(r-1)}$ other than $3$ and $\Delta_{C^{(r)}}(\bsu)=\Delta_{C^{(r-1)}}(\bsu)+1$. Since $\bsv$ is in $\Gamma_C(\bsu)$, and $\bsk_{s-1}^{(r)} \in C_a$ is not neighboring to $\bsu$, the clone to $\bsu$ in $C$ is in $C_a$. Therefore, $\bsu\prec_C\bsv$ is $(a,d)$ pairing. This proves statement (a.ii), and completes the proof.
\end{proof}

\subsection{Proof of Statement~(b)}

\begin{proof}[Proof of Lemma~\ref*{lemma:pairingrelations}~$(b)$]
    We assume that $\bsu \in C'$ is a vertex such that $h_{\bsu}$ and $h_{P^{(r)}}$ commute and $\bsu$ is in $U^{(r,s)}$ for some $s\in\{1,\dots,\ell_r\}$. We shall show that $\Delta_{C^{(r)}}(\bsu)=4$ and that there is a vertex $\bsv \in U^{(r,s')}$ with $s'<s$ in the same connected component $P$ of $G[C_b{\oplus}C'_d]$ as $\bsu$. We achieve this by showing that $s\neq1$ if $h_{\bsu}$ and $h_{P^{(r)}}$ commute. We argue that there is a vertex $\bsu'$ in $U^{(r,s')}$ with $s'<s$ such that $h_{\bsu'}$ and $h_{P^{(r)}}$ anticommute.
    
    First, suppose that $\bsu$ is in $P^{(r)}$, then $\bsu=\bsk_{s-2}^{(r)}$ by our assumptions. Hence, $s>1$ and $\ell_r>1$. Since $P^{(r)}$ is a connected component of $G[C_a{\oplus}C'_c]$ and $\bsk_{s-2}^{(r)}$ is in $C'$, then $\bsk_{s-2}^{(r)}$ is in $C_a \cap C'_d$ by construction. Then ${\Gamma_{C^{(r)}}(\bsu)=\Gamma_{C^{(r-1)}}(\bsu)\cup\{\bsj_{s-2}^{(r)},\bsj_{s-1}^{(r)}\}}$, and so $\Delta_{C^{(r)}}(\bsu)=4$ since $\Delta_{C^{(r-1)}}(\bsu)=2$ and $\Gamma_{C^{(r-1)}}(\bsu)\cap\{\bsj_{s-2}^{(r)},\bsj_{s-1}^{(r)}\}=\varnothing$. Further, $\bsj_{s-2}^{(r)}$ has an additional neighbor $\bsv \in C'$ other than $\bsk_{s-2}^{(r)}$, and thus $\bsv \in U^{(r,s-1)}$. By Corollary~\ref{corollary:evenholecloneexactlyoneneighbor}, $\bsv$ neighbors $\bsw \in C^{(r,s-3)}$, which is a mutual neighbor to $\bsj_{s-2}^{(r)}$ and $\bsk_{s-2}^{(r)}$. Then neither $\bsw$ is in $C'$ nor $\{\bsv, \bsw,\bsj_{s-2}^{(r)}\}$ induces a triangle in $C'$. Thus, $\bsw$ is in $C$, and hence in $C_b$ since it is neighboring to $\bsk_{s-2}^{(r)} \in C_a$.
    Similarly, $\bsv$ is in $C'_d$, since it is neighboring to $\bsj_{s-2}^{(r)} \in C'_c$. Thus, $\bsv$ and $\bsu$ are in $P$ since they both neighbor $\bsw \in C_b$.
    
    Suppose that $\bsu$ is not in $P^{(r)}$. If $s=1$, then $\bsu$ has an odd number of neighbors in $P^{(r)}$ by Lemma~\ref{lemma:vertexcyclerelations}, and so $s>1$, $\ell_r>1$, and $\bsk_{s-2}^{(r)}$ is in $\Gamma_{P^{(r)}}(\bsu)$ otherwise $\{\bsj_{s-1}^{(r)},\bsk_{s-2}^{(r)},\bsk_{s-1}^{(r)},\bsu\}$ induces a claw. Let $\bsw_{s-1}$ be the mutual neighbor in $C^{(r-1)}$ to every vertex in the set $\{\bsj_{s-2}^{(r)},\bsk_{s-2}^{(r)},\bsj_{s-1}^{(r)},\bsk_{s-1}^{(r)}\}$, as shown in Fig.~\ref{figure:deformationsequence}. We consider the case where $\bsu$ neighbors neither $\bsw_{s-1}$ nor $\bsj_{s-2}^{(r)}$ and then the case where $\bsu$ neighbors at least one of these vertices. Recall that $\abs{C}=2k$ by the labeling for $C$ in Eq.~(\ref{equation:evenholelabel}), and that $|C^{(r-1)}|=\abs{C}$.

    If $k=2$, let $C^{(r-1)}=\bsw_{s-2}\md\bsk_{s-2}^{(r)}\md\bsw_{s-1}\md\bsk_{s-1}^{(r)}\md\bsw_{s-2}$. If $\bsu$ does not neighbor $\bsw_{s-1}$ or $\bsj_{s-2}^{(r)}$, then $\bsu$ neighbors $\bsw_{s-2}$ by Corollary~\ref{corollary:vertexcycleneighboring}. However, then $\{\bsw_{s-2},\bsj_{s-2}^{(r)},\bsk_{s-1}^{(r)},\bsu\}$ induces a claw. Thus, $\bsu$ neighbors either $\bsw_{s-1}$ or $\bsj_{s-2}^{(r)}$.
    
    If $k=3$, let $C^{(r-1)}=\bsw_{s-2}\md\bsk_{s-2}^{(r)}\md\bsw_{s-1}\md\bsk_{s-1}^{(r)}\md\bsw_s\md\bst\md\bsw_{s-2}$. Assume that $\bsu$ does not neighbor $\bsw_{s-1}$ or $\bsj_{s-2}^{(r)}$. By applying Corollary~\ref{corollary:evenholecloneexactlyoneneighbor} to $\bsj_{s-1}^{(r)}\prec_{C^{(r,s-1)}}\bsk_{s-1}^{(r)}$ and using the assumption that $\bsw_{s-1}$ is not in $\Gamma_{C^{(r,s-1)}}(\bsu)$, we have that $\bsw_s\md\bst$ is in $\Gamma_{C^{(r-1)}}(\bsu)$. We require that $\bsk_{s-2}^{(r)}$ is in $\Gamma_{C^{(r-1)}}(\bsu)$, otherwise $\{\bsj_{s-1}^{(r)},\bsk_{s-1}^{(r)},\bsk_{s-2}^{(r)},\bsu\}$ induces a claw. Then $\Gamma_{C^{(r-1)}}(\bsu)=\bsw_s\md\bst\md\bsw_{s-2}\md\bsk_{s-2}^{(r)}$ by Corollary~\ref{corollary:vertexcycleneighboring} and the assumption that $\bsw_{s-1}$ is not in $\Gamma_{C^{(r,s-1)}}(\bsu)$. If $\bst$ is in $P^{(r)}$, then there is a vertex $\bss \in C'$ such that $\bst$ is the clone to $\bss$ in the deformation by $P^{(r)}$, i.e., $\bss$ is in $\{\bsj_{s-3}^{(r)},\bsj_s^{(r)}\}$. Then $\bsw_s\md\bst\md\bsw_{s-2}$ in in $\Gamma_{C^{(r-1)}}(\bss)$, and so $\bsu$ is in $\Gamma_{C'}(\bss)$, otherwise $\{\bsw_{s-2},\bsj_{s-2}^{(r)},\bss,\bsu\}$ induces a claw. Therefore, $\Delta_{C^{(r)}}(\bsu)=4$ regardless of whether $\bst$ is in $P^{(r)}$. We now prove that there is a vertex $\bsv$ in $U^{(r,s')}$ with $s'<s$ in $P$.
    
    Assuming that $k=3$ and $\bsu$ does not neighbor $\bsw_{s-1}$ or $\bsj_{s-2}^{(r)}$, then $\Delta_{C'}(\bsk_{s-2}^{(r)})>3$, since its neighbors $\{\bsj_{s-2}^{(r)},\bsj_{s-1}^{(r)},\bsu\}$ do not induce a path. Thus, there is a vertex $\bsv$ in $U^{(r,s-1)}$. This vertex is not in $C^{(r-1)}$ since then it is a neighbor to $\bsj_{s-2}^{(r)}$ in $C^{(r-1)} \cap C'$. If $\bsj_{s-2}^{(r)}$ has an additional neighbor in $C^{(r-1)}$, then the neighbor is $\bst$, however $\bst$ is not in $C'$, otherwise $\bst$ neighbors $\bsj_{s-2}^{(r)} \in C'_c$ and $\bsu \in C'_d$. Thus, let $\bsv$ be a in $U^{(r,s-1)}$, and so $\bsv$ is not in $C^{(r-1)}$. If $\bsv$ does not neighbor $\bsw_{s-2}$ or $\bsw_{s}$, then $\Gamma_{C^{(r-1)}}(\bsv)=\bsw_{s-1}\md\bsk_{s-1}^{(r)}$ by Corollary~\ref{corollary:evenholecloneexactlyoneneighbor} applied to $\bsj_{s-2}^{(r)}\prec_{C^{(r,s-2)}}\bsk_{s-2}^{(r)}$. Since both $\bsw_{s-1}$ and $\bsk_{s-1}^{(r)}$ have neighbors in $C'_c$ and $C'_d$, then $\bsv$ has an additional neighbor $\bsv' \in C'{\setminus}C^{(r-1)}$. If $\bsv'$ does not neighbor $\bsk_{s-1}^{(r)}$, then $\{\bsv,\bsv',\bsk_{s-1}^{(r)},\bsj_{s-2}^{(r)}\}$ induces a claw ($\bsv'$ and $\bsj_{s-2}^{(r)}$ are in the same coloring class of $C'$). Thus, $\bsv'$ neighbors $\bsk_{s-1}^{(r)}$. Then $\bsv'=\bsj_{s}^{(r)}$ and $\Gamma_{C^{(r-1)}}(\bsv')=\bsk_{s-1}^{(r)}\md\bsw_s\md\bst\md\bsw_{s-2}$. However, this induces the claw $\{\bsj_s^{(r)},\bsv,\bsw_{s-2},\bsw_s\}$. Thus, $\bsv$ neighbors at least one of $\bsw_{s-2}$ or $\bsw_s$, and so $\bsv$ is in $U^{(r,s-1)}$ and $\bsv$ is in $P$.
    
    If $k>3$, let $C^{(r-1)}=\dots\md\bsw_{s-2}\md\bsk_{s-2}^{(r)}\md\bsw_{s-1}\md\bsk_{s-1}^{(r)}\md\bsw_s\md\bst\md\bsw_{s+1}\md\dots$. Then $\bsk_{s-2}^{(r)}$ is in $\Gamma_{C^{(r-1)}}(\bsu)$, otherwise $\{\bsj_{s-1}^{(r)},\bsk_{s-2}^{(r)},\bsk_{s-1}^{(r)},\bsu\}$ induces a claw. If $\bsu$ does not neighbor $\bsw_{s-1}$ or $\bsj_{s-2}^{(r)}$, then $\bsw_{s-2}$ is in $\Gamma_{C^{(r-1)}}(\bsu)$ to satisfy Corollary~\ref{corollary:vertexcycleneighboring} for $\bsk_{s-2}^{(r)}$ in $\Gamma_{C^{(r-1)}}(\bsu)$. Further, $\bst$ is in $\Gamma_{C^{(r-1)}}(\bsu)$ to satisfy Corollary~\ref{corollary:evenholecloneexactlyoneneighbor} for $\bsj_{s-1}^{(r)}\prec_{C^{(r, s-1)}}\bsk_{s-1}^{(r)}$, and $\bst$ is not in $\Gamma_{C^{(r-1)}}(\bsj_{s-1}^{(r)})$ by construction. However, then $\{\bsu,\bsw_{s-2},\bsj_{s-1}^{(r)},\bst\}$ induces a claw. Thus, $\bsu$ neighbors at least one of $\bsw_{s-1}$ and $\bsj_{s-2}^{(r)}$.
    
    Let $\bsw_{s-1}$ be the mutual neighbor in $C^{(r-1)}$ to every vertex in the set $\{\bsj_{s-2}^{(r)},\bsk_{s-2}^{(r)},\bsj_{s-1}^{(r)},\bsk_{s-1}^{(r)}\}$, as shown in Fig.~\ref{figure:deformationsequence}. We assume that $\bsu$ is neighboring to $\bsw_{s-1}$ or $\bsj_{s-2}^{(r)}$, or, if there is a claw, then there is a vertex $\bsv \in U^{(r,s-1)}$ in the same component of $G[C_b{\oplus}C'_d]$ as $\bsu$. If $\bsu$ is neighboring to $\bsw_{s-1}$, then it is neighboring to $\bsj_{s-2}^{(r)}$, otherwise $\{\bsw_{s-1},\bsj_{s-2}^{(r)},\bsk_{s-1}^{(r)},\bsu\}$ induces a claw. Thus, $\bsu$ is neighboring to $\bsj_{s-2}^{(r)}$ regardless of whether it is neighboring to $\bsw_{s-1}$. By Lemma~\ref{lemma:vertexcyclerelations}, $\bsu$ is then  neighboring to $\bsk_{s-3}^{(r)}$ so that $h_{\bsu}$ and $h_{P^{(r)}}$ anticommute, and then $s>2$ with $\ell_r>2$. Further, $\bsu$ does not neighbor $\bsj_{s-3}^{(r)}$, since $\bsu$ has four neighbors in $P^{(r)}$ with $\Gamma_{P^{(r)}}(\bsu)=\bsk_{s-3}^{(r)}\md\bsj_{s-2}^{(r)}\md\bsk_{s-2}^{(r)}\md\bsj_{s-1}^{(r)}$. Let $\bsw_{s-2}$ be the mutual neighbor in $C^{(r-1)}$ to every vertex in the set $\{\bsj_{s-3}^{(r)},\bsk_{s-3}^{(r)},\bsj_{s-2}^{(r)},\bsk_{s-2}^{(r)}\}$. This vertex must exist by Lemma~\ref{lemma:evenholecloneneighbor}. Since $\ell_r>2$, then $k>2$, where we recall that $\abs{C}=2k$ by our labeling scheme for $C$. By Lemma~\ref{lemma:evenholecloneneighbor} with respect to $C^{(r,s-3)}$ (prior to the deformation by $\bsj_{s-3}^{(r)}\prec_{C^{(r,s-3)}}\bsk_{s-3}^{(r)}$), we have that $\bsu$ is neighboring to $\bsw_{s-2}$, since it is neighboring to $\bsk_{s-3}^{(r)} \in C^{(r,s-3)}$ and not $\bsj_{s-3}^{(r)}$. Then $\bsw_{s-2}$ is in $C{\setminus}C'$, since if $\bsv_1$ is in $C'$, then $\{\bsj_{s-1}^{(r)},\bsw_{s-2},\bsu\}$ induces a triangle in $C'$. Further, $\bsw_{s-2}$ is in $C_b$, since it neighbors $\bsk_{s-3}^{(r)}$ in $C_a$. By Lemma~\ref{lemma:evenholecloneneighbor} with respect to the deformation by $\bsj_{s-1}^{(r)}\prec_{C^{(r,s-1)}}\bsk_{s-1}^{(r)}$, we have that $\bsu$ is neighboring to $\bsw_{s-1}$, since it is neighboring to $\bsj_{s-2}^{(r)} \in C^{(r,s-1)}$. Then $\bsw_{s-1}$ is in $C{\setminus}C'$, since if $\bsw_{s-1}$ is in $C'$, then $\{\bsj_s^{(r)},\bsw_{s-1},\bsu\}$ induces a triangle in $C'$. Further, $\bsw_{s-1}$ is in $C_b$, since it neighbors $\bsk_{s-1}^{(r)}$ in $C_a$. Thus, $\Gamma_{C^{(r-1)}}(\bsu)=\bsk_{s-3}^{(r)}\md\bss_{s-2}\md\bsk_{s-2}^{(r)}\md\bsw_{s-1}$ and $\Gamma_{C^{(r)}}(\bsu)=\bsw_{s-2}\md\bsj_{s-2}^{(r)}\md\bsw_{s-1}\md\bsj_{s-1}^{(r)}$. Therefore, $\Delta_{C^{(r)}}(\bsu)=4$.
    
    Now, let $\bsv_{s-2}$ be the additional neighbor to $\bsj_{s-2}^{(r)}$ in $C_d'$ other than $\bsu$. Then $\bsv_{s-2}$ is in $C^{(r-1)}$, since $\Gamma_{C^{(r-1)}}(\bsj_{s-2}^{(r)})=\Gamma_{C^{(r-1)}}(\bsu)=\bsk_{s-3}^{(r)}\md\bsw_{s-2}\md\bsk_{s-2}^{(r)}\md\bsk_{s-1}$, and $\bsu$ is in $C'_d$. If $\bsv_{s-2}$ does not neighbor $\bsk_{s-2}^{(r)}$, then $\bsv_{s-2}$ is in $U^{(r,s-1)}$. Further, $\bsv_{s-2}$ neighbors either $\bsw_{s-2}$ or $\bsw_{s-1}$, otherwise $\{\bsj_{s-2}^{(r)},\bsw_{s-2},\bsw_{s-1},\bsv_{s-2}\}$ induces a claw. Both $\bsw_{s-2}$ and $\bsw_{s-1}$ are neighbors to $\bsu$ in $C_b$, and so $\bsv_{s-2}$ and $\bsu$ are in the same connected component of $G[C_b{\oplus}C'_d]$.
    
    Assume that $\bsv_{s-2}$ neighbors $\bsk_{s-2}^{(r)}$, then $\bsv_{s-2}$ neighbors $\bsw_{s-2}$, otherwise $\{\bsk_{s-2}^{(r)},\bsj_{s-1}^{(r)},\bsv_{s-2},\bsw_{s-2}\}$ induces a claw. By applying Lemma~\ref{lemma:evenholecloneneighbor} with respect to $\bsu\prec_{C^{(r,s-2)}}\bsk_{s-2}^{(r)}$, we have that $\bsj_{s-3}^{(r)}$ is in $\Gamma(\bsv_{s-2})$. If $\bsv_{s-2}$ does not neighbor $\bsk_{s-3}^{(r)}$, then $\bsv_{s-2}$ is in $U^{(r,s-2)}$. Further, $\bsv_{s-2}$ and $\bsu$ are in the same connected component of $G[C_b{\oplus}C'_d]$, since they both neighbor $\bsw_{s-2}$ in $C_b$. Assume that $\bsv_{s-2}$ neighbors $\bsk_{s-3}^{(r)}$. Let $\bsw_{s-3}$ be the mutual neighbor to $\bsj_{s-3}^{(r)}$ and $\bsk_{s-3}^{(r)}$ in $C^{(r-1)}$ other than $\bsw_{s-2}$. Then $\bsw_{s-3}$ is in $\Gamma(\bsv_{s-2})$, otherwise $\{\bsk_{s-3}^{(r)},\bsu,\bsv_{s-2},\bsw_{s-3}\}$ induces a claw, since $\bsu$ has four neighbors in $C^{(r-1)}$ which do not include $\bsw_{s-3}$. If $\bsw_{s-3}$ is in $C'$, then $\{\bsw_{s-3},\bsv_{s-2},\bsj_{s-3}^{(r)}\}$ induces a triangle in $C'$, and so $\bsw_{s-3}$ is in $C{\setminus}C'$. Additionally, $\bsw_{s-3}$ is in $C_b$, since $\bsw_{s-3}$ is neighboring to $\bsk_{s-3}^{(r)}$ in $C_a$. Let $\bsv_{s-3}$ be the additional neighbor to $\bsj_{s-3}^{(r)}$ in $C_d'$. Then $\bsv_{s-3}$ does not neighbor $\bsk_{s-3}^{(r)}$, since $\Gamma_{C'}(\bsk_{s-3}^{(r)})=\bsj_{s-3}^{(r)}\md\bsv_{s-2}\md\bsj_{s-2}^{(r)}\md\bsu$. Thus, $\bsv_{s-3}$ is in $U^{(r, s-2)}$. Then $\bsv_{s-3}$ is neighboring to $\bsw_{s-3}$, otherwise $\{\bsj_{s-3}^{(r)},\bsw_{s-3},\bsw_{s-2},\bsv_{s-3}\}$ induces a claw, and $\Gamma_{C'}(\bsw_{s-2})=\bsj_{s-3}^{(r)}\md\bsv_{s-2}\md\bsj_{s-2}^{(r)}\md\bsu$. Thus, $\bsv_{s-3}$ and $\bsu$ are in the same connected component of $G[C_b{\oplus}C'_d]$, which includes $\bsv_{s-3}\md\bsw_{s-3}\md\bsv_{s-2}\md\bsw_{s-2}\md\bsu$.
    
    We recursively apply this argument by replacing $\bsu$ with the only member of $U^{(r,s-\delta)}$ and $s$ with $s-\delta$ for $\delta\in\{1,2\}$. Since $s$ is strictly decreasing in this recursion, then there is eventually a $\bsu$ such that $h_{\bsu}$ and $h_{P^{(r)}}$ anticommute, since $h_{\bsu}$ and $h_{P^{(r)}}$ do not commute if $\bsu$ is in $U^{(r,1)}$. Therefore, there exists a vertex $\bsu'$ in $U^{(r,s')}$ with $s'<s$ such that $h_{\bsu'}$ and $h_{P^{(r)}}$ anticommute, and the unique path component containing $\bsu'$ with the same pairing type as $\overrightarrow{\mco}$ contains $\bsu$. This completes the proof.
\end{proof}

\subsection{Proof of Corollary~\ref*{corollary:vertexanticommutingdistinctpaths}}

\VertexAnticommutingDistinctPaths*

\begin{proof}
    Let $P^{(r)}\in\overrightarrow{\mco}$ and $P^{(g)}\in\overrightarrow{\mco}'$ be distinct path components for possibly non-distinct deformations $\overrightarrow{\mco}$ and $\overrightarrow{\mco}'$ of the same fixed pairing type. Assume there is a vertex $\bsu$ such that $h_{\bsu}$ anticommutes with $h_{P^{(r)}}$ and $h_{P^{(g)}}$ for $\bsu \in U^{(r,s)} \cap U^{(g,t)}$ for some $s\in\{1,\dots\ell_r\}$ and $t\in\{1,\dots\ell_g\}$. Since $P^{(r)}$ is a component of $G[C_a{\oplus}C'_c]$, then $\bsu$ is in $C'_d$, since it is a neighbor to $\bsj_0^{(r)} \in C'_c$. Thus, $P^{(g)}$ is a component of $G[C_a{\oplus}C'_c]$, since $\bsu$ is a neighbor to $\bsj_0^{(r)} \in C'$, and $\bsj_0^{(r)}$ is not in $C'_d$. Since $P^{(r)}$ and $P^{(g)}$ are distinct, they are therefore disjoint. It follows from the neighboring relations in the proof of Lemma~\ref{lemma:pairingrelations}~$(a)$, that either $\Gamma_{P^{(r)}}(\bsu)=\{\bsj_0^{(r)}\}$ or $\Delta_{C'_c \cap P^{(r)}}(\bsu)=2$, and similarly for $P^{(g)}$. Thus, $\Gamma_{P^{(r)}}(\bsu)=\{\bsj_0^{(r)}\}$ and $\Gamma_{P^{(g)}}(\bsu)=\{\bsj_0^{(g)}\}$. Otherwise, $\bsu$ would have three distinct neighbors in $C'_c$.
    
    Let $C^{(g-1)}\in\avg{C}$ be the even hole given by deforming $C$ by $\overrightarrow{\mco}'$ until the deformation by $P^{(g)}$. Similarly, let $C^{(r-1)}\in\avg{C}$ be the even hole given by deforming $C$ by $\overrightarrow{\mco}$ until the deformation by $P^{(r)}$. By Corollary~\ref{corollary:evenholecloneexactlyoneneighbor}, $\bsu$ is a neighbor to a mutual neighbor $\bsv_g$ to $\bsj_0^{(g)}$ and $\bsk_0^{(g)}$ in $C^{(g-1)}$. As in the proof of Corollary~\ref{corollary:vertexonepathcomponent}, $\bsv_g$ is in $C{\setminus}C'$, or $\{\bsj_0^{(g)},\bsv_g,\bsu\}$ induces a triangle in $C'$. Similarly, let $\bsv_r$ be a mutual neighbor to $\bsj_0^{(r)}$ and $\bsk_0^{(r)}$ in $C^{(r-1)}$ that is a neighbor to $\bsu$. By the same argument, $\bsv_r$ is in $C{\setminus}C'$. Since $\bsv_g$ is a neighbor to $\bsk_0^{(g)}$ in $C_a$, then $\bsv_g$ is in $C_b$. Thus, if $\bsu$ is in $C$, then $\bsu$ is in $C_a$. However, this contradicts the assumption that $\bsj_0^{(g)}$ is an endpoint of $P^{(g)}$, since $\bsj_0^{(g)}$ has two distinct neighbors $\bsu$ and $\bsk_0^{(g)}$ in $C_a{\setminus}C'_c$. Thus, $\bsu$ is in $C'_d{\setminus}C$ and $\bsu$ has two possibly non-distinct neighbors $\bsv_g$ and $\bsv_r$ in $C_b{\setminus}C'$.
    
    Assume that $\bsv_g$ is distinct from $\bsv_r$. By Lemma~\ref{lemma:vertexcyclerelations}, $\bsu$ has either three or four neighbors in $C$. Assume that $\bsu$ has four neighbors in $C$, then it has two distinct neighbors $\bsw$ and $\bsw'$ in $C_a$. Then $\bsj_0^{(g)}$ is not in $\{\bsw,\bsw'\}$ since $\bsj_0^{(g)}$ is in $C'_c{\setminus}C_a$, and $\bsk_0^{(g)}$ is not in $\{\bsw,\bsw'\}$ since $\bsk_0^{(g)}$ does not neighbor $\bsu$. Thus, neither $\bsw$ nor $\bsw'$ neighbor $\bsj_0^{(g)}$, since the only neighbor to $\bsj_0^{(g)}$ in $C_a$ is $\bsk_0^{(g)}$. This induces the claw $\{\bsu,\bsj_0^{(g)},\bsw,\bsw'\}$, and so $\bsu$ does not have four neighbors in $C$. Now assume that $\bsu$ has three neighbors in $C$, then it has a neighbor $\bsw$ in $C_a$, which is a mutual neighbor to $\bsv_g$ and $\bsv_r$. By a similar argument, $\bsw$ is not in $\{\bsk_0^{(g)},\bsk_0^{(r)}\}$, since neither $\bsk_0^{(g)}$ nor $\bsk_0^{(r)}$ neighbor $\bsu$. Thus, neither $\bsj_0^{(g)}$ nor $\bsj_0^{(r)}$ neighbor $\bsw$, since their only neighbors in $C_a$ are $\bsk_0^{(g)}$ and $\bsk_0^{(r)}$, respectively. This induces the claw $\{\bsu, \bsj_0^{(g)},\bsj_0^{(r)},\bsw\}$, and so $\bsv_g=\bsv_r$. If $\bsv_g=\bsv$, then $\{\bsv,\bsk_0^{(g)},\bsk_0^{(r)},\bsu\}$ induces a claw. Therefore, no such vertex $\bsu$ exsits, completing the proof.
\end{proof}

\section{Proof of Lemma~\ref*{lemma:solutioncriterion}}
\label{section:SolutionCriterionProof}

\SolutionCriterion*

We first prove the following lemma, which shows that any vertex in exactly one of $C'$ or $C^{(m)}$ is not involved in the deformation $\overrightarrow{\mco}$ when $U$ is a subset of $C^{(m)}$.
\begin{lemma}
    \label{lemma:vertexnotinpartialdeformation}
    Let $\overrightarrow{\mco}$ be a fixed-pairing-type deformation of $C$ by vertices in $C'$ with labeling as in Def.~\ref{definition:fixedpairingtypedeformation}, and such that $U$ is a subset of $C^{(m)}$. For any vertex $\bsu$ in $C'{\oplus}C^{(m)}$, $\bsu$ is not in $\partial\mco$.
\end{lemma}
\begin{proof}
    Assume $\bsu$ is in $C^{(m)}{\setminus}C'$. Then $\bsu$ is distinct from $\bsj_{s-1}^{(r)}$ for all $r$ and $s$. Otherwise, $\bsu$ is in $C'$ by our assumption  on $\overrightarrow{\mco}$. Now assume that $\bsu=\bsk_{s-1}^{(r)}$ for some $r$ and $s$, then $\bsu$ is not in $C^{(m)}$ unless $\bsu=\bsj_{t-1}^{(g)}$ for some $t$ and $g>r$. Thus, if $\bsu$ is in $C^{(m)}{\setminus}C'$, then $\bsu$ is not in $\partial\mco$.
    
    Assume that $\bsu$ is in $C'{\setminus}C^{(m)}$. If $\bsu=\bsk_{s-1}^{(r)}$ for some $r$ and $s$, then $\bsu$ has an additional neighbor $\bsj_{s}^{(r)} \in P^{(r)}$ in the same coloring class of $C'$ as $\bsj_{s-1}^{(r)}$, since $\bsu$ is in $C'$. Thus, $\bsu$ is in $U^{(r,s+1)}$. This contradicts our assumption that $U$ is a subset of $C^{(m)}$. If $\bsu=\bsj_{s-1}^{(r)}$, then $\bsu$ is in $C^{(m)}$ unless $\bsu=\bsk_{t-1}^{(g)}$ for some $t$ and $g>r$. Therefore, if $\bsu$ is in $C^{(m)}{\oplus}C'$, then $\bsu$ is not in $\partial\mco$. This completes the proof.
\end{proof}

We now prove Lemma~\ref{lemma:solutioncriterion}.
\begin{proof}[Proof of Lemma~\ref*{lemma:solutioncriterion}]
    Assume that $\overrightarrow{\mco}$ is such that $U$ is a subset of $C^{(m)}$, and let
    \begin{equation}
        C'^{(r,s)} = C'{\oplus}\left(\bigoplus_{g=r+1}^{m}P^{(g)}\right){\oplus}(P^{(r)}{\setminus}P^{(r)}_{2s-1}),
    \end{equation}
    with $C'^{(r)}=C'^{(r-1,\ell_r)}=C'{\oplus}\bigoplus_{g=r}^mP^{(g)}$ and $C'^{(m,\ell_m)}=C'^{(m+1)}=C'$. We shall show that $\bsk_{s-1}^{(r)}\prec_{C'^{(r,s)}}\bsj_{s-1}^{(r)}$ for all $s\in\{1,\dots,\ell_r\}$ for all $r\in\{0,\dots,m\}$. By our assumption on $\overrightarrow{\mco}$, we have $\bsk_{s-1}^{(r)}\prec_{C^{(r,s)}}\bsj_{s-1}^{(r)}$. Thus, if there is a vertex $\bsv$ in exactly one of $\Gamma_{C'^{(r,s)}}(\bsk_{s-1}^{(r)})$ or $\Gamma_{C'^{(r,s)}}[\bsj_{s-1}^{(r)}]$, then either
    \begin{enumerate}
        \item[(i)] $\bsv=\bsj_{s-1}^{(r)}$ and $\bsj_{s-1}^{(r)}$ is not in $C'^{(r,s)}$, since $\bsj_{s-1}^{(r)}$ is a neighbor to $\bsk_{s-1}^{(r)}$ in $C^{(r,s)}$. Thus, $\bsv$ is in $C^{(r,s)}{\setminus}C'^{(r,s)}$.
        \item[(ii)] $\bsv$ is in $C'^{(r,s)}{\setminus}C^{(r,s)}$, since $\Gamma_{C^{(r,s)}}[\bsj_{s-1}^{(r)}]=\Gamma_{C^{(r,s)}}(\bsk_{s-1}^{(r)})$ as $\bsk_{s-1}^{(r)}\prec_{C^{(r,s)}}\bsj_{s-1}^{(r)}$.
    \end{enumerate}
    Note that
    \begin{align}
        C'^{(r,s)}{\oplus}C^{(r,s)} &= (C{\oplus}C'){\oplus}\left(\bigoplus_{g=0}^mP^{(g)}\right) \\
        &= C'{\oplus}C^{(m)}.
    \end{align}
    By applying Lemma~\ref{lemma:vertexnotinpartialdeformation}, if $\bsv$ is in $\Gamma_{C'^{(r,s)}}(\bsk_{s-1}^{(r)}){\oplus}\Gamma_{C'^{(r,s)}}[\bsj_{s-1}^{(r)}]$, then $\bsv$ is in $C'^{(r,s)}{\oplus}C^{(r,s)}=C'{\oplus}C^{(m)}$, and so $\bsv$ is not in $\partial\mco$. Since $\bsj_{s-1}^{(r)}$ is in $\partial\mco \cap C^{(r,s)}$, then $\bsj_{s-1}^{(r)}$ is in $C'^{(r,s)}$. This rules out case (i). Similarly, $\bsk_{s-1}^{(r)}$ is in $\partial\mco{\setminus}C^{(r,s)}$, and so $\bsk_{s-1}^{(r)}$ is not in $C'^{(r,s)}$.
    
    If $\bsv$ is in $\Gamma_{C'^{(r,s)}}(\bsk_{s-1}^{(r)}){\oplus}\Gamma_{C'^{(r,s)}}[\bsj_{s-1}^{(r)}]$, then we are in case (ii), and so $\bsv$ is in $C'^{(r,s)}{\setminus}C^{(r,s)}$. If $\bsv$ is in $C^{(m)}{\setminus}C'$, then $\bsv$ is in $C$, which contradicts that $\bsv$ is not in $C^{(r,s)}$, since $\bsv$ is not in $\partial\mco$. Thus, $\bsv$ is in $C'{\setminus}C^{(m)}$.
    
    If $\bsv$ is in $\Gamma_{C'^{(r,s)}}[\bsj_{s-1}^{(r)}]{\setminus}\Gamma_{C'^{(r,s)}}(\bsk_{s-1}^{(r)})$, then $\bsv$ is in $U^{(r,s)}$, but this contradicts our assumption that $U$ is a subset of $C^{(m)}$. Thus, $\Gamma_{C'^{(r,s)}}[\bsj_{s-1}^{(r)}]$ is a subset of $\Gamma_{C'^{(r,s)}}(\bsk_{s-1}^{(r)})$.
    
    Now, suppose that $\bsv$ is in $\Gamma_{C'^{(r,s)}}(\bsk_{s-1}^{(r)}){\setminus}\Gamma_{C'^{(r,s)}}[\bsj_{s-1}^{(r)}]$. If $s=\ell_r$, then $\bsk_{\ell_r-1}^{(r)}$ has a neighbor $\bsv \in C'$ that is not neighboring to $\bsj_{\ell_r-1}^{(r)}$, but this contradicts our assumption that $\bsk_{\ell_r-1}^{(r)}$ is an endpoint of $P^{(r)}$ by Corollary~\ref{corollary:vertexcycleneighboring} (as in the proof of Corollary~\ref{corollary:pathuniqueendpoints}). Thus, $\Gamma_{C'^{(r,s)}}(\bsk_{s-1}^{(r)})=\Gamma_{C'^{(r,s)}}[\bsj_{s-1}^{(r)}]$ if $s=\ell_r$.
    
    These results hold regardless of whether $C'^{(r,s)}$ is an even hole. We complete the proof by strong induction on the number of single-vertex deformations to obtain $C'^{(r,s)}$. As our base case, we take $r=m$ and $s=\ell_m$, and so $C'^{(r,s)}=C'$, and the result that $\bsk_{\ell_m-1}^{(m)}\prec_{C'}\bsj_{\ell_m-1}^{(m)}$ follows. By our inductive hypothesis, $C'^{(r,s)}\in\avg{C'}$ is an even hole. To prove our inductive statement, suppose that $\bsv$ is in $\Gamma_{C'^{(r,s)}}(\bsk_{s-1}^{(r)}){\setminus}\Gamma_{C'^{(r,s)}}[\bsj_{s-1}^{(r)}]$. We have shown that $\Gamma_{C'^{(r,s)}}(\bsk_{s-1}^{(r)})=\Gamma_{C'^{(r,s)}}[\bsj_{s-1}^{(r)}]$ if $s=\ell_r$, and so assume $s<\ell_r$. Thus, there are vertices $\bsj_s^{(r)}$ and $\bsk_{s}^{(r)}$ in $P^{(r)}$, and $\bsv$ is not in $\{\bsj_{s}^{(r)},\bsk_{s}^{(r)}\}$, since $\bsv$ is not in $\partial\mco$. By our inductive hypothesis, $C'^{(r,s)}=C'^{(r,s+1)}{\oplus}\{\bsj_s^{(r)},\bsk_{s}^{(r)}\}$, and $C^{(r,s+1)}\in\avg{C'}$ is an even hole. Since neither vertex in $\{\bsj_s^{(r)}, \bsk_{s}^{(r)}\}$ is neighboring to $\bsj_{s-1}^{(r)}$, then $\Gamma_{C'^{(r,s+1)}}(\bsk_{s-1}^{(r)})=\Gamma_{C'^{(r,s)}}[\bsj_{s-1}^{(r)}]\cup\{\bsv,\bsj_s^{(r)}\}$. However, this contradicts Lemma~\ref{lemma:vertexcyclerelations}, since $\bsk_{s-1}^{(r)}$ has at least five neighbors in the even hole $C'^{(r,s+1)}$. Therefore, $\Gamma_{C'^{(r,s)}}(\bsk_{s-1}^{(r)})=\Gamma_{C'^{(r,s)}}[\bsj_{s-1}^{(r)}]$, and so $\bsk_{s-1}^{(r)}\prec_{C'^{(r,s)}}\bsj_{s-1}^{(r)}$. This completes the proof.
\end{proof}

\section{Proof of Lemma~\ref*{lemma:searchprocessbasecase}}
\label{section:SearchProcessBaseCaseProof}

\SearchProcessBaseCase*

We shall prove statements~(i)-(iv) separately. 

\subsection{Proof of Statement~(i)}

\begin{proof}[Proof of Lemma~\ref*{lemma:searchprocessbasecase}~$(i)$]
    We prove this statement inductively on the number of path components in $\overrightarrow{\mco}_t$. In the initial step of the obstruction search process, $\overrightarrow{\mco}_t$ is the empty sequence and $\bsu_t=\bsj$. Suppose that $\bsj\prec_C\bsk$ is $(a,c)$ pairing without loss of generality, so that $\overrightarrow{\mco}_t$ is also $(a,c)$ pairing. Let $P$ be the component of $G[C_a{\oplus}C'_c]$ with endpoint $\bsj$ in $C'_c$. If $P\cap(W{\setminus}\{\bsj\})=\varnothing$, then $\bsj$ is untethered with respect to $C$, and we set $\overrightarrow{\mco}_t=(P^{(0)})$ with $P^{(0)}=P$. Clearly $\overrightarrow{\mco}_t$ is updated to a fixed-pairing-type deformation.
    
    For our inductive step, let $\overrightarrow{\mco}_{\bsj,g}=(P^{(j)})_{j=0}^g$ correspond to $\overrightarrow{\mco}_t$ at step $g\geq0$ of our obstruction search process, and suppose that $\overrightarrow{\mco}_t$ is a non-empty deformation with $(a,c)$ pairing type. If $\bsu_t$ is defined, then $\bsu_t$ is in $U^{(r,s)}(\mco_t){\setminus}C^{(g)}$, and $h_{\bsu_t}$ and $h_{P^{(r)}}$ anticommute for $r \leq g$. Without loss of generality, suppose $P^{(r)}$ is a component of $G[C_a{\oplus}C'_c]$. By the proof of Corollary~\ref{corollary:vertexonepathcomponent}, $\bsu_t$ is in $C'_d{\setminus}C$, and so there is a unique component $P$ of $G[C_b{\oplus}C'_d]$ containing $\bsu_t$.
    
    We have that $\bsu_t$ is distinct from $\bsj$, since $\bsj$ is in $C^{(g)}$ as $\bsj=\bsj_0^{(0)}$ and $\bsj$ is not contained in any other path in $\overrightarrow{\mco}_t$ other than $P^{(0)}$ by Corollary~\ref{corollary:vertexonepathcomponent}. Thus, if there is no vertex in $P\cap(W{\setminus}\{\bsj\})$, then $\bsu_t$ is not in $W$. By Lemma~\ref{lemma:pairingrelations}~$(a)$, $\Delta_{C^{(r)}}(\bsu_t)=3$ and $\bsu_t\prec_{C^{(r)}}\bsv$ with $\bsv$ in $C_b$, i.e., $\{\bsu,\bsv\}$ has the same pairing type as $\overrightarrow{\mco}_t$. We shall show that $\bsu_t\prec_{C^{(g)}}\bsv$. If $h_{\bsu_t}$ and $h_{C^{(g)}}$ commute, then $h_{\bsu_t}$ anticommutes with a path operator $h_{P^{(q)}}$ for $r < q \leq g$, and $P^{(q)} \neq P^{(r)}$ by Corollary~\ref{corollary:pathcomponentsdistinct}. However, by Corollary~\ref{corollary:vertexanticommutingdistinctpathsendpoint} with $\overrightarrow{\mco}=\overrightarrow{\mco}'$, $\bsu_t=\bsj_0^{(q)}$, however $\bsu_t$ is contained in at most one path component of $\overrightarrow{\mco}_t$ by  Corollary~\ref{corollary:vertexonepathcomponent}, which contradicts the assumption that $\bsu_t$ is not in $C^{(g)}$. Thus, $h_{\bsu_t}$ and $h_{C^{(r)}}$ anticommute.

    Similarly, if $\bsv$ is not in $C^{(g)}$, then $\bsv=\bsk_{t-1}^{(q)}$ for $r < q \leq g$, and $P^{(q)}$ is a component of $G[C_b{\oplus}C'_d]$, since $\bsv$ is in $C_b{\setminus}C'_d$ by our assumption on $\overrightarrow{\mco}_t$. However, $\bsu_t$ is in $\{\bsj_{t-1}^{(q)},\bsj_t^{(q)}\}$, since $\bsu_t$ is neighboring to  $\bsv$ in $C'_d{\setminus}C_b$, and so $\bsu_t$ is also in $P^{(q)}$. This contradicts the assumption that $\bsu_t$ is not in $C^{(g)}$, since $\bsu_t$ is contained in at most one path component of $\overrightarrow{\mco}_t$ by Corollary~\ref{corollary:vertexonepathcomponent}. Thus, $\bsv$ is in $C^{(g)}$.
    
    If $\bsv$ is not the clone to $\bsu_t$ in $C^{(g)}$, then $\Gamma_{C^{(g)}}(\bsu_t)=\bsv\md\bss\md\bst$. Then $\bst$ is in $C^{(g)}{\setminus}C^{(r)}$, and so $\bst=\bsj_{p-1}^{(q)}$ for some path component $P^{(q)}\in\overrightarrow{\mco}_t$ with $r < q \leq g$. Further, $\bst$ is in $C'_c{\setminus}C_a$, since $\bst$ in not in $C'_d$ as $\bst$ is neighboring to $\bsu_t$ in $C'_d$. Thus, $\bss$ is neighboring to $\bsv \in C_b$, $\bsu_t \in C'_d$, and $\bsv \in C'_c$, and so $\bss$ is in $C_a{\setminus}C'$. However, then $\bss$ is in $\{\bsk^{(q)}_{p-1},\bsk^{(q)}_p\}$ as it is neighboring to $\bst$ in $C_a{\setminus}C'_c$, and so $\bss$ is in $P^{(q)}$. This contradicts the assumption that $\bss$ is in $C^{(g)}$, since it is contained in at most one path component of $\overrightarrow{\mco}_t$ by Lemma~\ref{lemma:vertexonecoloringclass}. Thus, $\bsu_t\prec_{C^{(g)}}\bsv$.
    
    Finally, if $\bsu_t$ has an additional neighbor $\bst$ in $C_b{\setminus}C^{(r)}$, i.e., $\bsu_t$ is not the endpoint of a path component of $G[C_b{\oplus}C'_d]$), then $\bst$ is in $\bsk_{p-1}^{(q)}$ for $q<r$, however then $\bsu_t$ is also in $\{\bsj_{p-1}^{(q)},\bsj_p^{(q)}\}$. This contradicts the assumption that $\bsu_t$ is not in $C^{(g)}$ by Corollary~\ref{corollary:vertexonepathcomponent}. Thus, $\bsu_t\prec_{C^{(g)}}\bsv$, and $\bsu_t$ is the endpoint of a path component $P$ of $G[C_b{\oplus}C'_d]$ if $\bsu_t$ is not in $W{\setminus}\{\bsj\}$.
    
    If there is a vertex $\bsv \in P \cap C_b$ that is not in $C^{(g)}$, then $\bsv=\bsk_{p-1}^{(q)}$ for $q \leq g$ and $P^{(q)}$ is a component of $G[C_b{\oplus}C'_d]$. However, then $\bsu_t=\bsj_{0}^{(q)}$ and this contradicts the assumption that $\bsu_t$ is not in $C^{(g)}$ by Corollary~\ref{corollary:vertexonepathcomponent}. Thus, $P \cap C_b$ is a subset of $C^{(g)}$. Let $P=\bsj_0^{(g+1)}\md\bsk_0^{(g+1)}\md\dots$ with $\bsu_t=\bsj_0^{(g+1)}$ in $C'_d{\setminus}C_b$. Suppose that $t>1$ is the smallest index such that $\bsk_{t-1}^{(g+1)}$ is not the clone to $\bsj_{t-1}^{(g+1)}$ in $C^{(g+1,t-1)}$ (we have proven this for $t=1$). Then $\bsj_{t-1}^{(g+1)}$ is not in $C^{(g)}$, since it is not neighboring to $\bsj_{t-2}^{(g+1)}$, however it is neighboring to its clone $\bsk_{t-2}^{(g+1)}$ in $C^{(g-1,t-2)}$.
    
    If $\bsj_{t-1}^{(g+1)}$ is an endpoint of $P$, i.e., $\bsk_{t-1}^{(g+1)}$ does not exist, then $\bsk^{(g+1)}_{t-2}\md\bss\md\bst $ is a subset of $\Gamma_{C^{(g)}}(\bsj_{t-1}^{(g+1)})$. Since $\bss$ and $\bst$ are neighboring, they both neighbor $\bsj_{t-1}^{(g+1)} \in C'_d$, and neither is contained in $C_b$, then precisely one of them is contained in $C_a$, and the other in $C'_c$. Thus, precisely one of these vertices $\bsv$ is contained in $C'_c{\setminus}C$, since it is neighboring to a vertex in $C_a$ and is not in $C_b$. Similarly, the other vertex $\bsv'$ is contained in $C_a{\setminus}C'$, since it is neighboring to $\bsj_{t-1}^{(g+1)} \in C'_d$ and $\bsv \in C'_c$. Then $\bsv=\bsj_{p-1}^{(q)}$ for a component $P^{(q)}$ of $G[C_a{\oplus}C'_c]$ in $\overrightarrow{\mco}_t$ and $q \leq g$. Thus, $\bsv'$ is also in $P^{(q)}$, and $\bsv'$ is in $\{\bsk_{p-2}^{(q)},\bsk_{p-1}^{(q)}\}$. However, $\bsv'$ is contained in at most one path component of $\overrightarrow{\mco}_t$ by Lemma~\ref{lemma:vertexonecoloringclass}, which contradicts the requirement that $\bsv'$ is in $C^{(g)}$. Thus, $\bsj_{t-1}^{(g+1)}$ has two neighbors $\{\bsk_{t-2}^{(g+1)},\bsk_{t-1}^{(g+1)}\}$ in $C_b$.
    
    If $\bsj_{t-1}^{(g+1)}$ has four neighbors in $C^{(g)}$, but $\bsj_{t-1}^{(g+1)}\nprec_{C^{(g+1,t-1)}} \bsk_{t-1}^{(g+1)}$, then $\Gamma_{C^{(g)}}(\bsj_{t-1}^{(g+1)})=\bsk_{t-2}^{(g+1)}\md\bss\md\bst\md\bsk_{t-1}^{(g+1)}$. By a similar argument, precisely one of $\bss$ and $\bst$ are contained in $C_a$, and the other is contained in $C'_c$, since they are neighboring, both have a neighbor in $C_b$, and neighbor $\bsj_{t-1}^{(g+1)} \in C'_d$. Without loss of generality, suppose that $\bss$ is in $C_a$ and $\bst$ is in $C'_c$. Thus, $\bss$ is in $C_a{\setminus}C'$, since it is neighboring $\bsj_{t-1}^{(g+1)} \in C'_d$ and $\bst \in C'_c$, and $\bst \in C'_c{\setminus}C$ as it is neighboring $\bss \in C_a$ and $\bsk_{t-1}^{(g+1)} \in C_b$. Then $\bst=\bsj_{p-1}^{(q)}$ for a component $P^{(q)}\in\overrightarrow{\mco}_t$ of $G[C_a{\oplus}C'_c]$ with $q \leq g$. However, then $\bss$ is in $\{\bsk_{p-2}^{(q)},\bsk_{p-1}^{(q)}\}$. This contradicts the assumption that $\bss$ is in $C^{(g)}$ by Lemma~\ref{lemma:vertexonecoloringclass}. Thus, if $\bsj_{t-1}^{(g+1)}\nprec_{C^{(g+1,t-1)}}\bsk_{t-1}^{(g+1)}$, then it has three neighbors in $C^{(g)}$, two of which are in $C_b$, i.e., $\bsj_0^{(g+1)}$ is tethered to $\bsj_{t-1}^{(g+1)}$ relative to $C^{(g)}$.
    
    Suppose then that $\bsj_0^{(g+1)}$ is tethered to $\bsj_{t-1}^{(g+1)}$ relative to $C^{(g)}$, and assume that $P\cap(W{\setminus}\{\bsj\})=\varnothing$. For simplicity, let $\bsu=\bsj_{t-1}^{(g+1)}$, then $\bsu$ is distinct from $\bsj$ by the same argument that $\bsu_t$ is distinct from $\bsj$, since because $\bsu$ is not in $C^{(g)}$. Thus, $\bsu$ is not in $W$, and so there is a path operator $P^{(q)}$ with $q \leq g$ such that $h_{\bsu}$ and $h_{P^{(q)}}$ anticommute. If $\bsu=\bsj_0^{(q)}$, then $P^{(q)}$ is a component of $G[C_b{\oplus}C'_d]$, since $\bsu$ is in $C'_d{\setminus}C_b$ by construction of $\overrightarrow{\mco}_t$. However, $\bsu_t$ is also in $P^{(q)}$, which contradicts the requirement that $\bsu_t$ is the endpoint of this path in $C'_d$. Thus, by Lemma~\ref{lemma:pairingrelations}~$(a)$, $\bsu$ is in $U^{(q,p)}$ for some $p$ with $\bsu\prec_{C^{(q)}}\bsv$ for some vertex $\bsv$ in $C_b$. Since $\bsu$ has two neighbors in $C_b$, it has an additional neighbor $\bsv'=\bsk_{u-1}^{(v)}$ in $C_b{\setminus}C'_d$ for a component$P^{(v)}\in\overrightarrow{\mco}_t$ of $G[C_b{\oplus}C'_d]$ with $v<q$. However, then $P^{(v)}=P$, which contradicts our requirement that $\bsu_t$ is not in $C^{(g)}$. Thus, $\bsj_{t-1}^{(g+1)}$ has four neighbors in $C^{(g)}$, and two of which are in $C_b$.
    
    Therefore, $\bsj_{t-1}^{(g+1)}\prec_{C^{(g+1,t-1)}}\bsk_{t-1}^{(g+1)}$ for all $t\in\{1,\dots,\ell_{g+1}\}$. The search process thereby updates $\overrightarrow{\mco}_t$ to $(P^{(j)})_{j=0}^{g+1}$, which is a valid fixed-pairing-type deformation. This completes the proof of statement~(i).
\end{proof}

\subsection{Proof of Statement~(ii)}

\begin{proof}[Proof of Lemma~\ref*{lemma:searchprocessbasecase}~$(ii)$]
    We consider the case where there is a vertex in $P\cap(W{\setminus}\{\bsj\})$ at some step of our obstruction search process. At the initial step, if $P\cap(W{\setminus}\{\bsj\})\neq\varnothing$, then $\bsj$ is tethered to another vertex $\bsj'$ relative to $C$. The search process then returns the empty sequence $\overrightarrow{\mco}_{\bsj}$ with $(\bsj\rightarrow\bsj')$ in $D$. Since $\bsj'$ is also tethered to $\bsj$ with respect to $C$, then $(\bsj'\rightarrow\bsj)$ is in $D$, and the deformations $\overrightarrow{\mco}_{\bsj}$ and $\overrightarrow{\mco}_{\bsj'}$ have opposite pairing type.
    
    Let $\overrightarrow{\mco}_{\bsj,g}=(P^{(j)})_{j=0}^g$ correspond to $\overrightarrow{\mco}_t$ at step $g\geq0$ of the obstruction search process, and assume that $\overrightarrow{\mco}_t$ is $(a,c)$ pairing. If $\bsu_t$ is defined, then $\bsu_t$ is in $U^{(r,s)}(\mco_t){\setminus}C^{(g)}$. Assume without loss of generality that $P^{(r)}\in\overrightarrow{\mco}_t$ is a component of $G[C_a{\oplus}C'_c]$. Thus, $\bsu_t$ is in $C'_d{\setminus}C_b$ by the proof of Corollary~\ref{corollary:vertexonepathcomponent}. If $\bsu_t$ is in $W{\setminus}\{\bsj\}$, then $\bsu_t$ is in $W$. Let $P$ be the unique component of $G[C_b{\oplus}C'_d]$ containing $\bsu_t$, then there is a unique closest vertex in $W{\setminus}\{\bsj\}$ to $\bsu_t$ in $P$ as required, i.e, $\bsu_t$ itself. By Lemma~\ref{lemma:pairingrelations}, $\bsu_t\prec_{C}\bsv$ with opposite pairing type to $\overrightarrow{\mco}_{\bsj}$. Thus, the search process terminates with $\overrightarrow{\mco}_f=\overrightarrow{\mco}_t$, $U_f=\{\bsu_t\}$, and $V_f=\{\bsu_t\}$. Let $\bsu_t=\bsj'$, then $(\bsj\rightarrow\bsj')$ is in $D$, and $\overrightarrow{\mco}_{\bsj}$ has opposite pairing type to $\overrightarrow{\mco}_{\bsj'}$.
    
    If $\bsu_t$ is not in $W{\setminus}\{\bsj\}$, then our proof follows that given previously. We have that $\bsu_t\prec_{C^{(g)}}\bsv$ with $\bsv$ in $C_b$, and $\bsu_t$ is the endpoint of a path component $P$ of $G[C_b{\oplus}C'_d]$. Thus, if there is a vertex in $P\cap(W{\setminus}\{\bsj\})$, then there is a unique closest vertex $\bsw$ to $\bsu_t$ in $P$ as required. Let $P=\bsj_0^{(g+1)}\md\bsk_0^{(g+1)}\md\dots$, and assume that $\bsj_{t-1}^{(g+1)}$ is the closest vertex to $\bsu_t$ in $P$ such that $\bsj_{t-1}^{(g+1)}\nprec_{C^{(g+1,t-1)}}\bsk_{t-1}^{(g+1)}$. As in the proof of statement~(i), $\bsj_{t-1}^{(g+1)}$ has three neighbors in $C^{(g)}$, and two of which are in $C_b$. There is a contradiction if $\bsj_{t-1}^{(g+1)}$ is not in $W{\setminus}\{\bsj\}$, and so $\bsj_{t-1}^{(g+1)}$ is in $W$. If $\bsj_{t-1}^{(g+1)}=\bsw$, then $\bsw\prec_C\bsv$ has opposite pairing type to $\overrightarrow{\mco}_t$, since $\bsw \in C'_d{\setminus}C$ has two neighbors in $C_b$. If $\bsw$ is closer to $\bsu_t$ than $\bsj_{t-1}^{(g+1)}$, then $\bsw$ has two neighbors in $C_b$, and so $\bsw\prec_C\bsv$ has opposite pairing type to $\overrightarrow{\mco}_t$. In either case, the search process terminates with $\overrightarrow{\mco}_f=\overrightarrow{\mco}_t$, $U_f=\{\bsu_t\}$, and $V_f=\{\bsw\}$. Let $\bsw=\bsj'$, then $(\bsj\rightarrow\bsj')$ is in $D$, and $\overrightarrow{\mco}_{\bsj}$ has opposite pairing type to $\overrightarrow{\mco}_{\bsj'}$. This completes the proof of statement~(ii).
\end{proof}

\subsection{Proof of Statement~(iii)}

\begin{proof}[Proof of Lemma~\ref*{lemma:searchprocessbasecase}~$(iii)$]
    Let $\overrightarrow{\mco}_{\bsj}=(P^{(j)}_{\bsj})_{j=0}^m$, $\overrightarrow{\mco}_{\bsj'}=P^{(j)}_{\bsj'})_{j=0}^{m'}$, and let $r$ be the smallest index such that $P_{\bsj}^{(r)}$ is in $\mco_{\bsj}\cap\mco_{\bsj'}$. Further, suppose without loss of generality that $\overrightarrow{\mco}_{\bsj}$ and $\overrightarrow{\mco}_{\bsj'}$ are both $(a,c)$ pairing deformations with $\bsj$ distinct from $\bsj'$.
    
    If $\bsj_0^{(r)}=\bsj$, then $\bsj_0^{(r)}$ is distinct from $\bsj'$. Thus, there is a path $P_{\bsj'}^{(g)}$ such that $h_{\bsj_0^{(r)}}$ and $h_{P_{\bsj'}^{(g)}}$ anticommute with $\bsj_0^{(r)}$ in $U^{(g,t)}(\mco_{\bsj',t})$. However, this cannot be the case if $\bsj_{0}^{(r)}$ is in $W$, since then $\bsj_0^{(r)}\prec_C\bsk_0^{(r)}$ has opposite pairing type to $\overrightarrow{\mco}_{\bsj'}$, and we have assumed that $\overrightarrow{\mco}_{\bsj}$ and $\overrightarrow{\mco}_{\bsj'}$ are both $(a,c)$ pairing. A similar argument holds if $\bsj_0^{(r)}=\bsj'$. Thus, $\bsj_{0}^{(r)}$ is not in $\{\bsj,\bsj'\}$. Then $\bsj_0^{(r)}$ is in $U^{(g,t)}(\mco_{\bsj,t}) \cap U^{(p,q)}(\mco_{\bsj',t})$, and $h_{\bsj_0^{(r)}}$ anticommutes with $h_{P_{\bsj}^{(g)}}$ and $h_{P_{\bsj'}^{(p)}}$ for $g<r$. This is a contradiction to Corollary~\ref{corollary:vertexanticommutingdistinctpaths} unless $P_{\bsj}^{(g)}=P_{\bsj'}^{(p)}$, however this contradicts the assumption that $r$ is the smallest index. Therefore, there is no path component in the intersection $\overrightarrow{\mco}_{\bsj}\cap\overrightarrow{\mco}_{\bsj'}$, and $\overrightarrow{\mco}_{\bsj}$ and $\overrightarrow{\mco}_{\bsj'}$ are disjoint as collections of induced paths. This completes the proof of statement~(iii).
\end{proof}

\subsection{Proof of Statement~(iv)}

\begin{proof}[Proof of Lemma~\ref*{lemma:searchprocessbasecase}~$(iv)$]
    In the setting of statement~(iii), assume that $\bsj''$ obstructs $\overrightarrow{\mco}_{\bsj}$ and $\overrightarrow{\mco}_{\bsj'}$. Let $C_{\bsj}^{(r,s)}$ correspond to the deformation until $\bsj_{s-1}^{(r)}\prec_{C^{(r,s-1)}}\bsk_{s-1}^{(r)}$ of $C$ by $\overrightarrow{\mco}_{\bsj}$, and similarly for $C_{\bsj'}^{(r,s)}$.
    
    If $\bsu$ is in $U_{\bsj} \cup U_{\bsj'}$, then $h_{\bsu}$ anticommutes with $h_{P_{\bsj}^{(r)}}$ and $h_{P_{\bsj'}^{(g)}}$, and $\bsu$ is in $U_{\bsj}^{(r,s)} \cap U_{\bsj'}^{(g,t)}$. Since $P_{\bsj}^{(r)}$ and $P_{\bsj'}^{(g)}$ are distinct by statement~(iii), this contradicts Corollary~\ref{corollary:vertexanticommutingdistinctpaths}. Thus, $\bsj''$ is not in $ U_{\bsj} \cap U_{\bsj'}$ as $U_{\bsj} \cap U_{\bsj'} = \varnothing$.
    
    If $\bsj''$ is not in $U_{\bsj} \cup U_{\bsj'}$, then let $U_{\bsj}=\{\bsu_{\bsj}\}$ and $U_{\bsj'}=\{\bsu_{\bsj'}\}$. Since $\bsu_{\bsj}$ and $\bsu_{\bsj'}$ are distinct, and neither vertex is in $W$, then both vertices are the endpoints of the same path component $P$. Since $\bsj''$ is the closest vertex in $P \cap W$ to both $\bsu_{\bsj}$ and $\bsu_{\bsj'}$, then $\bsj''$ is the only vertex in $P \cap W$. Without loss of generality, let $P=\bsj^{(m+1)}_0\md\bsk^{(m+1)}_0\md\dots\md\bsk^{(m+1)}_{\ell_{m+1}-1}\md\bsj^{(m+1)}_{\ell_{m+1}}$ be a component of $G[C_b{\oplus}C'_d]$, where $\{\bsj_{s-1}^{(m+1)}\}_{s=1}^{\ell_{m+1}+1}$ is a subset of $C'_d{\setminus}C_b$, $\{\bsk_{s-1}^{(m+1)}\}_{s=1}^{\ell_{m+1}}$ is a subset of $C_b{\setminus}C'_d$, $\bsj^{(m+1)}_0=\bsu_{\bsj}$, and $\bsj^{(m+1)}_{\ell_{m+1}}=\bsu_{\bsj'}$. As in the proof of statement~(ii), $P \cap C_b$  is a subset of $C^{(m)}_{\bsj} \cap C^{(m')}_{\bsj'}$. Let $C_{\bsj}^{(m+1,s)}=C_{\bsj}^{(m)}{\oplus}P_{2s-1}$. $P$ cannot have both endpoints in $C'_d$ unless there is a vertex $\bsj_{s-1}^{(m+1)}\nprec_{C_{\bsj}^{(m+1,s-1)}}\bsk_{s-1}^{(m+1)}$. If $\bsj_{s-1}^{(m+1)}$ is the closest such vertex to $\bsj_{0}^{(m+1)}$, then, as in the proof of statements~(i) and (ii), this vertex is $\bsj''$ and has three neighbors in $C_{\bsj}^{(m)}$, and two of which are in $C_b$. Similarly, $\bsu_{\bsj'}$ is tethered to $\bsj''$ relative to $C_{\bsj'}^{(m')}$, and two of the three neighbors to $\bsj''$ in $C_{\bsj'}^{(m')}$ are in $C_b$.
    
    If $\bsj''$ has the same clone in both $C_{\bsj}^{(m)}$ and $C_{\bsj'}^{(m')}$, then let $\bsj''\prec_{C_{\bsj}^{(m)}}\bsv$ denote this clone, and suppose that $\bsj''=\bsj_{p-1}^{(m+1)}$ with $1<p\leq\ell_{m+1}$. Then $\{\bsv,\bsj_{p-2}^{(m+1)},\bsj_{p-1}^{(m+1)},\bsj_{p}^{(m+1)}\}$ induces a claw. Let $\bsj''\prec_{C_{\bsj}^{(m)}}\bsv$ and $\bsj''\prec_{C_{\bsj'}^{(m')}}\bsv'$ for $\bsv$ distinct from $\bsv'$. Let $C$ be such that $\Gamma_C(\bsj'')=\bsb_0\md\bsa_0\md\bsb_1$, i.e., $\bsb_0=\bsk_{p-2}^{(m+1)}$, $\bsb_1=\bsk_{p-1}^{(m+1)}$, and $\bsj''\prec_C\bsa_0$. Then $\bsj''$ is contained in at most one path component in $\mco_{\bsj}\cup\mco_{\bsj'}$ by Corollary~\ref{corollary:vertexonepathcomponent}. Further, $\bsj''=\bsj_{s-1}^{(r)}$ is in $P_{\bsj}^{(r)}$ or $\bsj''=\bsj_{s-1}^{(r)}$ is in $P_{\bsj'}^{(r)}$ for some $r$ and $s$, as $\bsj''$ is in $C'_d{\setminus}C$. Thus, $\bsj''$ is not contained in any path component of $\mco_{\bsj}\cup\mco_{\bsj'}$, since this contradicts the requirement that $\bsj''$ is not in $C_{\bsj}^{(m)} \cup C_{\bsj'}^{(m')}$. If $\bsj''$ anticommutes with the corresponding operator to any path component of $\mco_{\bsj}\cup\mco_{\bsj'}$, then $\bsj''$ is in $U^{(r,s)}_{\bsj}$ or $\bsj''$ is in $U^{(r,s)}_{\bsj'}$ by Lemma~\ref{lemma:pairingrelations}~$(a)$, which contradicts the requirement that $\bsj''$ has three neighbors in both $C_{\bsj}^{(m)}$ and $C_{\bsj'}^{(m')}$ by Lemma~\ref{lemma:pairingrelations} and Corollary~\ref{corollary:vertexanticommutingdistinctpaths}. Therefore $h_{\bsj''}$ commutes with the operator corresponding to any path component in $\mco_{\bsj}\cup\mco_{\bsj'}$ and $\bsj''$ is not contained in any such path. Thus, $\bsj''$ has three neighbors in every even hole $C_{\bsj}^{(r,s)}$ and $C_{\bsj'}^{(r,s)}$ for all $r$ and $s$. Similarly, $\bsb_0$ and $\bsb_1$ are not contained in any path component in $\mco_{\bsj}\cup\mco_{\bsj'}$ as $\bsj''$ is in this same component. If $r$ is the smallest index such that this vertex belongs to $P^{(r)}$, then the vertex is $\bsk_{s-1}^{(r)}$ for some $s$ since it is in $C_b$. Thus, two of the three neighbors to $\bsj''$ in $C_{\bsj}^{(r,s)}$ and $C_{\bsj'}^{(r,s)}$ are $\bsb_0$ and $\bsb_1$ for all $r$ and $s$. In particular, $\Gamma_{C_{\bsj}^{(m)}}=\bsb_0\md\bsv\md\bsb_1$ and $\Gamma_{C_{\bsj'}^{(m')}}=\bsb_0\md\bsv'\md\bsb_1$.
    
    At least one of $\bsv$ and $\bsv'$ is distinct from $\bsv''$, and so $\bsv''$ is contained in a path in either $\overrightarrow{\mco}_{\bsj}$ or $\overrightarrow{\mco}_{\bsj'}$. Without loss of generality, assume that $\bsv''=\bsk_{s-1}^{(r)}$ for a component $P_{\bsj}^{(r)}\in\overrightarrow{\mco}_{\bsj}$ of $G[C_a{\oplus}C'_c]$. If $\bsj''$ is in $U^{(r,s')}$ for any $s'$, then this contradicts Lemma~\ref{lemma:pairingrelations}, since $h_{\bsj''}$ and $h_{P_{\bsj}^{(r)}}$ commute. If $\bsj''$ is neighboring to $\bsk_{s'-1}^{(r)}$ for any $s' \neq s$, then this contradicts the requirement that $\bsj''$ has only one neighbor in $C_a$, however, if $\bsj''$ is neighboring to $\bsj_{s'-1}^{(r)}$ for any $s' \neq s$, then $\bsj''$ is in $U^{(r,s')}$. Thus, $\Gamma_{P_{\bsj}^{(r)}}(\bsj'')=\bsk_{s-1}^{(r)}\md\bsj_{s-1}^{(r)}$ for $h_{\bsj''}$ and $h_{P_{\bsj}^{(r)}}$ to commute. Then $\Gamma_{C^{(r)}}(\bsj'')=\bsb_0\md\bsj_{s-1}^{(r)}\md\bsb_1$, and so $\bsj_{s-1}^{(r)}$ is in $C'_c{\setminus}C$, since it has a neighbor in both $C_a$ and $C_b$. Thus, $\bsj_{s-1}^{(r)}$ is not contained in any other path component in $\overrightarrow{\mco}_{\bsj}$, and $\bsj''\prec_{C_{\bsj}^{(m)}}\bsj_{s-1}^{(r)}$, i.e., $\bsj_{s-1}^{(r)}=\bsv$. Similarly, $\bsv''$ is not contained in any path component of $\overrightarrow{\mco}_{\bsj'}$, since this component would be $P_{\bsj}^{(r)}$, which contradicts statement~(iii). Then $\bsj''\prec_{C_{\bsj}^{(m)}}\bsv$ and $\bsj''\prec_{C_{\bsj'}^{(m')}}\bsv''$. Further, $\bsj_{p-2}^{(m+1)}$ is neighboring to $\bsv$ and not to $\bsv''$, and $\bsj_{p}^{(m+1)}$ is neighboring to $\bsv''$ and not to $\bsv$ by construction. Since $\bsu_{\bsj}$ is tethered to $\bsj''$ relative to $C_{\bsj}^{(m)}$, then $\bsj_{p-2}^{(m+1)}$ is neighboring to $\bsv$, and since $\bsu_{\bsj'}$ is tethered to $\bsj''$ relative to $C_{\bsj'}^{(m')}$, then $\bsj_{p}^{(m+1)}$ is neighboring to $\bsv''$. If either of these vertices neighbors both $\bsv$ and $\bsv''$, then this induces a claw with $\{\bsv,\bsj_{p-2}^{(m+1)},\bsj_{p-1}^{(m+1)},\bsj_{p}^{(m+1)}\}$. Further $\bsv$ is the mutual neighbor to $\bsj_{p-2}^{(m+1)}$ and $\bsj''$ in $C'_c$.
    
    Now, note that $\bsj''$ has an additional neighbor $\bsu$ in $C'_c$. $\bsu$ is not in $C$, since $\Gamma_C(\bsj'')=\bsb_0\md\bsv''\md\bsb_1$, and these vertices all neighbor $\bsv \in C'_c$. Further, $\bsu$ does not neighbor $\bsv''$, since then $\bsu$ is in $P_{\bsj}^{(r)}$, which contradicts the requirement that $\Gamma_{P_{\bsj}^{(r)}}(\bsj'')=\bsv''\md\bsv$. By Lemma~\ref{lemma:evenholecloneneighbor}, $\bsu$ neighbors at least one vertex in $\{\bsb_0,\bsb_1\}$. If $\bsu$ neighbors $\bsb_0$, then it also neighbors $\bsj_{p-2}^{(m+1)}$, otherwise $\{\bsb_0,\bsj_{p-1}^{(m+2)},\bsu,\bsv''\}$ induces a claw. However, $\bsj_{p-2}^{(m+1)}\md\bsv\md\bsj''\md\bsu\md\bsj_{p-2}^{(m+1)}$ induces a hole of length four, which contradicts the requirement that $\abs{C'_d}\geq3$, since $\{\bsu_{\bsj},\bsu_{\bsj'},\bsj''\}$ is a subset of $C'_d$. Thus, $\bsu$ is neighboring to $\bsb_1$ and not to $\bsb_0$, and so $\Gamma_{C'}(\bsb_1)=\bsv\md\bsj''\md\bsu\md\bsj_{p}^{(m+1)}$ by Lemma~\ref{lemma:vertexcyclerelations} and since $\Gamma_{C'_d}(\bsv)=\{\bsj_{p-2}^{(m+1)},\bsj''\}$ and $\bsj_{p-2}^{(m+1)}$ is not neighboring to $\bsb_1$ by construction.
    
    The vertex $\bsj_{p}^{(m+1)}$ has an additional neighbor $\bsu'$ in $C'_c$, and $\bsu'$ is not in $\{\bsv'',\bsb_1\}$, since both of these vertices have neighbors in $C'_d$. Then $\bsu'$ is neighboring to $\bsv''$, otherwise $\{\bsj_p^{(m+1)},\bsu,\bsu',\bsv''\}$ induces a claw. Further, $\Gamma_{C'}(\bsv'')=\{\bsv\md\bsj'',\bsj_{p}^{(m+1)}\md\bsu'\}$ and $\bsu'$ is in $P_{\bsj}^{(r)}$. Additionally, $\bsu'$ is neighboring to $\bsb_0$ by Corollary~\ref{corollary:vertexcycleneighboring}, and $\bsb_1$ has four neighbors in $C'$. Since $\Gamma_{C'}(\bsb_0)=\{\bsj_{p-2}^{(m+2)}\md\bsv\md\bsj'',\bsu'\}$, then $\bsu'$ is neighboring to $\bsj_{p-2}^{(m+1)}$ by Lemma~\ref{lemma:vertexcyclerelations} and since $\bsj''$ is neighboring to $\bsv$ and $\bsu$ in $C'_c$. Then $k'=3$, where $\abs{C'}=2k'$, and $\Gamma_{C'}(\bsb_0)=\bsu'\md\bsj_{p-2}^{(m+2)}\md\bsv\md\bsj''$. Thus, $P=\bsj_{p-2}^{(m+1)}\md\bsb_0\md\bsj''\md\bsb_1\md\bsj_{p}^{(m+1)}$, since there are no additional vertices in $C'_d$ that are contained in this path.
    
    Let $\bsu''$ be the additional neighbor to $\bsb_1$ in $C$. $\bsu''$ is not in $C'$, since all neighbors to $\bsb_1$ in $C'$ are neighboring to $\bsv''$ in $C_a$ with the exception of $\bsu$ which is not in $C$. Then $\bsu''$ is neighboring to $\bsu$, otherwise $\{\bsb_1,\bsv'',\bsu'',\bsu\}$ induces a claw, and $\bsu''$ is neighboring to $\bsj_{p}^{(m+1)}$ by Corollary~\ref{corollary:vertexcycleneighboring} and the fact that $\bsj''$ does not have any additional neighbors in $C$. By applying Lemma~\ref{lemma:evenholecloneneighbor} to $\bsj''\prec_{C}\bsv''$, $\bsu$ has at most one additional neighbor in $C$. If $\bsu$ has an additional neighbor in $C$, then $\bsu$ is in $W$ and this initializes a deformation with $(a,c)$ pairing. However, the deformation path $P^{(g)}$ initialized by $\bsu$ is not contained in either $\overrightarrow{\mco}_{\bsj}$ or $\overrightarrow{\mco}_{\bsj'}$ as then $\bsj''$ is in $U^{(g,1)}$, which contradicts our previous argument. Then $\overrightarrow{\mco}_{\bsj}=(P_{\bsj}^{(r)})$ and $\overrightarrow{\mco}_{\bsj'}$ is the empty sequence. Then $\bsj_{p}^{(m+1)}=\bsj'$ is in $W$, and so $\bsj'$ and $\bsj''$ are tethered with respect to $C$. If $\bsv$ does not initialize $P_{\bsj}^{(r)}$, then it has four neighbors in $C$. If $k=2$, then $\bsv$ is neighboring to $\bsu''$ as it is neighboring every neighbor in $C$. If $k>2$, then $\bsv$ and $\bsu'$ have only two mutual neighbors, $\bsb_0$ and $\bsv''$, and so $\bsv$ is neighboring to $\bsu''$ by Lemma~\ref{lemma:vertexcyclerelations}. This contradicts our assumption that $\bsu$ is not in $P^{(r)}$, since $\bsu$ is a neighbor to $\bsu''$ in $C'_c$. Thus, $\bsv=\bsj$ as it initializes $P^{(r)}_{\bsj}$.
    
    Finally, let $\bsw$ be the additional neighbor to $\bsb_0$ in $C_a$. This vertex must exist, since if $k=2$, then $\bsu'$ neighbors every vertex in $C$ as it does not initialize $P_{\bsj}^{(r)}$, however, then $\bsb_1$ has five neighbors in $C'$. Further, $\bsw$ is not in $C'$, since every neighbor to $\bsb_0$ in $C$ is neighboring to $\bsv''$ in $C_a$ with the exception of $\bsj_{p-2}^{(m+2)}$, however this vertex is neighboring to $\bsv$ and not neighboring to its clone $\bsv''$ in $C$. Then $\bsw$ is neighboring to $\bsj_{p-2}^{(m+1)}$, otherwise $\{\bsb_0,\bsj_{p-2}^{(m+1)},\bsw,\bsj''\}$ induces a claw. By Corollary~\ref{corollary:vertexcycleneighboring}, $\bsw$ is neighboring to $\bsu'$, since $\bsv$ has no additional neighbors in $C$. Then $\bsw$ is in $P_{\bsj}^{(r)}$ and $\bsu'\prec_{C_{\bsj}^{(r,1)}}\bsw$. Thus, $\bsw$ is the endpoint of $P_{\bsj}^{(r)}$. Further, $\bsj_{p-2}^{(m+1)}$ is in $U^{(r,1)}$, and so $\bsj_{p-2}^{(m+1)}$ has no additional neighbors in $C$. Let $\bsw'$ be the additional neighbor to $\bsw$ in $C_b$. This vertex is not in $C'$, since both of the neighbors to $\bsw$ in $C'$ have neighbors in $C_b$, and $\bsw$ has no additional neighbors in $C'$ other than those specified. Further, $\bsw'$ is neighboring to $\bsu'$ for $\bsu'\prec_{C_{\bsj}^{(r,1)}}\bsw$. However, this induces the claw $\{\bsu',\bsj_{p-2}^{(m+1)},\bsw'',\bsv''\}$. This claw exists since $\bsj_{p-2}^{(m+1)}$ does not have any additional neighbors in $C$, and $\bsv''$ has two neighbors, $\bsb_0$ and $\bsb_1$, in $C$. This contradicts our assumption that $\bsj''$ is not in $U_{\bsj} \cup U_{\bsj'}$, and so $\bsj''$ is in $U_{\bsj} \cup U_{\bsj'}$.
    
    We have that $\bsj''$ is in $U_{\bsj} \cup U_{\bsj'}$ and $U_{\bsj} \cap U_{\bsj'} = \varnothing$. Therefore, $\bsj''$ is the only member of exactly one of $U_{\bsj}$ or $U_{\bsj'}$.
    If any vertex $\bsj''$ has three or more incoming arcs in the obstruction graph, then there is a pair of deformations $\overrightarrow{\mco}_{\bsj}$ and $\overrightarrow{\mco}_{\bsj'}$ with arcs incoming to $\bsj''$, such that $\bsj''$ is not in $U_{\bsj} \cup U_{\bsj'}$ or $\bsj''$ is in $U_{\bsj} \cap U_{\bsj'}$. Therefore, no such vertex exists. This completes the proof of statement~(iv).
\end{proof}

\section{Proof of Lemma~\ref*{lemma:searchprocessinduction}}
\label{section:SearchProcessInductionProof}

\SearchProcessInduction*

\begin{proof}
    We prove this lemma in four statements corresponding to the statements of Lemma~\ref{lemma:searchprocessbasecase} as follows.
    \begin{enumerate}
        \item[(i)] The deformation $\overrightarrow{\mco}^{(i+1)}_{\bsj}$ is a possibly empty fixed-pairing-type deformation for all $\bsj$ in $W^{(i+1)}$
        \item[(ii)] The obstruction graph $\mcd^{(i+1)}$ is bipartite with coloring classes given by the pairing types.
        \item[(iii)] The deformations $\{\overrightarrow{\mco}^{(i+1)}_{\bsj}\}_{\bsj \in W^{(i+1)}}$ of a given pairing type are pairwise disjoint as sets of induced paths.
        \item[(iv)] Every vertex in the obstruction graph $\mcd^{(i+1)}$ has at most two incoming arcs. If $(\bsj\rightarrow\bsj'')$ and $(\bsj'\rightarrow\bsj'')$ are in $D^{(i+1)}$, then $\bsj''$ is the only member of exactly one of $U_{\bsj}^{(i+1)}$ or $U_{\bsj'}^{(i+1)}$.
    \end{enumerate}
    We assume that these statements hold at step $i$ of the search process. The proof for each of these statements closely follows that of the corresponding statement in Lemma~\ref{lemma:searchprocessbasecase}.

    Clearly the statement (i) holds if $\overrightarrow{\mco}^{(i+1)}_{\bsj}=\overrightarrow{\mco}^{(i)}_{\bsj}$ as output by Eq.~(\ref{equation:searchprocesstrivialupdate}). We shall prove that $\overrightarrow{\mco}_{\bsj}^{(i+1)}$ as output by Eq.~(\ref{equation:searchprocessupdate}) is a fixed-pairing-type deformation. We first prove this for the concatenation of sequences $\overrightarrow{\mco}_{(\bsj',\bsj)}^{(i)}=(\overrightarrow{\mco}^{(i)}_{\bsj'}\|\overrightarrow{\mco}^{(i)}_{\bsj})$. Let $\overrightarrow{\mco}^{(i)}_{\bsj}=(P^{(j)}_{\bsj})_{j=0}^m$, $\overrightarrow{\mco}^{(i)}_{\bsj'}=(P^{(j)}_{\bsj'})_{j=0}^{m'}$, and $\overrightarrow{\mco}_{(\bsj',\bsj)}^{(i)}=(P_{(\bsj', \bsj)}^{(j)})_{j=0}^{m+m'+1}$, where
    \begin{equation}
    P_{(\bsj',\bsj)}^{(j)} = 
        \begin{cases}
            P^{(j)}_{\bsj'} & j\in\{0,\dots,m'\}, \\
            P^{(j-m'-1)}_{\bsj} & j\in\{m'+1,\dots,m+m'+1\}.
        \end{cases}
    \end{equation}
    If $\overrightarrow{\mco}^{(i)}_{\bsj}$ or $\overrightarrow{\mco}^{(i)}_{\bsj'}$ is empty, then let $m=-1$ or $m'=-1$, respectively. Further, let $\overrightarrow{\mco}_{(\bsj',\bsj),g}^{(i)}=(P_{(\bsj',\bsj)}^{(j)})_{j=0}^{g}$. Finally, let $C_{\bsj}^{(r,s)}$, $C_{\bsj'}^{(r,s)}$, and $C_{(\bsj',\bsj)}^{(r,s)}$ denote the respective even holes upon deforming $C$ by $\overrightarrow{\mco}^{(i)}_{\bsj}$, $\overrightarrow{\mco}^{(i)}_{\bsj'}$, and $\overrightarrow{\mco}^{(i)}_{(\bsj',\bsj)}$ until step $(r,s)$.
    
    By assumption, $\overrightarrow{\mco}_{(\bsj',\bsj),m'}$ is a fixed-pairing-type deformation, and we shall show that $\bsj_{s-1}^{(r)}\prec_{C_{(\bsj',\bsj)}^{(r,s-1)}}\bsk_{s-1}^{(r)}$ for all $s\in\{1,\dots,\ell_r\}$ and all $r\in\{m'+1,\dots,m+m'+1\}$. Suppose this is true for all $g<r$ with $r>m'$. Let 
    \begin{equation}
        P^{(r)}_{(\bsj',\bsj)} = \bsj_0^{(r)}\md\bsk_0^{(r)}\md\dots \bsj_{\ell_r-1}^{(r)}\md\bsk_{\ell_r-1}^{(r)}
    \end{equation}
    be a component of $G[C_a{\oplus}C'_c]$ so that all deformations are $(a,c)$ pairing. Either $\bsj_0^{(r)}$ is in $W$ or $h_{\bsj_{0}^{(r)}}$ and $h_{P^{(q)}}$ anticommute for $\bsj_0^{(r)}$ not in $P^{(q)}$ and $P^{(q)}$ in $\overrightarrow{\mco}_{\bsj}^{(i)}$. In either case $\bsj_0^{(r)}$ is in $C'_c{\setminus}C$ as in the proof of Corollary~\ref{corollary:vertexonepathcomponent}, and so $\bsj_0^{(r)}$ is not contained in any path component in $\overrightarrow{\mco}^{(i)}_{(\bsj',\bsj),r-1}$ by Lemma~\ref{lemma:vertexonecoloringclass} and the assumption that $\overrightarrow{\mco}_{\bsj'}^{(i)}$ and $\overrightarrow{\mco}_{\bsj}^{(i)}$ are disjoint.
    
    If $\bsj_0^{(r)}$ is in $W$, then $\bsj_0^{(r)}\prec_{C}\bsk_0^{(r)}$ by the assumption that $P^{(r)}_{(\bsj',\bsj)}$ is an induced path component of $G[C_a{\oplus}C'_c]$ with $\bsj_0^{(r)}$ as an endpoint. Thus, $h_{\bsj_0^{(r)}}$ does not anticommute with any operator corresponding to a path in $\overrightarrow{\mco}_{\bsj'}^{(i)}$, since $\bsj_0^{(r)}$ is not contained in this path, and this contradicts Lemma~\ref{lemma:pairingrelations}, since $\bsj_0^{(r)}\prec_{C}\bsk_0^{(r)}$ has the same pairing type as $\overrightarrow{\mco}^{(i)}_{\bsj'}$. If $h_{\bsj_{0}^{(r)}}$ and $h_{P^{(q)}}$ anticommute for $\bsj_0^{(r)}$ not in $P^{(q)}$ and $P^{(q)}$ in $\overrightarrow{\mco}_{\bsj}^{(i)}$, then $h_{\bsj_{0}^{(r)}}$ does not anticommute with any operator corresponding to a path in $\overrightarrow{\mco}_{\bsj'}^{(i)}$, since $\bsj_0^{(r)}$ is not contained in this path, and this contradicts Corollary~\ref{corollary:vertexanticommutingdistinctpaths}. Thus, $h_{\bsj_0^{(r)}}$ and $h_{C_{(\bsj',\bsj)}^{(r-1)}}$ anticommute by the assumption that $\overrightarrow{\mco}_{\bsj}$ is a fixed-pairing-type deformation, and $h_{\bsj_0^{(r)}}$ commutes with every operator corresponding to a path in $\overrightarrow{\mco}_{\bsj'}^{(i)}$.
    
    If $\bsk_0^{(r)}$ is not in $C_{(\bsj',\bsj)}^{(r-1)}$, then $\bsk_0^{(r)}=\bsk_{t-1}^{(q)}$ for some $q<r$, however, then $P^{(r)}=P^{(r)}$, which contradicts either the assumption that $\overrightarrow{\mco}^{(i)}_{\bsj}$ is a fixed-paring-type deformation in which path components do not repeat by Corollary~\ref{corollary:pathcomponentsdistinct}, or the assumption that $\overrightarrow{\mco}^{(i)}_{\bsj}$ and $\overrightarrow{\mco}^{(i)}_{\bsj'}$ are disjoint. Thus, $\bsj_0^{(r)}$ is in $C_{(\bsj',\bsj)}^{(r-1)}$. If $\bsk_0^{(r)}$ is not the clone to $\bsj_0^{(r)}$ in $C_{(\bsj',\bsj)}^{(r-1)}$, then $\Gamma_{C_{(\bsj',\bsj)}^{(r-1)}}(\bsj_0^{(r)})=\bsk_0^{(r)}\md\bss\md\bst$. Since $\bst$ is in $\Gamma_{C_{(\bsj',\bsj)}^{(r-1 )}}(\bsj_0^{(r)}){\setminus}\Gamma_{C_{(\bsj',\bsj)}^{(r-1)}}(\bsk_0^{(r)})$, then $\bst $ is in $C_{(\bsj',\bsj)}^{(r-1)}{\setminus}C_{\bsj}^{(r-m'-2)}$. Thus, $\bst=\bsj_{t-1}^{(q)}$ for $P_{\bsj'}^{(q)}$ in $\overrightarrow{\mco}_{\bsj'}$. Since $\bst$ is neighboring to $\bsj_0^{(r)}$ in $C'_c{\setminus}C_a$, then $P_{\bsj'}^{(q)}$ is a component of $G[C_b{\oplus}C'_d]$ with $\bst$ in $C'_d{\setminus}C_b$. Thus, $\bss$ is neighboring to $\bsk_0^{(r)} \in C_a$, $\bsj_0^{(r)} \in C'_c$, and $\bst \in C'_d$, and so $\bss$ is in $C_b{\setminus}C'$, however, then $\bss$ is also in $P_{\bsj'}^{(q)}$. This contradicts the assumption that $\bss$ is in $C_{(\bsj',\bsj)}^{(r-1)}$. Thus $\bsj_0^{(r)}\prec_{C_{(\bsj',\bsj)}^{(r-1)}}\bsk_0^{(r)}$. 
    
    If $s>1$ is such that $\bsk_{s-1}^{(r)}$ is not in $C_{(\bsj',\bsj)}^{(r-1)}$, then $\bsk_{s-1}^{(r)}$ is in $C_{\bsj}^{(r-m'-2)}{\setminus}C_{(\bsj',\bsj)}^{(r-1)}$. Thus $\bsk_{s-1}^{(r)}=\bsk_{t-1}^{(q)}$ for $P_{\bsj'}^{(q)}\in\overrightarrow{\mco}_{\bsj'}$. However, then $P^{(r)}_{(\bsj',\bsj)}=P^{(q)}_{\bsj'}$, which contradicts the assumption that $\overrightarrow{\mco}^{(i)}_{\bsj}$ and $\overrightarrow{\mco}_{\bsj'}^{(i)}$ are disjoint. Thus, $P^{(r)}_{(\bsj',\bsj)} \cap C_a$ is a subset of $C^{(r-1)}_{(\bsj',\bsj)}$.
    
    Let $s>1$ be the smallest index such that $\bsj^{(r)}_{s-1}\nprec_{C_{(\bsj',\bsj)}^{(r,s-1)}}\bsk_{s-1}^{(r)}$. If $\bsj^{(r)}_{s-1}$ has four neighbors in $C_{(\bsj',\bsj)}^{(r-1)}$, then $\Gamma_{C_{(\bsj',\bsj)'}^{(r-1)}}(\bsj_{s-1}^{(r)})=\bsk_{s-2}^{(r)}\md\bss\md\bst\md\bsk_{s-1}^{(r)}$. Since $\bss$ is a neighbor to $\bsk_{s-2}^{(r)}$ that is not a neighbor to $\bsk_{s-1}^{(r)}$, then $\bss$ is in $C_{(\bsj',\bsj)}^{(r-1)}{\setminus}C_{\bsj}^{(r-m'-2)}$. Thus, $\bss=\bsj_{t-1}^{(q)}$ in $C'_d{\setminus}C_b$ for $P_{\bsj'}^{(q)}$ a component of $G[C_b{\oplus}C'_d]$. Then $\bst$ is neighboring to $\bss \in C'_d$, $\bsj_{s-1}^{(r)} \in C'_c$, and $\bsk_{s-1}^{(r)} \in C_a$, and so $\bst$ is in $C_b{\setminus}C'$. However, then $\bst$is also in $P^{(q)}$, which contradicts the assumption that $\bst$ is in $C_{(\bsj',\bsj)}^{(r-1)}$. Thus, if $s>1$ is the smallest index such that $\bsj^{(r)}_{s-1}\nprec_{C_{(\bsj',\bsj)}^{(r,s-1)}}\bsk_{s-1}^{(r)}$, then $\bsj_{s-1}^{(r)}$ has three neighbors in $C_{(\bsj',\bsj)}^{(r-1)}$, and two of which are in $C_b$. That is, $\bsj_{0}^{(r)}$ is tethered to $\bsk_{s-1}^{(r)}$ with respect to $C_{(\bsj',\bsj)}^{(r-1)}$.
    
    Let $\Gamma_{C_{(\bsj',\bsj)}^{(r-1)}}(\bsj_{s-1}^{(r)})=\bsk_{s-2}^{(r)}\md\bss\md\bsk_{s-1}^{(r)}$. Then $\bsj_{s-1}^{(r)}$ anticommutes with an operator corresponding to some path $P^{(q)}_{\bsj'}$ in $\overrightarrow{\mco}_{\bsj'}^{(i)}$. If $\bsj_{s-1}^{(r)}$ is in $P_{\bsj'}^{(q)}$, then $\bsj_{s-1}^{(r)}$ is an endpoint of $P_{\bsj'}^{(q)}$, and thus $\bsj_{s-1}^{(r)}=\bsj_{0}^{(q)}$ is in $C'_c{\setminus}C$ as in the proof of Corollary~\ref{corollary:pathuniqueendpoints}. This contradicts the assumption that $\overrightarrow{\mco}_{\bsj}$ and $\overrightarrow{\mco}_{\bsj'}$ are disjoint. Thus, $\bsj_{s-1}^{(r)}$ is in $U^{(q,t)}(\mco_{\bsj'})$, and $\bsj_{s-1}^{(r)}$ and $\bsj_{t-1}^{(q)}$ are in $C'{\setminus}C$, and their mutual neighbor $\bsv \in C_{\bsj'}^{(q,t-1)}$ is in $C{\setminus}C'$ as in the proof of Corollary~\ref{corollary:vertexonepathcomponent}. Since $P^{(r)}$ is a subset of $G[C_a{\oplus}C'_c]$, then $\bsj_{t-1}^{(q)}$ is in $C'_d$ and $\bsv$ is in $\{\bsk_{s-2}^{(r)},\bsk_{s-1}^{(r)}\}$ with $\bsv$ in $C_a$. Thus, $\bsj_{t-1}^{(q)}$ is in $C_{(\bsj',\bsj)}^{(r-1)}$, since it is not contained in another path component of either $\overrightarrow{\mco}_{\bsj}$ or $\overrightarrow{\mco}_{\bsj'}$, which would contradict the assumption that these deformations are disjoint. Since $\bsj_{t-1}^{(q)}$ is a mutual neighbor to $\bsj_{s-1}^{(r)}$ and $\bsv$ in this hole, then $\bss=\bsj_{t-1}^{(q)}$, and so $\bss$ is in $C'_d$. If $\bsk_{s-1}^{(r)}$ is not in $C_{\bsj'}^{(q)}$, i.e., $\bsv=\bsk_{s-2}^{(q)}$), then $\bsk_{s-1}^{(r)}=\bsj_{v-1}^{(u)}$ is in $P_{(\bsj',\bsj)}^{(u)}$ for $q<u<r$, which requires that $\bsk_{s-1}^{(r)}$ is in $C'$. However, this is a contradiction, since $\bsk_{s-1}^{(r)}$ neighbors $\bsj_{s-1}^{(r)} \in C'_c$ and $\bss \in C'_d$. If $\bsk_{s-1}^{(r)}$ is in $C_{\bsj'}^{(q)}$, then $\bsj_{s-1}^{(r)}$ is not in $U^{(q,t)}(\mco_{\bsj'})$ unless $k=2$ for $|C_{\bsj'}^{(q,t-1)}|=2k$ by Corollary~\ref{corollary:evenholecloneexactlyoneneighbor}. Then $\bsj_{s-1}^{(r)}$ neighbors a mutual neighbor $\bsw$ to $\bsk_{s-1}^{(r)}$ and $\bsj_{s-2}^{(r)}$ in $C_{\bsj'}^{(q,t-1)}$. If $\bsw$ is in $C_{\bsj'}^{(q)}$, then $\Gamma_{C_{\bsj'}^{(q)}}(\bsj_{s-1}^{(r)})=\bsw\md\bsk_{s-2}^{(r)}\md\bss\md\bsk_{s-1}^{(r)}\md\bsw$. In order for $\bsj_{s-1}^{(r)}$ to have three neighbors in $C_{(\bsj',\bsj)}^{(r-1)}$, it must anticommute with the operator corresponding to another path $P_{(\bsj',\bsj)}^{(u)}$ with $q<u<r$, however, then $\bsj_{s-1}^{(r)}=\bsj_{0}^{(u)}$, since $\bsj_{s-1}^{(r)}$ is in $C'_c{\setminus}C$ and by Corollary~\ref{corollary:vertexanticommutingdistinctpathsendpoint}. However, this contradicts the requirement that $\bsj_{s-1}^{(r)}$ is not in $C_{(\bsj',\bsj)}^{(r-1)}$, since $\bsj_{s-1}^{(r)}$ is contained in at most one path component between $\overrightarrow{\mco}_{\bsj}$ and $\overrightarrow{\mco}_{\bsj'}$. If $\bsw$ is not in $C_{\bsj'}^{(q)}$, then $\bsw=\bsk_{t}^{(q)}$. Further, $\{\bsj_{t-1}^{(q)},\bsk_t^{(q)}\}$ is a subset of $\Gamma_{P_{\bsj'}^{(q)}}(\bsj_{s-1}^{(r)})$, which requires that $\bsj_{s-1}^{(r)}$ is neighboring to $\bsj_t^{(q)}$ for $h_{\bsj_{s-1}^{(r)}}$ and $h_{P_{\bsj'}^{(q)}}$ to anticommute. However, then $\Gamma_{C_{\bsj'}^{(q)}}(\bsj_{s-1}^{(r)})=\bsj_{t}^{(q)}\md\bsk_{s-2}^{(r)}\md\bss\md\bsk_{s-1}^{(r)}\md\bsj_t^{(q)}$, and the argument follows as in the case where $\bsw$ is in $C_{\bsj'}^{(q)}$. Therefore, $\bsj_{s-1}^{(r)}$ is not tethered to $\bsj_{0}^{(r)}$ with respect to $C_{(\bsj',\bsj)}^{(r-1)}$, and $\overrightarrow{\mco}_{(\bsj',\bsj)}^{(i)}$ is a fixed-pairing-type deformation.
    
    We now show that $\overrightarrow{\mco}_{\bsj}^{(i+1)}$ is a fixed-pairing-type deformation. To achieve this, we consider a tree $\mct^{(i)}$ whose vertices are labeled by deformations $\overrightarrow{\mco}^{(j)}_{\bsj}$, and $\overrightarrow{\mco}^{(j)}_{\bsj}$ has children $\overrightarrow{\mco}^{(j-1)}_{\bsj}$ and $\overrightarrow{\mco}^{(j-1)}_{\bsj'}$ for $j \leq i$ if these deformations satisfy Eq.~(\ref{equation:searchprocessupdate}). $\overrightarrow{\mco}^{(j)}_{\bsj}$ has only one child $\overrightarrow{\mco}^{(j-1)}_{\bsj}$ for $j \leq i$ if they are related by Eq.~(\ref{equation:searchprocesstrivialupdate}).
    
    Proceeding with the update to $\overrightarrow{\mco}_{\bsj}^{(i+1)}$, recall that $\bsu_{\bsj,f}^{(i)}$ is constructed such that $\bsu_{\bsj,f}^{(i)}$ is in $W$ or $\bsu_{\bsj,f}^{(i)}$ is in $U^{(r,s)}(\mco_{\bsj}^{(i)}){\setminus}C^{(m)}_{\bsj}$ such that $h_{\bsu_{\bsj,f}^{(i)}}$ and $h_{P_{\bsj}^{(r)}}$ anticommute. Suppose that the latter case is true, then if $\bsu_{\bsj,f}^{(i)}$ is also in $W$, then $\bsu_{\bsj,f}^{(i)}$ is not in $W^{(i)}$ by our assumption that $\bsu_{\bsj, f}^{(i)}$ is not the obstructing vertex to $\overrightarrow{\mco}_{\bsj}^{(i)}$. By Lemma~\ref{lemma:pairingrelations}~$(a)$, the pairing type of $\bsu_{\bsj,f}^{(i)}$ with its clone in $C$ is opposite to that of $\overrightarrow{\mco}_{\bsj}^{(i)}$, and so $\bsu_{\bsj,f}^{(i)}$ is contained in the obstruction set of $W$. However, a vertex in the obstruction set of $W$ can only be removed if there is a deformation $\overrightarrow{\mco}_{\bsv}^{(j)}$ such that $\bsu_{\bsj,f}^{(i)}$ is in $U^{(q,t)}(\mco_{\bsv}^{(j)})$ where $h_{\bsu_{\bsj,f}^{(i)}}$ and $h_{P_{\bsv}^{(q)}}$ anticommute for some $j<i$. If $\overrightarrow{\mco}_{\bsv}^{(j)}$ is not a descendant of $\overrightarrow{\mco}_{\bsj}^{(i)}$ in $\mct^{(i)}$, then $P_{\bsv}^{(q)}$ is not in $\overrightarrow{\mco}_{(\bsj',\bsj)}^{(i)}$ by the assumption that the deformations $\overrightarrow{\mco}_{\bsj}^{(i)}$ are pairwise disjoint for all $\bsj$ in $W^{(i)}$. Thus, we have a contradiction to Corollary~\ref{corollary:vertexanticommutingdistinctpaths}. If $\overrightarrow{\mco}_{\bsv}^{(j)}$ is a descendant of $\overrightarrow{\mco}_{\bsj}^{(i)}$ in $\mct^{(i)}$, then this contradicts Corollary~\ref{corollary:vertexanticommutingdistinctpaths} unless $P_{\bsj}^{(r)}=P_{\bsv}^{(q)}$. Then there is no vertex $\bsu$ in $U^{(g,f)}(\mco_{\bsj}^{(i)}){\setminus}C_{\bsj}^{(m)}$ such that $h_{\bsu}$ and $h_{P_{\bsj}^{(g)}}$ anticommute for $\mcl_{g,f}>\mcl_{q,t}$. This contradicts the assumption that $\bsu_{\bsj,f}^{(i)}$ is not in $C_{\bsj}^{(m)}$, since $\bsu_{\bsj,f}^{(i)}$ is in $P_{\bsj}^{(r+1)}$, and $\bsu_{\bsj,f}^{(i)}$ is contained in at most one path component of $\overrightarrow{\mco}^{(i)}_{\bsj}$ by Corollary~\ref{corollary:vertexonepathcomponent}. If $r=m$, then $\bsu^{(i)}_{\bsj,f}$ is in $W^{(i)}$. Thus, if $\bsu_{\bsj,f}^{(i)}$ is in $U^{(r,s)}(\mco_{\bsj}^{(i)}){\setminus}C^{(m)}_{\bsj}$, and $h_{\bsu_{\bsj,f}^{(i)}}$ and $h_{P_{\bsj}^{(r)}}$ anticommute, then $\bsu_{\bsj,f}^{(i)}$ is not in $W$. Then $\bsu_{\bsj,f}^{(i)}=\bsj$, since this is the only instance in the search process in which $\bsu_{\bsj, f}^{(i)}$ is not assigned to the only member of $U^{(r,s)}(\mco_{\bsj}^{(i)}){\setminus}C_{\bsj}^{(m)}$.
    
    The remainder of the proof follows the proof of Lemma~\ref{lemma:searchprocessbasecase} with the observation that if there is a vertex $\bsu$ whose operator anticommutes with a path operator in $\mco^{(i)}_{t}$ and is not an element of that path, then if $\bsu$ is not in $W^{(i)}$, and so $\bsu$ is not in $W$. This completes the proof.
\end{proof}

\section{Proof of Lemma~\ref*{lemma:generalizedcharacteristicpolynomial}}
\label{section:GeneralizedCharacteristicPolynomialProof}

\GeneralizedCharacteristicPolynomial*

\begin{proof}
    By definition,
    \begin{align}
        Z_{G}(-u^2) &= T_G(u)T_G(-u) \\
        &= \left(\sum_{s=0}^{\alpha(G)}(-u)^s\qkg{s}{G}\right)\left(\sum_{t=0}^{\alpha(G)}u^t\qkg{t}{G}\right) \\
        &= \sum_{s,t=0}^{\alpha(G)}(-1)^su^{s+t}\qkg{s}{G}\qkg{t}{G}.
    \end{align}
    If $s+t=1\pmod2$, then, by Theorem~\ref{theorem:conservedcharges},
    \begin{equation}
        (-1)^s\qkg{s}{G}\qkg{t}{G}+(-1)^t\qkg{t}{G}\qkg{s}{G} = (-1)^s[\qkg{s}{G},\qkg{t}{G}] = 0.
    \end{equation}
    Hence,
    \begin{align}
        Z_{G}(-u^2) &= \sum_{\substack{s,t=0 \\ s+t=0\pmod2}}^{\alpha(G)}(-1)^su^{s+t}\qkg{s}{G}\qkg{t}{G} \\
        &= \sum_{\substack{S,T\in\mcs_{G} \\ \abs{S}+\abs{T}=0\pmod2}}(-1)^\abs{S}u^{\abs{S}+\abs{T}}h_Sh_T \\
        &= \sum_{\substack{S,T\in\mcs_{G} \\ \abs{S}+\abs{T}=0\pmod2}}(-1)^\abs{S}u^{\abs{S}+\abs{T}}\left(h_{S \cap T}\right)^2h_{S{\setminus}T}h_{T{\setminus}S}.
    \end{align}
    As in the proof of Theorem~\ref{theorem:conservedcharges}~(i), we use the fact that every factor $h_{\bsj}$ with $\bsj \in S \cap T$ commutes with every factor in $h_Sh_T$. Commuting a given factor of $h_{\bsj}$ for $\bsj \in S{\setminus}T$ through $h_{T{\setminus}S}$ gives a factor of $(-1)$ for every neighbor to $\bsj$ in $T{\setminus}S$. Using
    \begin{equation}
        S{\oplus}T = (S{\setminus}T)\cup(T{\setminus}S),
    \end{equation}
    gives
    \begin{equation}
        h_{S{\setminus}T}h_{T{\setminus}S} = (-1)^{\abs{E[S{\oplus}T]}}h_{T{\setminus}S}h_{S{\setminus}T},
    \end{equation}
    since $G[S{\oplus}T]$ is bipartite by Lemma~\ref{lemma:clawfreesymmetricdifference}. Thus, if $\abs{E[S{\oplus}T]}=1\pmod2$, then $h_{S{\setminus}T}h_{T{\setminus}S}+h_{T{\setminus}S}h_{S{\setminus}T}=0$. Hence,
    \begin{align}
        Z_{G}(-u^2) &= \sum_{\substack{S,T\in\mcs_{G} \\ \abs{S}+\abs{T}=0\pmod2 \\ \abs{E[S{\oplus}T]}=0\pmod2}}(-1)^\abs{S}u^{\abs{S}+\abs{T}}\left(h_{S \cap T}\right)^2h_{S{\setminus}T}h_{T{\setminus}S} \\
        &= \sum_{\substack{S,T\in\mcs_{G} \\ \abs{S}+\abs{T}=0\pmod2 \\ \abs{E[S{\oplus}T]}}=0\pmod2}(-1)^\abs{S}u^{\abs{S}+\abs{T}}h_Sh_T.
    \end{align}
    By Lemma~\ref{lemma:clawfreesymmetricdifference}, $G[S{\oplus}T]$ is a disjoint union of paths and even holes. Suppose $G[S{\oplus}T]$ contains a path component $P$. Since $P$ either has odd-many vertices or odd-many edges, it cannot be the only component of $G[S{\oplus}T]$. Define
    \begin{align}
        \wt{S} &= S{\oplus}P, \\
        \wt{T} &= T{\oplus}P.
    \end{align}
    This gives distinct independent sets $\wt{S}$ and $\wt{T}$ for which $S{\oplus}T=\wt{S}{\oplus}\wt{T}$ and $\abs{S}+\abs{T}=|\wt{S}|+|\wt{T}|$. This gives
    \begin{align}
        (-1)^{|\wt{S}|}u^{|\wt{S}|+|\wt{T}|}h_{\wt{S}}h_{\wt{T}} &= (-1)^{|\wt{S}|}u^{|\wt{S}|+|\wt{T}|}h_{S{\setminus}P}h_{T \cap P}h_{S \cap P}h_{T{\setminus}P} \\ 
        &= (-1)^{|\wt{S}|+\abs{E[P]}}u^{|\wt{S}|+|\wt{T}|}h_{S{\setminus}P}h_{S \cap P}h_{T \cap P}h_{T{\setminus}P}, \\
        &= -(-1)^{\abs{S}}u^{\abs{S}+\abs{T}}h_{S}h_{T}. 
    \end{align}
    Hence,
    \begin{equation}
        (-1)^{\abs{S}}u^{\abs{S}+\abs{T}}h_{S}h_{T}+(-1)^{|\wt{S}|}u^{|\wt{S}|+|\wt{T}|}h_{\wt{S}}h_{\wt{T}} = 0,
    \end{equation}
    where we have used the fact that, if $P$ has odd length, then $|\wt{S}|=\abs{S}$, and if $P$ has even length, then $|\wt{S}|=\abs{S} \pm 1$. Thus, $(-1)^{|\wt{S}|+\abs{E[P]}} = -(-1)^{\abs{S}}$. For a given collection of pairs $(S,T)$ such that $G[S{\oplus}T]$ is fixed, we can choose a path component by which to pair terms to cancel. Therefore, $G[S{\oplus}T]$ contains no path components, so it must be a collection of disjoint and non-neighboring even holes. In this case $\abs{S}=\abs{S}$, and we have
    \begin{align}
        Z_{G}(-u^2) &= \sum_{\substack{S, T \in \mcs_{G} \\ S{\oplus}T=\partial\mcx \\ \mcx\in\mathscr{C}^{\text{(even)}}_G}}(-u^2)^{\abs{S}}h_Sh_T, \label{equation:setsymmetric} \\
        Z_{G}(-u^2) &= \sum_{\mcx\in\mathscr{C}^{\text{(even)}}_{G}}(-u^2)^{\abs{\partial\mcx}/2}2^{\abs{\mcx}}I_{G{\setminus}\Gamma[\mcx]}(-u^2)\prod_{C\in\mcx}h_C.
    \end{align}
    By Lemma~\ref{lemma:evenholeneighbor}, this gives
    \begin{equation}
        Z_{G}(-u^2) = \sum_{\avg{\mcx}\in\avg{\mathscr{C}_{G}^{\text{(even)}}}}(-u^2)^{\abs{\partial\avg{\mcx}}/2}2^{\abs{\mcx}}I_{G{\setminus}\Gamma[\avg{\mcx}]}(-u^2)\prod_{\avg{C_0}\in\avg{\mcx}}\jkg{C_0}{G},
    \end{equation}
    completing the proof.
\end{proof}

\section{Proof of Lemma~\ref*{lemma:fundamentalidentity}}
\label{section:FundamentalIdentityProof}

\FundamentalIdentity*

To prove Lemma~\ref{lemma:fundamentalidentity}, we first require the following lemma.
\begin{lemma}
    \label{lemma:kjrecursion}
    Let $K_s$ be a simplicial clique in $G$, and define $K_{\bsj} = \Gamma[\bsj]{\setminus}(K_s{\setminus}\{\bsj\})$ to be the clique such that $\Gamma[\bsj] = K_s \cup K_{\bsj}$. Then the following recursion relation holds.
    \begin{equation}
        \qkg{k}{G} = \qkg{k}{G{\setminus}K_s}+\sum_{\bsj \in K_s}\qkg{k-1}{G{\setminus}K_{\bsj}}h_{\bsj}. \notag
    \end{equation}
\end{lemma}
\begin{proof}
    By Eq.~(\ref{equation:chargecliquerecursion}), we have
    \begin{equation}
        \qkg{k}{G} = \qkg{k}{G{\setminus}K_s}+\sum_{\bsj \in K_s}\qkg{k-1}{G{\setminus}(K_s \cup K_{\bsj})}h_{\bsj},
    \end{equation}
    and
    \begin{equation}
    \qkg{k-1}{G{\setminus}K_{\bsj}} = \qkg{k-1}{G{\setminus}(K_s \cup K_{\bsj})}+\sum_{\bsk \in K_s{\setminus}\{\bsj\}}\qkg{k-2}{G{\setminus}(K_s \cup K_{\bsj} \cup K_{\bsk})}h_{\bsk}.
    \end{equation}
    This gives
    \begin{align}
        \qkg{k}{G} &= \qkg{k}{G{\setminus}K_s}+\sum_{\bsj \in K_s}\left(\qkg{k-1}{G{\setminus}K_{\bsj}}-\sum_{\bsk \in K_s{\setminus}\{\bsj\}}\qkg{k-2}{G{\setminus}(K_s \cup K_{\bsj} \cup K_{\bsk})}h_{\bsk}\right)h_{\bsj} \\
        &= \qkg{k}{G{\setminus}K_s}+\sum_{\bsj \in K_s}\qkg{k-1}{G{\setminus}K_{\bsj}}h_{\bsj}-\sum_{\substack{\bsj,\bsk \in K_s \\ \bsj\neq\bsk}}\qkg{k-2}{G{\setminus}(K_s \cup K_{\bsj} \cup K_{\bsk})}h_{\bsk}h_{\bsj}.
    \end{align}
    The third term vanishes since $\qkg{k-2}{G{\setminus}(K_s \cup K_{\bsj} \cup K_{\bsk})}$ is symmetric in $\bsj$ and $\bsk$, but $h_{\bsj}$ and $h_{\bsk}$ anticommute for $\bsj,\bsk \in K_s$ and $\bsj\neq\bsk$. Therefore
    \begin{equation}
        \qkg{k}{G} = \qkg{k}{G{\setminus}K_s}+\sum_{\bsj \in K_s}\qkg{k-1}{G{\setminus}K_{\bsj}}h_{\bsj},
    \end{equation}
    completing the proof.
\end{proof}

\begin{proof}[Proof of Lemma~\ref*{lemma:fundamentalidentity}]
    By Def.~\ref{definition:generalizedcharacteristicpolynomial}, it is sufficient to show that
    \begin{equation}
        \left(1+u\sum_{\bsj \in K_s}h_{\bsj}\right)\chi T_{G}(-u) = T_{G}(-u)\left(1-u\sum_{\bsj \in K_s}h_{\bsj}\right)\chi.
    \end{equation}
    By equating coefficients of $u^k$, this is equivalent to showing that
    \begin{equation}
        \chi \qkg{k}{G}+\sum_{\bsj \in K_s}h_{\bsj}\chi\qkg{k-1}{G} = \left(\qkg{k}{G}-\qkg{k-1}{G}\sum_{\bsj \in K_s}h_{\bsj}\right)\chi.
    \end{equation}
    We expand the left-hand side by applying Eq.~(\ref{equation:chargecliquerecursion}) to the clique $K_s$ in the first term and the clique $K_{\bsj}=\Gamma[G]{\setminus}(K_s{\setminus}\{\bsj\})$ in the second term. This gives
    \begin{equation}
        \chi\qkg{k}{G}+\sum_{\bsj \in K_s}h_{\bsj}\chi\qkg{k-1}{G} = 
        \chi\left(\qkg{k}{G{\setminus}K_s}+\sum_{\bsj' \in K_s}h_{\bsj'}\qkg{k-1}{G{\setminus}(K_s \cup K_{\bsj'})}\right)+\sum_{\bsj \in K_s}h_{\bsj}\chi\left(\qkg{k-1}{G{\setminus}K_{\bsj}}+\sum_{\bsj' \in K_{\bsj}}h_{\bsj'}\qkg{k-2}{G{\setminus}\Gamma[\bsj']}\right).
    \end{equation}
    For $\bsj \in K_s$, we see that $h_{\bsj}\chi$ only anticommutes with $h_{\bsk}$ if $\bsk$ is in $K_{\bsj}$. Thus
    \begin{align}
        \chi\qkg{k}{G}+\sum_{\bsj \in K_s}h_{\bsj}\chi\qkg{k-1}{G} &= 
        \left(\qkg{k}{G{\setminus}K_s}-\sum_{\bsj' \in K_s}h_{\bsj'}\qkg{k-1}{G{\setminus}(K_s \cup K_{\bsj'})}\right)\chi+\sum_{\bsj \in K_s}\left(\qkg{k-1}{G{\setminus}K_{\bsj}}-\sum_{\bsj' \in K_{\bsj}}h_{\bsj'}\qkg{k-2}{G{\setminus}\Gamma[\bsj']}\right)h_{\bsj}\chi \\
        &= 
        \left(\qkg{k}{G{\setminus}K_s}+\sum_{\bsj \in K_s}\qkg{k-1}{G{\setminus}K_{\bsj}}h_{\bsj}\right)\chi-\sum_{\bsj \in K_s}\left(\qkg{k-1}{G{\setminus}(K_s \cup K_{\bsj})}+\sum_{\bsj' \in K_{\bsj}}h_{\bsj'}\qkg{k-2}{G{\setminus}\Gamma[\bsj']}\right)h_{\bsj}\chi \\
        &= 
        \qkg{k}{G}\chi-\sum_{\bsj \in K_s}\left(\qkg{k-1}{G{\setminus}(K_s \cup K_{\bsj})}+\qkg{k-1}{G}-\qkg{k-1}{G{\setminus}K_{\bsj}}\right)h_{\bsj}\chi,
    \end{align}
    and
    \begin{equation}
        \qkg{k-1}{G{\setminus}(K_s \cup K_{\bsj})}-\qkg{k-1}{G{\setminus}K_{\bsj}} = -\sum_{\bsj' \in K_s{\setminus}\{\bsj\}}\qkg{k-2}{G{\setminus}(K_s \cup K_{\bsj} \cup K_{\bsj'})}h_{\bsj'}.
    \end{equation}
    This gives
    \begin{equation}
        \chi\qkg{k}{G}+\sum_{\bsj \in K_s}h_{\bsj}\chi\qkg{k-1}{G} = \qkg{k}{G}\chi-\qkg{k-1}{G}\sum_{\bsj \in K_s}h_{\bsj}\chi+\sum_{\substack{\bsj, \bsj' \in K_s \\ \bsj\neq\bsj'}}\qkg{k-2}{G{\setminus}(K_s \cup K_{\bsj} \cup K_{\bsj'})}h_{\bsj'}h_{\bsj}\chi.
    \end{equation}
    In the last term, $\qkg{k-2}{G{\setminus}(K_s \cup K_{\bsj} \cup K_{\bsj'})}$ is symmetric in $\bsj$ and $\bsj'$, but $h_{\bsj}$ and $h_{\bsj'}$ anticommute for $\bsj$, $\bsj' \in K_s$ with $\bsj\neq\bsj'$. Thus, this term vanishes, and we have
    \begin{equation}
        \chi\qkg{k}{G}+\sum_{\bsj \in K_s}h_{\bsj}\chi\qkg{k-1}{G} = \left(\qkg{k}{G}-\qkg{k-1}{G}\sum_{\bsj \in K_s}h_{\bsj}\right)\chi,
    \end{equation} 
    completing the proof.
\end{proof}

\section{Proof of Lemma~\ref*{lemma:ladderanticommutationrelations}}
\label{section:LadderAnticommutationRelationsProof}

\LadderAnticommutationRelations*

\begin{proof}
    We have
    \begin{align}
        T_{G}(u)\psi_{\mcj,+j}T_{G}(-u) &= \frac{1}{N_{\mcj,j}}T_{G}(u)\left[\Pi_{\mcj}T_{G}(-u_{\mcj,j}) \chi T_{G}(u_{\mcj,j})\right]T_{G}(-u) \\
        &= \frac{\Pi_{\mcj}}{N_{\mcj,j}}T_{G}(-u_{\mcj,j})\left[T_{G}(u) \chi T_{G}(-u)\right]T_{G}(u_{\mcj,j}) \\
        &= \frac{\Pi_{\mcj}}{N_{\mcj,j}}T_{G}(-u_{\mcj, j})\left[Z_{G}(-u^2)\left(1-u\sum_{\bsj \in K_s}h_{\bsj}\right)\chi-T_{G}(u)\left(u\sum_{\bsj \in K_s}h_{\bsj}\right) \chi T_{G}(-u)\right]T_{G}(u_{\mcj,j}), \notag
    \end{align}
    where we have applied Lemma~\ref{lemma:fundamentalidentity} in the last line. From our proof of Lemma~\ref{lemma:commutationrelation}, we have
    \begin{align}
        T_{G}(u)\psi_{\mcj,+j}T_{G}(-u) &= \frac{1}{N_{\mcj,j}}\Pi_{\mcj}T_{G}(-u_{\mcj,j})\left\{Z_{G}(-u^2)\chi-\frac{u}{2}Z_{G}(-u^2)[H,\chi]-\frac{u}{2}T_{G}(u)[H,\chi]T_{G}(-u)\right\}T_{G}(u_{\mcj,j}) \\
        &= Z_{G}(-u^2)\psi_{\mcj,+j}-\frac{u}{2}Z_{G}(-u^2)[H,\psi_{\mcj,+j}]-\frac{u}{2}T(u)[H,\psi_{\mcj,+j}]T_{G}(-u) \\
        &= Z_{G}(-u^2)\psi_{\mcj,+j}-\frac{u}{u_{\mcj,j}}Z_{G}(-u^2)\psi_{\mcj,+j}-\frac{u}{u_{\mcj,j}}T(u)\psi_{\mcj, +j}T_{G}(-u) \\
        &= \frac{1}{u_{\mcj,j}}\left[Z_{G}(-u^2)\left(u_{\mcj,j}-u\right)\psi_{\mcj,+j}-u T_{G}(u)\psi_{\mcj,+j}T_{G}(-u)\right].
    \end{align}
    Now, by rearranging and applying $Z^{-1}_{G}(-u^2)T_G(u)$ to the right on both sides, we have
    \begin{equation}
        (u_{\mcj,j}+u)T_{G}(u)\psi_{\mcj,+j} = (u_{\mcj,j}-u)\psi_{\mcj,+j}T_{G}(u).
    \end{equation}
    This requires choosing $u \neq \pm u_{\mcj,j}$ for any pair $(\mcj,j)$, as $Z_{G}(-u_{\mcj,j})$ is not invertible. Applying this allows us to write
    \begin{equation}
        \{\psi_{\mcj,+j},T_{G}(u)\chi T_{G}(-u)\} = \frac{u_{\mcj,j}+u}{u_{\mcj,j}-u}T_{G}(u)\{\psi_{\mcj,+j},\chi\}T_{G}(-u).
    \end{equation}
    We now compute the anticommutation relation $\{\psi_{\mcj,+j},\chi\}$. By applying Eq.~(\ref{equation:transfermatrixrecursion}), we obtain
    \begin{align}
        \{\psi_{\mcj,+j},\chi\} &= \frac{\Pi_{\mcj}}{N_{\mcj,j}}\left[T_G(-u_{\mcj,j}) \chi T_G(u_{\mcj,j}) \chi + \chi T_G(-u_{\mcj,j}) \chi T_G(u_{\mcj,j})\right] \\
        &= \frac{2\Pi_{\mcj}}{N_{\mcj,j}}\left[Z_{G\setminus K}(-u_{\mcj, j}^2)-\left(-u_{\mcj,j}\sum_{\bsj \in K}h_{\bsj}T_{G{\setminus}\Gamma[\bsj]}(-u_{\mcj,j})\right)\left(u_{\mcj,j}\sum_{\bsj \in K}h_{\bsj}T_{G{\setminus}\Gamma[\bsj]}(u_{\mcj,j})\right)\right].
    \end{align}
    Now, by applying Eq.~(\ref{equation:transfermatrixrecursion}) again, we have
    \begin{align}
        \{\psi_{\mcj,+j},\chi\} &= \frac{2\Pi_{\mcj}}{N_{\mcj,j}}\left[2Z_{G\setminus K}(-u_{\mcj, j}^2)-Z_{G}(-u_{\mcj, j}^2)\right] \\
        &= \frac{4\Pi_{\mcj}}{N_{\mcj,j}}Z_{G\setminus K}(-u_{\mcj, j}^2).
    \end{align}
    By setting
    \begin{equation}
        N_{\mcj,j} = 4u_{\mcj,j}\left(Z_{G{\setminus}K_s}(-u_{\mcj, j}^2)\frac{\partial Z_{G}(x)}{\partial x}\Bigr|_{x=-u_{\mcj,j}^2}\right)^\frac{1}{2},
    \end{equation}
    we have
    \begin{equation}
        \{\psi_{\mcj,+j}, \psi_{\mcj',-k}\} = \delta_{\mcj,\mcj' }\delta_{jk}\Pi_{\mcj},
    \end{equation}
completing the proof.
\end{proof}

\section{Proof of Lemma~\ref*{lemma:diagonalrelation}}
\label{section:DiagonalRelationProof}

\DiagonalRelation*

\begin{proof}
    By following a similar analysis to the proof of Lemma~\ref{lemma:ladderanticommutationrelations}, the commutator $[\psi_{+j},\psi_{-j}]$ may be expressed as
    \begin{align}
        [\psi_{\mcj,+j},\psi_{\mcj,-j}] &= \frac{1}{N_{\mcj,j}}\Pi_{\mcj}\lim_{u \to u_{\mcj,j}}\left(\frac{u_{\mcj,j}+u}{u_{\mcj,j}-u}T(u)[\psi_{\mcj,+j},\chi]T(-u)\right). \\
        &= -\frac{2u_{\mcj,j}}{N_{\mcj,j}}\Pi_{\mcj}\left(\frac{\partial T_{G}(u_{\mcj,j})}{\partial u_{\mcj,j}}[\psi_{\mcj,+j},\chi]T(-u_{\mcj,j})+T(u_{\mcj,j})[\psi_{\mcj,+j},\chi]\frac{\partial T_{G}(-u_{\mcj,j})}{\partial u_{\mcj,j}}\right) \\
        &= -\frac{2u_{\mcj,j}}{N_{\mcj,j}}\Pi_{\mcj}\left(\frac{\partial T_{G}(u_{\mcj,j})}{\partial u_{\mcj,j}}\psi_{\mcj,+j}\chi T(-u_{\mcj,j})-T(u_{\mcj,j})\chi\psi_{\mcj,+j}\frac{\partial T_{G}(-u_{\mcj,j})}{\partial u_{\mcj,j}}\right).
    \end{align}
    Now, by using the anticommutation relation from the proof of Lemma~\ref{lemma:ladderanticommutationrelations}, we obtain
    \begin{align}
        [\psi_{\mcj,+j},\psi_{\mcj,-j}] &= -\frac{8u_{\mcj,j}}{N_{\mcj,j}^2}\Pi_{\mcj}\left(T(-u_{\mcj,j})\frac{\partial T_{G}(u_{\mcj,j})}{\partial u_{\mcj,j}}-T(u_{\mcj,j})\frac{\partial T_{G}(-u_{\mcj,j})}{\partial u_{\mcj,j}}\right) \\
        &= -\frac{1}{2u_{\mcj,j}}\Pi_{\mcj}\left(\frac{\partial Z_{G}(x)}{\partial x}\Bigr|_{x=-u_{\mcj,j}^2}\right)^{-1}\left(T(-u_{\mcj,j})\frac{\partial T_{G}(u_{\mcj,j})}{\partial u_{\mcj,j}}-T(u_j)\frac{\partial T_{G}(-u_{\mcj,j})}{\partial u_{\mcj,j}}\right).
    \end{align}
    Then,
    \begin{equation}
        \sum_{j=1}^{\alpha(G)}\varepsilon_{\mcj,j}[\psi_{\mcj,+j},\psi_{\mcj,-j}] = -\sum_{j=1}^{\alpha(G)}\frac{1}{2u_{\mcj,j}^2}\Pi_{\mcj}\left(\frac{\partial Z_{G}(x)}{\partial x}\Bigr|_{x=-u_{\mcj,j}^2}\right)^{-1}\left(T_{G}(-u_{\mcj,j})\frac{\partial T_{G}(u_{\mcj,j})}{\partial u_{\mcj,j}}-T_{G}(u_{\mcj,j})\frac{\partial T_{G}(-u_{\mcj,j})}{\partial u_{\mcj,j}}\right). \notag
    \end{equation}
    By using
    \begin{equation}
       \frac{\partial Z_{G}(x)}{\partial x}\Bigr|_{x=-u_{\mcj,j}^2} = \frac{1}{u_{\mcj,j}^2}\prod_{\substack{k=1 \\ k \neq j}}^{\alpha(G)}\left(\frac{u_{\mcj,k}^2-u_{\mcj,j}^2}{u_{\mcj,k}^2}\right),
    \end{equation}
    we have
    \begin{equation}
        \sum_{j=1}^{\alpha(G)}\varepsilon_{\mcj,j} [\psi_{\mcj+j},\psi_{\mcj-j}] = -\frac{1}{2}\Pi_{\mcj}\sum_{j=1}^{\alpha(G)}\left(T_{G}(-u_{\mcj,j})\frac{\partial T_{G}(u_{\mcj,j})}{\partial u_{\mcj,j}}-T_{G}(u_{\mcj,j})\frac{\partial T_{G}(-u_{\mcj,j})}{\partial u_{\mcj,j}}\right)\prod_{\substack{k=1 \\ k \neq {\mcj,j}}}^{\alpha(G)}\left(\frac{-u_{\mcj,k}^2}{u_{\mcj,j}^2-u_{\mcj,k}^2}\right).
    \end{equation}
    Finally, by the Lagrange interpolation formula,
    \begin{equation}
        \sum_{\mcj}\left(\sum_{j=1}^{\alpha(G)}\varepsilon_{\mcj,j}[\psi_{\mcj,+j},\psi_{\mcj,-j}]\right)\Pi_{\mcj} = \left(T_{G}(u)\frac{\partial T_{G}(u)}{\partial u}\right)\Bigr|_{u=0} = H,
    \end{equation}
    completing the proof.
\end{proof}

\twocolumngrid

\bibliography{bibliography}

\end{document}